\documentclass[12pt]{article}
\usepackage{tikz}
\usetikzlibrary{shapes}

\usepackage[utf8]{inputenc}
\usepackage{amsmath,amsfonts,amssymb,amsthm,mathrsfs,graphicx,
	setspace,wrapfig,lscape,longtable,multirow,chngcntr,
	soul,bm,booktabs,changepage,float,enumitem,
	titlesec,quoting,xparse,pifont,apptools,makecell,array}
\usepackage[T1]{fontenc}
\usepackage[table, dvipsnames]{xcolor}
\usepackage[labelfont=bf]{caption,subcaption}
\usepackage[bottom]{footmisc}
\usepackage[page,title]{appendix}
\usepackage{natbib}
\usepackage[bookmarks=false]{hyperref}
\hypersetup{colorlinks=true,linkcolor=blue,urlcolor=blue,citecolor=blue}
\makeatletter
\AtBeginDocument{\typeout{get arXiv to do 4 passes: \csname theHpage\endcsname}}
\makeatother
\usepackage{xurl}
\usepackage{tikz-3dplot}\usetikzlibrary{decorations.pathreplacing,arrows.meta}
\usepackage[top=1.25in, bottom=1.25in, left=1.25in, right=1.25in]{geometry}
\usepgflibrary{arrows}
\usepackage{ragged2e}
\usepackage{alphalph}
\usepackage{placeins}
\usepackage{etoc}

\renewcommand{\today}{April 2026}

\AtAppendix{\counterwithin{theorem}{section}}
\newtheorem{theorem}{Theorem}

\AtAppendix{\counterwithin{corollary}{section}}
\newtheorem{corollary}{Corollary}
\AtAppendix{\counterwithin{definition}{section}}
\newtheorem{definition}{Definition}
\AtAppendix{\counterwithin{lemma}{section}}
\newtheorem{lemma}{Lemma}
\AtAppendix{\counterwithin{assumption}{section}}
\newtheorem{assumption}{Assumption}
\AtAppendix{\counterwithin{proposition}{section}}
\newtheorem{proposition}{Proposition}
\AtAppendix{\counterwithin{remark}{section}}
\newtheorem{remark}{Remark}

\newcommand{\cov}{\mathrm{cov}}

\DeclareMathOperator{\Var}{Var}
\DeclareMathOperator{\Cov}{Cov}
\newcommand{\Phinv}{\Phi^{-1}}
\newcommand{\E}{\mathbb{E}}

\newcolumntype{H}{>{\setbox0=\hbox\bgroup}c<{\egroup}@{}}
\newcolumntype{C}{>{$}c<{$}}
\newcolumntype{R}{>{$}r<{$}}
\newcolumntype{L}{>{$}l<{$}}

\title{
	{\textbf{The Corporate Bond Factor Replication Crisis}}\thanks{{\scriptsize 
	An earlier version of this paper circulated under the title ``Common pitfalls in the evaluation of corporate bond strategies.'' We gratefully acknowledge comments and suggestions from Patrick Augustin, Hendrik Bessembinder, Mike Chernov, Nicholas Chen, Mathieu Fournier, Alexey Ivashchenko, Soohun Kim, Jong-Myun Moon (Insight Investment), Yoshio Nozawa, Piotr Or{\l}owski, Yancheng Qiu, Lukas Schmid, Botao Wu, and seminar and conference participants at Aarhus University, Louvain Finance, UNSW Business School, USI Lugano, University of Melbourne, University of Sydney,  University of Technology Sydney, University of Toronto, Warwick Business School, and the 2023 ESSEC-CYU-Warwick Econometrics Workshop. 
    The companion website to this paper, \href{https://openbondassetpricing.com/}{\nolinkurl{openbondassetpricing.com}}, provides error-corrected daily and monthly TRACE data, bias-corrected factors, and replication code. The data are also available on \href{https://wrds-www.wharton.upenn.edu/pages/get-data/contributed-data-forms/corporate-bond-data-dickerson-daily/}{WRDS Contributed Data (daily)} and \href{https://wrds-www.wharton.upenn.edu/pages/get-data/contributed-data-forms/corporate-bond-data-dickerson-monthly/}{WRDS Contributed Data (monthly)}.
    The companion factor construction software, \href{https://pypi.org/project/PyBondLab/}{\texttt{PyBondLab}}, is available at \href{https://pypi.org/project/PyBondLab/}{\nolinkurl{pypi.org/project/PyBondLab/}}.
}
}  
}

\author{
{\normalsize Alexander Dickerson\thanks{{\scriptsize
School of Banking \& Finance, The University of New South Wales; \url{alexander.dickerson1@unsw.edu.au}}}}
\and
{\normalsize Cesare Robotti\thanks{{\scriptsize
Warwick Business School, The University of Warwick;
\url{cesare.robotti@wbs.ac.uk}; \url{giulio.rossetti.1@wbs.ac.uk}}}}
\and
{\normalsize Giulio Rossetti\footnotemark[\value{footnote}]}
}

\date{{\small   \today }}

\begin{document}
\maketitle

\vspace{-.75cm}

\begin{abstract}
\singlespacing
\noindent
{Corporate bond factor research faces a replication crisis. The crisis stems from two sources that inflate reported factor premia: transaction prices whose measurement error enters both sorting signals and return denominators, creating a correlated errors-in-variables bias, and asymmetric ex-post return filtering that embeds future information into factor construction. Applying our framework to a `factor zoo' of 108 signals across nine thematic clusters, we show that the majority of previously documented factors do not produce statistically significant bond CAPM alphas after correction. We provide an open source framework via \href{https://openbondassetpricing.com/}{Open Bond Asset Pricing}, including error-corrected TRACE data, bias-corrected factors, and software for reproducible research.
}
\end{abstract}
\bigskip
\singlespacing{
}
\medskip

\bigskip
\thispagestyle{empty}

\pagebreak
\onehalfspacing
\setstretch{1.4}
\setlength{\belowdisplayskip}{5pt} 
\setlength{\belowdisplayshortskip}{5pt}
\setlength{\abovedisplayskip}{5pt} 
\setlength{\abovedisplayshortskip}{5pt}


\clearpage

\begin{flushright}
{\footnotesize
\textit{``\ldots if you torture the data enough, nature will always confess''}\\
\hfill \citep{coase1982choose}
}
\end{flushright}

\section{Introduction}\label{sec: intro}

Corporate bond factor research faces a replication crisis. We construct a `factor zoo' of 108 corporate bond factors and show that, after correcting for measurement error and look-ahead bias, most are spanned by the bond market factor. A small subset of factors retains significant bond CAPM alphas, primarily those formed on credit-spread-based value signals. The apparent risk-adjusted performance of many published factors is an artifact of two biases whose magnitudes have not been systematically quantified. The first, Latent Implementation Bias (LIB), arises because the same noisy transaction price enters both sorting signals and return denominators, creating a correlated errors-in-variables (CEIV) problem, and because the signal price is a historical transaction, not an actionable quote. The second, Look-Ahead Bias (LAB), arises from asymmetric ex-post return filtering that embeds future information into factor construction. Both systematically inflate measured factor premia and alphas. A third problem compounds the first two because corporate bonds lack a standardized data source, and researchers make idiosyncratic filtering and portfolio-construction choices whose collective variation rivals sampling uncertainty.

Among factors formed on yield, credit spread, value, and reversal signals, only two retain statistically significant bond CAPM alphas after adjustment for measurement error. Both are credit-spread-based value and spread-change factors. Credit-spread-based value factors retain roughly half their premium after adjustment and remain statistically significant. Short-term reversal illustrates the severity of LIB, with the premium dropping from $-0.99$\% to $-0.09$\% per month after correction ($t$-statistic: $-4.46$ to $-0.51$), and bias comprising over 90\% of the documented effect. A second bias independently inflates reported premia. Many factors replicate only when researchers apply ex-post return filters, such as winsorizing or trimming at thresholds computed using 
full sample data, after the portfolio formation month. For the six-month momentum factor, the 0.30\% monthly premium is entirely attributable to asymmetric ex-post winsorization. Without winsorization, the premium is zero. For twelve-month momentum, the base premium is negative ($-0.13\%$) and ex-post filtering converts a losing factor into an apparent winner, with the bias concentrated during periods of financial distress. Idiosyncratic volatility factors and downside risk measures show similar patterns, with 57--78\% of the measured alpha attributable to asymmetric ex-post return filtering.

Across 432 factor-specification combinations (108 signals $\times$ 2 weighting schemes $\times$ 2 sorting methods), only 26 (6.0\%) bond CAPM alphas survive a Benjamini-Hochberg (BH) false discovery rate (FDR) correction, concentrated among credit-spread-based value factors. Results are consistent when FDR is applied to mean return $p$-values. Before correction, 119 of 432 specifications produce nominally significant premia ($t(\mu) > 1.96$), but only 22 survive a FDR correction.\footnote{We report the BH adjustment as an illustrative benchmark to show that the already limited number of statistically significant specifications declines further once one accounts for multiple testing. More conservative procedures that allow for general dependence across strategies would only strengthen this conclusion.} Our point is not that the bond CAPM is the correct model, but rather that, once LIB and LAB are taken into account, most factors no longer earn economically meaningful premia or deliver superior risk-adjusted performance relative to the bond market factor.

Beyond LIB and LAB, a third source of variation affects measured premia. The same market characteristics that generate noisy prices and extreme returns (infrequent trading, wide bid-ask spreads, and no consolidated data source) also fragment the data infrastructure, since no two research teams process Trade Reporting and Compliance Engine (TRACE) data identically, and different choices, not all equally defensible, generate large dispersion in estimated premia. We distinguish between data uncertainty, which arises from how the admissible investment universe is defined, and methodological uncertainty, which arises from how a given signal is constructed into a factor. Across 648 data-filtering configurations (69,984 factor paths across all 108 signals), the interquartile range of estimated premia (the non-standard error, NSE) averages 0.35\% per month, exceeding the average premium of 0.33\% per month, with an NSE/SE ratio of 1.15. Holding the data fixed and varying only portfolio construction across 168 economically distinct specifications per signal (18,128 factor paths) produces comparable magnitudes, with the NSE/SE ratio rising to 1.45 and exceeding unity for all nine factor clusters. Researcher degrees of freedom rival sampling uncertainty as a source of variation in measured factor performance, yet fewer than 6\% of specifications improve bond market risk-adjusted performance.
Fig.~\ref{fig:research_framework} maps these three problems to their causes, the biases they produce, and our proposed solutions. The first bias, LIB, arises from two sources: the same noisy transaction price enters both the sorting signal and the return denominator, creating a CEIV problem \citep*{blume1983biases,stambaugh1988information,duarte2024very}, and the observed transaction price is not executable in over-the-counter (OTC) markets. Both are severe in corporate bonds, where bid-ask spreads are an order of magnitude larger than in equities and trading is sporadic.\footnote{\citet*{dickerson2024factor} documents that in a representative year, 70\% of bonds trade on 10 business days or fewer, and less than 0.50\% of bonds trade every day.} The second bias, LAB, arises from ex-post return filtering, where winsorization thresholds computed from the full sample embed future information into factor construction, mechanically inflating reported premia. Section~\ref{sec:mmn} develops the formal framework for LIB, and Section~\ref{sec:lab} for LAB.
We propose a protocol for credible corporate bond factor research. To address LIB, we develop gap procedures that ensure the signal and the return denominator use different prices, and that measure returns from the earliest feasible execution price. Breaking this link removes the CEIV component by construction while preserving genuine premia, so that the difference between standard and adjusted approaches identifies the bias without requiring observation of the true price. To address LAB, we implement ex-ante filtering, where winsorization thresholds use only historical information available at portfolio formation. The third element quantifies data and methodological uncertainty, documenting how standard filtering and portfolio-construction choices affect magnitude and inference even after bias corrections.
All data and code are open source. The full replication code is available on \href{https://github.com/Alexander-M-Dickerson/trace-data-pipeline}{GitHub}. The companion software \href{https://pypi.org/project/PyBondLab/}{\texttt{PyBondLab}} enables reproducible factor construction from any signal. The `factor zoo', with bond-month data for 108 signals and pre-formed factors, is publicly available at \href{https://openbondassetpricing.com/}{Open Bond Asset Pricing} and on \href{https://wrds-www.wharton.upenn.edu/pages/get-data/contributed-data-forms/corporate-bond-data-dickerson-daily/}{WRDS Contributed Data (daily)} and \href{https://wrds-www.wharton.upenn.edu/pages/get-data/contributed-data-forms/corporate-bond-data-dickerson-monthly/}{WRDS Contributed Data (monthly)}.
Our work relates to several strands of literature. \citet*{harvey2016and} and \citet*{hou2020replicating} document concerns about the credibility of equity factor research, though \citet*{chen2021open} and \citet*{jensen2023there} find high reproduction rates when methodologies are applied consistently. In equity options, \citet*{duarte2023too} document replication failures from ex-post sample filters that generate infeasible factors, inflating reported Sharpe ratios by an order of magnitude (over 1000\%).
We build on earlier work documenting pricing challenges in corporate bonds. \citet*{dickerson2023priced} show that most proposed risk factors, with liquidity as a marginal exception, have no incremental pricing power beyond the bond market factor. \citet*{feldhutter2023} attribute replication failures to data errors in TRACE. We find that measurement error and look-ahead bias are the primary drivers. \citet*{ghaderi2024pricing} show that extending the cross-section to nearly a century reveals pricing power for several factors, emphasizing the value of longer samples. Our analysis differs from prior work in two respects. First, we provide a unifying framework for the corporate bond replication crisis, formalizing and quantifying LIB, LAB, and non-standard errors across 108 corporate bond factors. Second, we offer practical solutions through error-corrected open source data and software that enable researchers to construct bias-free factors and reproducible research going forward.

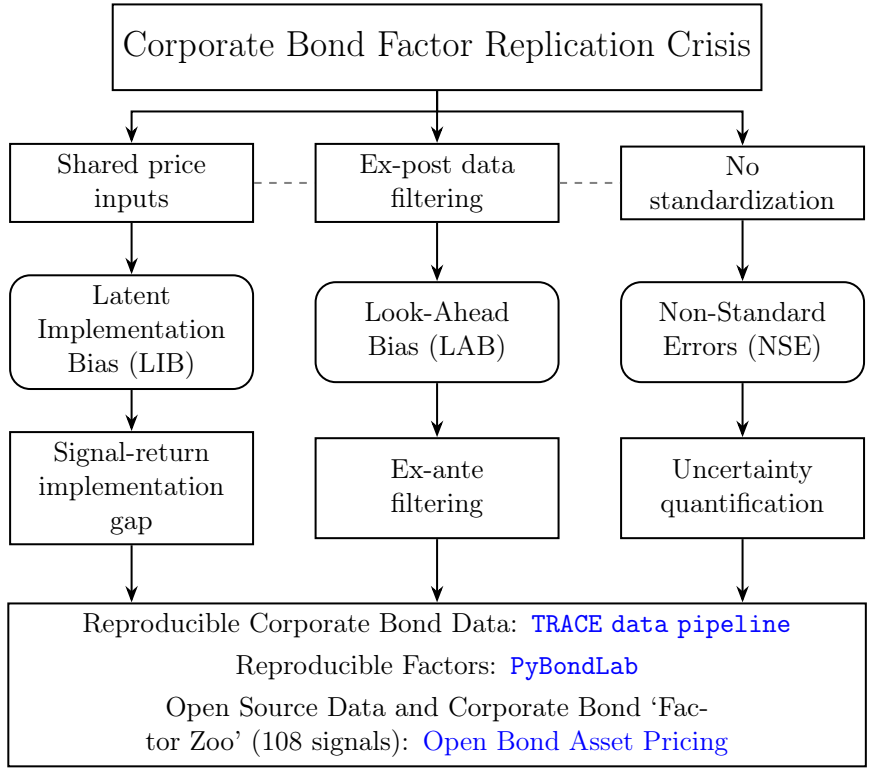
\begin{figure}[tbh!]
\begin{center}
\vspace{.25cm}
\scalebox{0.95}{%
\begin{tikzpicture}[
    >=Stealth,
    scale=0.95,
    titlebox/.style={draw, thick, minimum width=8cm, minimum height=1.2cm,
                     align=center, font=\large},
    causebox/.style={draw, thick, minimum width=3.4cm, minimum height=1.0cm,
                     text width=3.0cm, align=center, font=\small},
    biasbox/.style={draw, thick, rounded corners=8pt, minimum width=3.4cm, minimum height=1.4cm,
                     text width=3.0cm, align=center, font=\small},
    solutionbox/.style={draw, thick, minimum width=3.4cm, minimum height=1.4cm,
                        text width=3.0cm, align=center, font=\small},
    outputbox/.style={draw, thick, minimum width=12cm, minimum height=2.0cm,
                      text width=11.5cm, align=center, font=\small},
    arr/.style={->, thick}
]

\node[titlebox] (title) at (0, 8) {Corporate Bond Factor Replication Crisis};

\node[causebox] (cause_lib) at (-4.5, 6) {Shared price\\inputs};
\node[causebox] (cause_lab) at (0, 6) {Ex-post data\\filtering};
\node[causebox] (cause_nse) at (4.5, 6) {No\\standardization};

\draw[dashed, thick, gray] (cause_lib.east) -- (cause_lab.west);
\draw[dashed, thick, gray] (cause_lab.east) -- (cause_nse.west);

\draw[arr] (title.south) -- ++(0, -0.3) -| (cause_lib.north);
\draw[arr] (title.south) -- ++(0, -0.3) -| (cause_lab.north);
\draw[arr] (title.south) -- ++(0, -0.3) -| (cause_nse.north);

\node[biasbox] (lib) at (-4.5, 3.8) {Latent\\Implementation\\Bias (LIB)};
\node[biasbox] (lab) at (0, 3.8) {Look-Ahead\\Bias (LAB)};
\node[biasbox] (nse) at (4.5, 3.8) {Non-Standard\\Errors (NSE)};

\draw[arr] (cause_lib.south) -- (lib.north);
\draw[arr] (cause_lab.south) -- (lab.north);
\draw[arr] (cause_nse.south) -- (nse.north);

\node[solutionbox] (sol_lib) at (-4.5, 1.5) {Signal-return\\implementation gap};
\node[solutionbox] (sol_lab) at (0, 1.5) {Ex-ante\\filtering};
\node[solutionbox] (sol_nse) at (4.5, 1.5) {Uncertainty\\quantification};

\draw[arr] (lib.south) -- (sol_lib.north);
\draw[arr] (lab.south) -- (sol_lab.north);
\draw[arr] (nse.south) -- (sol_nse.north);

\node[outputbox] (output) at (0, -1.4) {Reproducible Corporate Bond Data: \href{https://github.com/Alexander-M-Dickerson/trace-data-pipeline}{\texttt{TRACE data pipeline}}\\[3pt]Reproducible Factors: \href{https://pypi.org/project/PyBondLab/}{\texttt{PyBondLab}}\\[3pt]Open Source Data and Corporate Bond `Factor Zoo' (108 signals): \href{https://openbondassetpricing.com/}{Open Bond Asset Pricing}};

\draw[arr] (sol_lib.south) -- (sol_lib.south |- output.north);
\draw[arr] (sol_lab.south) -- (sol_lab.south |- output.north);
\draw[arr] (sol_nse.south) -- (sol_nse.south |- output.north);

\end{tikzpicture}
} 
\end{center}
\caption{Framework for credible corporate bond factor research.}
\begin{justify}
\begin{spacing}{1}
\footnotesize{The figure maps the three sources of the corporate bond replication crisis to their causes, biases, and solutions. Shared price inputs and non-executable transaction prices give rise to Latent Implementation Bias (LIB). Ex-post data filtering introduces Look-Ahead Bias (LAB). The absence of standardized data and methods generates Non-Standard Errors (NSE) through researcher degrees of freedom. The dashed line indicates that infrequent trading and wide bid-ask spreads amplify all three problems. The solutions are a signal-return implementation gap for LIB, ex-ante filtering for LAB, and uncertainty quantification for NSE. The resulting open source protocol (data, software, and a corporate bond `factor zoo' of 108 signals) is available at \href{https://openbondassetpricing.com/}{Open Bond Asset Pricing}.
}
\end{spacing}
\end{justify}
\label{fig:research_framework}
\end{figure}

\section{Open source corporate bond asset pricing data}\label{sec:osbap}

We provide an open source end-to-end pipeline that transforms raw TRACE intraday transactions into a monthly corporate bond asset pricing data set. The pipeline first cleans transactions, then computes daily prices, yields, and credit spreads, and finally aggregates to monthly frequency. It differs from existing approaches in four ways. First, we include Rule~144A corporate bonds, which are excluded from most existing studies. Second, we track bonds through and after default rather than dropping them at the event boundary, avoiding censoring of the return distribution \citep*{baumann2025life,baumann2025defaulted}. Third, every filter decision is documented through companion data reports that are generated automatically when the pipeline is executed, plotting the complete price time series for each affected CUSIP so that researchers can inspect every correction and exclusion.\footnote{Detailed daily data error reports for TRACE enhanced are available on \href{https://openbondassetpricing.com/wp-content/uploads/2026/03/enhanced_data_report_osbap_0k_volume_2025.pdf}{Open Bond Asset Pricing}.} Fourth, all filter thresholds are configurable, and researchers can modify any parameter and re-run the pipeline from raw transactions up to the monthly data output. The sample covers August 2002 to December 2024.

\subsection{From raw transactions to daily prices}\label{sec:data-daily}

The pipeline begins with raw TRACE enhanced and 144A transaction records, merged with Mergent Fixed Income Securities Database (FISD) for bond characteristics. We retain USD-denominated, fixed-rate, non-convertible, non-asset-backed bonds with \$1,000 par value and original maturity of at least one year. Standard filters from \citet*{DickNielsen2014HowTC} remove cancellations, corrections, reversals, and agency-side duplicates.
Beyond these standard filters, we develop two transaction-level corrections for recording errors that survive the Dick-Nielsen procedures. The \emph{decimal shift corrector} identifies prices recorded with incorrect decimal placement (for example, 10.5 entered as 105.0) and applies a multiplicative correction when a set of acceptance conditions is satisfied. The \emph{bounce-back filter} detects transient price spikes that deviate from a backward-looking anchor and revert within a short window. Both filters are designed to correct errors where possible and flag anomalies otherwise, minimizing data loss. Appendix~\ref{app:filters} provides complete filter specifications and parameter values.

After filtering, transactions are aggregated to daily frequency using volume-weighted average prices. We then use \href{https://quantlib-python-docs.readthedocs.io/en/latest/instruments/bonds.html}{\texttt{QuantLib}} to compute duration, convexity, yield to maturity, and credit spreads, and merge credit ratings, equity identifiers, and Fama-French industry classifications. A final set of four daily-level filters targets anomalous prices in distressed bonds, addressing isolated ultra-low prices, transient spikes, stale plateaus, and intraday inconsistencies. Appendix~\ref{app:filters} documents each filter. Detailed data reports with descriptive statistics, industry breakdowns, and additional information are available for \href{https://openbondassetpricing.com/wp-content/uploads/2026/03/stage1_data_report_osbap_0k_volume_2025.pdf}{download}. The Internet Appendix reports data availability and descriptive statistics for the daily data.
The code automatically generates companion data reports that plot the price time series of every bond CUSIP affected by a filter, providing a transparent record of each filtering decision.

\subsection{Monthly returns and default handling}\label{sec:data-monthly}

For each bond-month, we compute month-end returns (for comparability with prior literature), month-begin returns (computed from the first (last) available price within the first five (last five) business days of month $t+1$, measuring implementable performance after signal observation),\footnote{The median gap between the month-end price and the first available price the following month is one business day.} and duration-adjusted excess returns that remove interest rate exposure. The monthly price is the last available daily price within the last five business days of the month (New York Stock Exchange (NYSE) calendar). A valid return requires such a price in both months $t$ and $t+1$ (bonds without a trade in either window are excluded for that month). Excess returns subtract the Fama-French one-month T-bill rate. The bond market factor (\texttt{MKTB}) is the value-weighted excess return on all corporate bonds in the sample. We use a single-factor bond CAPM (CAPMB) comprising \texttt{MKTB} to compute alphas throughout the paper.\footnote{In unreported results, we use multi-factor models that include default, credit, liquidity, and equity factors and observe lower alphas on average.} Appendix~\ref{app:returns} provides the formal definitions, and Section~\ref{sec:mmn} discusses the distinction between month-end and month-begin measurement.

Rule~144A bonds enter the sample from June 2014 onward and are processed identically to TRACE enhanced. These securities, issued under SEC Rule~144A to qualified institutional buyers, account for over 20\% of new corporate bond issuance as of 2025. Excluding them omits a large and growing segment of the investable universe.\footnote{Rule~144A transaction-level data became available on Wharton Research Data Services (WRDS) from June 2014 onward. The June 2014 entry point is therefore a data constraint, not a sample design choice. In unreported robustness checks, excluding 144A bonds from the full sample does not materially change our main results.}

Dropping defaulted bonds from the sample censors extreme outcomes and can bias measured return distributions. We track bonds through and after default. In the default month, we compute the standard return, a default-event return excluding the final coupon, and a trading-in-default return on a flat basis without accrued interest. Appendix~\ref{app:defaults} gives the complete return formulas for each case. The Internet Appendix reports monthly data coverage and descriptive statistics, and documents the prevalence and time concentration of extreme returns.

For price-based signals, we use two complementary gap procedures to break the mechanical link between signal measurement error and the return denominator. The first gap procedure observes the sorting variable at least one business day before the month-end price used for return computation, while the second retains the standard month-end signal but computes returns from prices observed in the first (last) five business days of the following month. Section~\ref{sec:mmn} provides the formal treatment.

\subsection{The corporate bond `factor zoo'}\label{sec:data-zoo}

We construct 108 signals spanning nine clusters (Spreads, Yields, Size; Value; Momentum \& Reversal; Illiquidity; Volatility \& Risk; Market Risk; Credit \& Default Betas; Volatility \& Liquidity Betas; and Macro \& Other Betas). The Internet Appendix provides definitions and original citations and documents construction details.

The sample spans August 2002 to December 2024 (269 months), with 52,656 unique bonds and an average of 6,790 bond--month observations per month. Variable names are consistent across the database and code, and match the Internet Appendix. For example, \texttt{cs} denotes credit spread, \texttt{str} denotes short-term reversal, and \texttt{mom6\_1} denotes six-month momentum. To our knowledge, this is the first open source corporate bond `factor zoo' containing the majority of signals proposed in the literature.

The data pipeline addresses one dimension of the replication crisis, the absence of a standardized, transparent data source for corporate bond research. The remaining sections investigate three sources of bias that affect factor premia independently of the data source. Section~\ref{sec:mmn} examines measurement error in price-based signals. Section~\ref{sec:lab} addresses look-ahead bias from ex-post return filtering. Section~\ref{sec:uncertainty} quantifies the sensitivity of estimated premia to methodological variation.

\section{Measurement error and latent implementation bias}\label{sec:mmn}

The month-end bid--ask averaged TRACE price, which underlies most corporate bond factor research, is a synthetic construct. It is neither a precise measure of value nor an executable quote.\footnote{Quote-based data from sources such as Bank of America Merrill Lynch or Bloomberg Barclays are not necessarily a panacea: dealer quotes can be stale. However, existing evidence suggests that these quotes are broadly reliable; see \citet*{choi2013drives} and \citet*{andreani2023computing}.}

Let $P_{i,t}$ and $s_{i,t}$ denote the true price and true sorting signal, and let $\hat{P}_{i,t}$ and $\hat{s}_{i,t}$ denote their observed counterparts. Following \citet*{blume1983biases}, observed prices satisfy $\hat{P}_{i,t} = (1+\delta_{i,t})P_{i,t}$, where $\delta_{i,t}$ is mean-zero price noise, and observed signals satisfy $\hat{s}_{i,t} = s_{i,t} + \eta_{i,t}$, where $\eta_{i,t}$ denotes signal measurement error. In TRACE, the month-end price is proxied by the volume-weighted average of transactions on the last trading day within the final five business days of the month, which averages across customer buy and sell transactions.\footnote{The measurement error $\delta$ is a reduced-form representation encompassing bid-ask bounce, dealer inventory effects, informed trading, and stale pricing. Our framework does not require distinguishing among these sources; any discrepancy that enters both the sorting signal and the return denominator generates CEIV bias.} When the same price disturbance $\delta_{i,t}$ enters both $\eta_{i,t}$ and the return measurement error $\epsilon_{i,t} \approx \delta_{i,t+1} - \delta_{i,t}$ (Definition~\ref{def:return_decomp} in Appendix~\ref{app:sorting-bias-formal}), the resulting correlation between signal and return measurement errors does not cancel in long-short factors. The component $\delta_{i,t+1}$ is realized after portfolio formation and, under serial independence, has zero conditional expectation given portfolio assignment. Thus, the bias-relevant component of the return measurement error is approximately $-\delta_{i,t}$: the price error at time $t$ mechanically distorts the return denominator, understating the return when $\delta_{i,t}$ is positive and overstating it when negative. This error is correlated with the signal error because both load on the same price disturbance.

For all price-based signals in our `factor zoo', higher price maps to a lower sorting signal (Proposition~\ref{prop:sign} in Appendix~\ref{app:sorting-bias-formal}): a bond with $\delta_{i,t} > 0$ is sorted toward the short leg and simultaneously has its return understated; a bond with $\delta_{i,t} < 0$ is sorted toward the long leg with its return overstated.\footnote{Consider two bonds with true price \$100, true return 0\%, and noise $\delta_A = +0.50\%$, $\delta_B = -0.50\%$. Bond~A is sorted toward the short leg and earns a measured return of $100/100.50 - 1 = -0.50\%$. Bond~B is sorted toward the long leg and earns $100/99.50 - 1 = +0.50\%$. The measured long-short return is $1.00\%$; the true long-short return is zero. Neither misranking nor return distortion alone produces this bias: with correct returns, both legs earn 0\% (misranking attenuates); without sorting, the return errors cancel across bonds (by $\E[\delta_{i,t}] = 0$). The bias arises because the same $\delta$ that determines portfolio assignment also determines the return error.} The long leg holds bonds with systematically positive return errors, and the short leg holds bonds with systematically negative return errors, spuriously widening the measured premium. Misranking alone does not create this bias: if the signal were noisy but returns measured without error, sorting on the noisy signal would attenuate the true premium toward zero (Appendix~\ref{app:sorting-bias-formal}, Remark~\ref{rem:caveat}). The bias arises because the same $\delta_{i,t}$ that misranks bonds also contaminates their measured returns; misranking and return distortion are mechanically correlated, and this correlated error structure creates a spurious additive premium.

This is the CEIV mechanism, named by \citet*{duarte2024very}, building on the measurement-error literature of \citet*{blume1983biases}, \citet*{fama1984information}, and \citet*{stambaugh1988information}. The CEIV mechanism builds on \citet*{blume1983biases}, who showed that bid-ask bounce biases equally-weighted portfolio returns, and \citet*{stambaugh1988information}, who addressed an analogous correlated measurement error in Treasury bill forward premium regressions by using prices from different maturities. Our gap procedure applies the same logic using temporal rather than cross-maturity independence. \citet*{duarte2024very} applied the mechanism to option pricing and showed that their adjustment reverses the sign of the estimated volatility risk premium in individual stock options from positive to negative, the theoretically correct sign. In equities, \citet*{jegadeesh1990evidence} and \citet*{conrad1997profitability} recognized the same contamination in short-term reversal profits and addressed it by excluding the shared boundary price. In corporate bonds, \citet*{lair2024valuations} show that month-end reversal profits vanish when implemented at month-begin, consistent with an implementation shortfall, though pricing distortions relative to independent index valuations retain predictive power even after controlling for bid-ask bounce.

CEIV bias is of first-order importance in the price noise, unlike the classic \citet*{blume1983biases} bias ($\sigma_\delta^2$), which is of second-order importance and negligible in long-short factors. It also differs from standard errors-in-variables (EIV) attenuation \citep*[e.g.,][]{FamaMacBeth_1973}, which biases regression coefficients toward zero. Appendix~\ref{app:sorting-bias-formal} derives the closed-form expression under distributional assumptions.

There is also an \emph{implementation problem}. The month-end TRACE price records past transactions, not standing offers. A transaction price observed in a backtest is not necessarily available to trade at, because corporate bonds trade OTC and trading is not continuous. Standard backtests that use the signal price as the return denominator measure returns that no investor could actually capture. This gap between the observed signal and the earliest feasible trade creates what we call LIB.\footnote{We use the term ``latent'' because the bias is invisible to researchers who follow the standard approach: it does not appear as an explicit adjustment or error term but is embedded in the correlation structure between signal noise and return noise. The bias becomes apparent only when the signal-return price overlap is broken.} LIB encompasses both the CEIV bias from measurement error and the non-executability of observed prices. Both widen the gap between reported and achievable returns.

\paragraph{Breaking the link.} 
Both problems share a common source: using the same month-end price for signal computation and return measurement.\footnote{\citet*{chen2024reaching} impose a one-month gap between signal observation and return computation, and \citet*{Bartram-Grinblatt-Nozawa-2021} vary the return denominator price to compute within-month returns; both recognize the nonsynchronous trading problem but neither provides a formal CEIV framework or quantifies the bias systematically across signals.} The measurement problem arises because the same noise $\delta_{i,t}$ enters both; the implementation problem arises because the signal price is not executable. A time gap between signal and return can partially address both: the price used to compute the sorting signal must come from a different date than the price in the return denominator. \citet*{duarte2024very} validate a conceptually analogous lag adjustment via simulation in their option-pricing context, confirming that the lag approach reliably removes the CEIV component when measurement errors are serially uncorrelated.\footnote{This approach assumes measurement errors are serially uncorrelated, i.e., $\cov(\delta_{i,t-\Delta}, \delta_{i,t}) = 0$ for $\Delta > 0$. If errors exhibit positive autocorrelation, the gap reduces but does not eliminate the mechanical correlation between portfolio weights $\omega_{i,t}$ (formed from price-based signals) and measured returns $r_{i,t+1}$. For illiquid bonds with stale or persistent pricing, some residual correlation may remain.}
A \emph{signal gap} computes the signal from an earlier price $\hat{P}_{i,t-\Delta}$, which contains noise $\delta_{i,t-\Delta}$. The return denominator still uses $\hat{P}_{i,t}$, containing $\delta_{i,t}$. Under serial independence (Assumption~\ref{assum:iid} in Appendix~\ref{app:sorting-bias-formal}), $\delta_{i,t-\Delta}$ and $\delta_{i,t}$ are independent, and the CEIV-induced mechanical correlation is eliminated. A \emph{return gap} measures the return from month-begin, where the investor observes the signal at month-end $t$ and buys into the position at the first available price in month $t{+}1$. The return denominator is now $\hat{P}_{i,t+1}^{\text{bgn}}$, which contains $\delta_{i,t+1}^{\text{bgn}}$, independent of the $\delta_{i,t}$ that drove sorting under Assumption~\ref{assum:iid}.
\paragraph{Quantifying the bias.} 
To proxy for the magnitude of LIB, we compute two types of monthly returns. The standard \emph{month-end return} measures performance from the end of month $t$ to the end of month $t+1$ (see Panel~A of Fig.~\ref{fig:return_timeline}):
\begin{equation}
r_{i,t+1}^{\text{End}} = \frac{P_{i,t+1}^{\text{end}} + AI_{i,t+1}^{\text{end}} + C_{i,t+1}}{P_{i,t}^{\text{end}} + AI_{i,t}^{\text{end}}} - 1,
\label{eq:ret_end}
\end{equation}
where $P^{\text{end}}$ is the clean price observed within the last 5 business days of the month, $AI$ is accrued interest, and $C$ is any coupon payment. The \emph{month-begin return} measures performance from the beginning to the end of month $t+1$ (see Panel~C of Fig.~\ref{fig:return_timeline}):
\begin{equation}
r_{i,t+1}^{\text{Bgn}} = \frac{P_{i,t+1}^{\text{end}} + AI_{i,t+1}^{\text{end}} + C_{i,t+1}}{P_{i,t+1}^{\text{bgn}} + AI_{i,t+1}^{\text{bgn}}} - 1,
\label{eq:ret_bgn}
\end{equation}
where $P^{\text{bgn}}$ is observed within the first 5 business days of the month. The month-begin return captures what a trader could actually earn after observing a signal at month-end: the earliest feasible execution occurs at month-begin, not at the signal observation price.
The month-end return approximately decomposes into the implementable return plus the latent implementation bias:
\begin{equation}
r^{\text{End}}_{i,t+1} \approx \underbrace{\frac{P_{i,t+1}^{\text{bgn}}}{P_{i,t}^{\text{end}}} - 1}_{\text{LIB}_{i,t+1}} + r^{\text{Bgn}}_{i,t+1}.
\label{eq:lib}
\end{equation}
The LIB term measures the price change between signal observation at month-end $t$ and the earliest feasible trade at month-begin $t{+}1$. Because this price change occurs before any investor could act on the signal, the LIB component of the month-end return is not achievable in practice. It captures both the statistical bias from shared price inputs and the economic wedge between what a backtest reports and what a trader could earn. \citet*{lair2024valuations} propose a similar decomposition for the short-term reversal factor, which they term the implementation shortfall. Panel~D of Fig.~\ref{fig:return_timeline} illustrates this decomposition.

We compare factor performance under three approaches. \emph{Approach~1 (Unadjusted)} uses the month-end price $P_{i,t}^{\text{end}}$ for both signal computation and return measurement, which is the standard practice in the literature.  Portfolio weights $\omega_{i,t}$ are computed from signals observed at $P_{i,t}^{\text{end}}$, and returns are the month-end returns $r_{i,t+1}^{\text{End}}$. \emph{Approach~2 (Adjusted Signal)} breaks the correlation by computing signals from prices observed at least 1 business day before month-end (up to 10 business days).\footnote{The median signal gap in the data is 1 business day. Capping the maximum gap at a lower number makes an immaterial difference to the results.} The signal noise now depends on $\delta_{i,t-\Delta}$, which is independent of $\delta_{i,t}$ in the return denominator. Portfolio weights $\omega_{i,t}^{a}$ are computed from these gapped signals; returns remain month-end returns. Panel~B of Fig.~\ref{fig:return_timeline} illustrates the signal gap.
\emph{Approach~3 (Adjusted Return)} breaks the correlation by measuring the return from month-begin: weights $\omega_{i,t}$ are computed from month-end signals, but returns are $r_{i,t+1}^{\text{Bgn}}$, which uses $P_{i,t+1}^{\text{bgn}}$ in the denominator. The return noise now depends on $\delta_{i,t+1}^{\text{bgn}}$, independent of the $\delta_{i,t}$ that drove sorting. This approach also measures the implementable return, that is, the return a trader could actually capture after observing the signal.

Approaches 2 and 3 both break the shared-price link between signal and return, whereas Approach~1 retains it. The difference between the standard and adjusted approaches identifies the bias empirically, without requiring observation of the true price. We quantify the bias as
\begin{align}
\text{Bias}_{1-2} &= (\omega_t - \omega_t^{a}) \times r_{t+1}^{\text{End}}, \label{eq:bias_12} \\[4pt]
\text{Bias}_{1-3} &= \omega_t \times (r_{t+1}^{\text{End}} - r_{t+1}^{\text{Bgn}}). \label{eq:bias_13}
\end{align}
Bias$_{1-2}$ isolates the effect of using standard versus gapped signals while holding returns constant. Bias$_{1-3}$ isolates the effect of using infeasible versus implementable returns while holding weights constant. Positive bias indicates that the standard approach overstates factor performance.

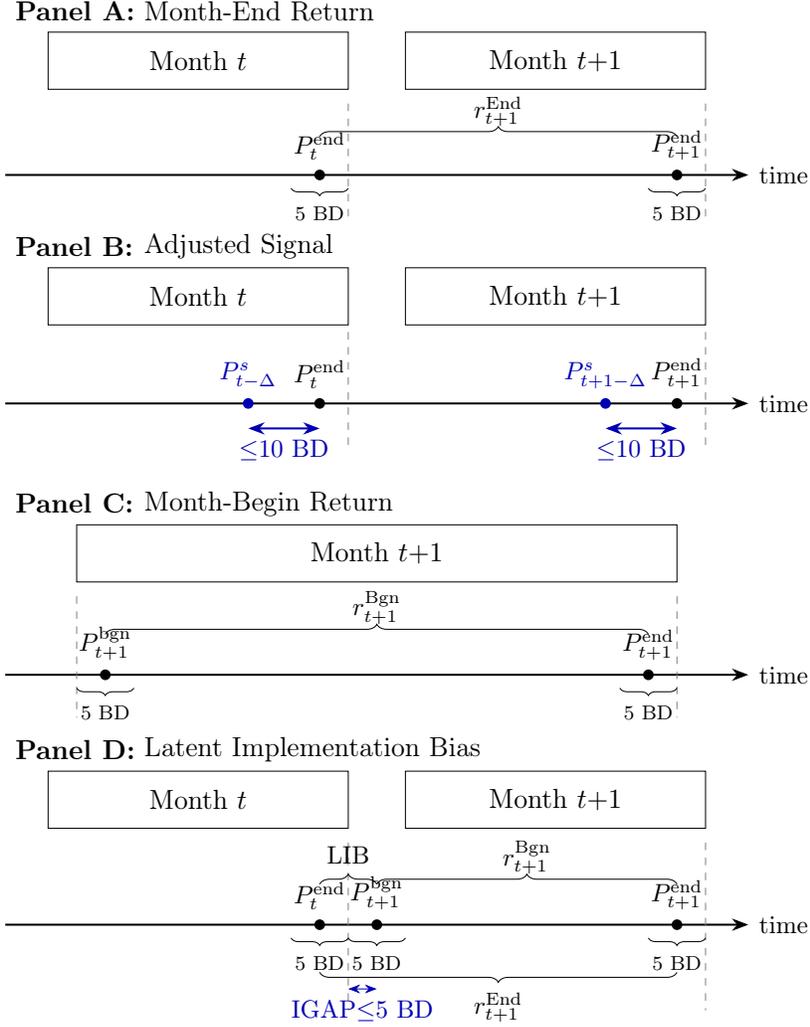
\begin{figure}[tbh!]
\begin{center}
\vspace{.25cm}
\scalebox{0.95}{%
\begin{tikzpicture}[
    >=Stealth,
    month/.style={draw, minimum width=4.2cm, minimum height=0.8cm, font=\small},
    timepoint/.style={circle, fill=black, inner sep=1.5pt},
    signalpoint/.style={circle, fill=blue!70!black, inner sep=1.5pt},
    brace/.style={decorate, decoration={brace, amplitude=4pt, raise=1pt}},
    bracebelow/.style={decorate, decoration={brace, amplitude=4pt, raise=1pt, mirror}},
    smallbrace/.style={decorate, decoration={brace, amplitude=3pt, raise=1pt, mirror}},
    lbl/.style={font=\footnotesize},
    panel/.style={font=\bfseries\small},
    paneldesc/.style={font=\small}
]

\node[panel, anchor=west] at (-5.2, 4.7) {Panel A:};
\node[paneldesc, anchor=west] at (-3.4, 4.7) {Month-End Return};
\node[month] (m1a) at (-2.5, 4) {Month $t$};
\node[month] (m2a) at (2.5, 4) {Month $t{+}1$};
\draw[thick, ->] (-5.2, 2.4) -- (5.2, 2.4) node[right, font=\footnotesize] {time};
\draw[gray, dashed] (-0.4, 3.4) -- (-0.4, 1.8);
\draw[gray, dashed] (4.6, 3.4) -- (4.6, 1.8);
\node[timepoint] (p1a) at (-0.8, 2.4) {};
\node[timepoint] (p2a) at (4.2, 2.4) {};
\node[lbl, above=2pt] at (p1a) {$P_{t}^{\text{end}}$};
\node[lbl, above=2pt] at (p2a) {$P_{t+1}^{\text{end}}$};
\draw[smallbrace] (-1.2, 2.25) -- node[below=4pt, font=\scriptsize] {5 BD} (-0.4, 2.25);
\draw[smallbrace] (3.8, 2.25) -- node[below=4pt, font=\scriptsize] {5 BD} (4.6, 2.25);
\draw[brace] ($(p1a.north)+(0,0.42)$) -- node[above=1pt, lbl] {$r_{t+1}^{\text{End}}$} ($(p2a.north)+(0,0.42)$);

\node[panel, anchor=west] at (-5.2, 1.4) {Panel B:};
\node[paneldesc, anchor=west] at (-3.4, 1.4) {Adjusted Signal};
\node[month] (m1b) at (-2.5, 0.7) {Month $t$};
\node[month] (m2b) at (2.5, 0.7) {Month $t{+}1$};
\draw[thick, ->] (-5.2, -0.8) -- (5.2, -0.8) node[right, font=\footnotesize] {time};
\draw[gray, dashed] (-0.4, 0.2) -- (-0.4, -1.4);
\draw[gray, dashed] (4.6, 0.2) -- (4.6, -1.4);
\node[timepoint] (p1b) at (-0.8, -0.8) {};
\node[timepoint] (p2b) at (4.2, -0.8) {};
\node[signalpoint] (s1b) at (-1.8, -0.8) {};
\node[signalpoint] (s2b) at (3.2, -0.8) {};
\node[lbl, above=2pt] at (p1b) {$P_{t}^{\text{end}}$};
\node[lbl, above=2pt] at (p2b) {$P_{t+1}^{\text{end}}$};
\node[lbl, above=2pt, blue!70!black] at (s1b) {$P^{s}_{t-\Delta}$};
\node[lbl, above=2pt, blue!70!black] at (s2b) {$P^{s}_{t+1-\Delta}$};
\draw[<->, thick, blue!70!black] (-1.8, -1.15) -- (-0.8, -1.15);
\node[lbl, blue!70!black] at (-1.3, -1.45) {$\leq$10 BD};
\draw[<->, thick, blue!70!black] (3.2, -1.15) -- (4.2, -1.15);
\node[lbl, blue!70!black] at (3.7, -1.45) {$\leq$10 BD};

\node[panel, anchor=west] at (-5.2, -2.2) {Panel C:};
\node[paneldesc, anchor=west] at (-3.4, -2.2) {Month-Begin Return};
\node[month, minimum width=8.4cm] (m1c) at (0, -2.9) {Month $t{+}1$};
\draw[thick, ->] (-5.2, -4.6) -- (5.2, -4.6) node[right, font=\footnotesize] {time};
\draw[gray, dashed] (-4.2, -3.5) -- (-4.2, -5.2);
\draw[gray, dashed] (4.2, -3.5) -- (4.2, -5.2);
\node[timepoint] (p1c) at (-3.8, -4.6) {};
\node[timepoint] (p2c) at (3.8, -4.6) {};
\node[lbl, above=2pt] at (p1c) {$P_{t+1}^{\text{bgn}}$};
\node[lbl, above=2pt] at (p2c) {$P_{t+1}^{\text{end}}$};
\draw[smallbrace] (-4.2, -4.75) -- node[below=4pt, font=\scriptsize] {5 BD} (-3.4, -4.75);
\draw[smallbrace] (3.4, -4.75) -- node[below=4pt, font=\scriptsize] {5 BD} (4.2, -4.75);
\draw[brace] ($(p1c.north)+(0,0.45)$) -- node[above=1pt, lbl] {$r_{t+1}^{\text{Bgn}}$} ($(p2c.north)+(0,0.45)$);

\node[panel, anchor=west] at (-5.2, -5.65) {Panel D:};
\node[paneldesc, anchor=west] at (-3.4, -5.65) {Latent Implementation Bias};
\node[month] (m1d) at (-2.5, -6.35) {Month $t$};
\node[month] (m2d) at (2.5, -6.35) {Month $t{+}1$};
\draw[thick, ->] (-5.2, -8.1) -- (5.2, -8.1) node[right, font=\footnotesize] {time};
\draw[gray, dashed] (-0.4, -6.95) -- (-0.4, -9.3);
\draw[gray, dashed] (4.6, -6.95) -- (4.6, -9.3);
\node[timepoint] (pe) at (-0.8, -8.1) {};
\node[timepoint] (pb) at (0.0, -8.1) {};
\node[timepoint] (pf) at (4.2, -8.1) {};
\node[lbl, above=2pt] at (pe) {$P_{t}^{\text{end}}$};
\node[lbl, above=2pt] at (pb) {$P_{t+1}^{\text{bgn}}$};
\node[lbl, above=2pt] at (pf) {$P_{t+1}^{\text{end}}$};
\draw[smallbrace] (-1.2, -8.25) -- node[below=4pt, font=\scriptsize] {5 BD} (-0.4, -8.25);
\draw[smallbrace] (-0.4, -8.25) -- node[below=4pt, font=\scriptsize] {5 BD} (0.4, -8.25);
\draw[smallbrace] (3.8, -8.25) -- node[below=4pt, font=\scriptsize] {5 BD} (4.6, -8.25);
\draw[brace] ($(pe.north)+(0,0.45)$) -- node[above=5pt, lbl] {LIB} ($(pb.north)+(0,0.45)$);
\draw[brace] ($(pb.north)+(0,0.45)$) -- node[above=1pt, lbl] {$r_{t+1}^{\text{Bgn}}$} ($(pf.north)+(0,0.45)$);
\draw[bracebelow] ($(pe.south)+(0,-0.55)$) -- node[below=5pt, lbl] {$r_{t+1}^{\text{End}}$} ($(pf.south)+(0,-0.55)$);
\draw[<->, thin, blue!70!black] (-0.4, -9.0) -- (0.0, -9.0);
\node[lbl, blue!70!black] at (-0.2, -9.3) {IGAP$\leq$5 BD};

\end{tikzpicture}
} 
\end{center}
\caption{Return measurement windows and bias decomposition.}
\begin{justify}
\begin{spacing}{1}
\footnotesize{
The figure illustrates the four return measurement windows used in the paper. Panel~A plots the standard month-end return, measured from the last 5 business days (BD) of month $t$ to the last 5 business days of month $t{+}1$. Panel~B depicts the adjusted signal timing, where the investor observes signals at $P^{s}_{t-\Delta}$ and $P^{s}_{t+1-\Delta}$, between 1 and 10 business days before the month-end prices. Panel~C plots the month-begin return, measured within month $t{+}1$ from the first 5 to the last 5 business days. Panel~D decomposes the month-end return into LIB (from $P_{t}^{\text{end}}$ to $P_{t+1}^{\text{bgn}}$) plus the month-begin return (from $P_{t+1}^{\text{bgn}}$ to $P_{t+1}^{\text{end}}$). The implementation gap (IGAP) is the minimum one-day gap between the month-end and month-begin prices. All business days follow the NYSE trading calendar.
}
\end{spacing}
\end{justify}
\label{fig:return_timeline}
\end{figure}

Table~\ref{tab:mmn_1} presents bias estimates for seven price-based factors constructed from commonly used signals: yield to maturity (\texttt{ytm}), credit spread (\texttt{cs}), bond book-to-market (\texttt{bbtm}), 6-month change in log spreads (\texttt{dcs6}), IPR value (\texttt{val\_ipr}), HZ value (\texttt{val\_hz}), and short-term reversal (\texttt{str}). In Panel~A, we sort bonds into deciles each month and form value-weighted portfolios (using bond market capitalization) that are long the top decile and short the bottom decile. In Panel~B, we construct within-firm factors following \citet*{feldhutter2023}: for each firm with at least two bonds, we form long-short portfolios based on within-firm signal rankings, then aggregate across firms using market-value weights.

Short-term reversal exhibits the largest bias. In Panel~A (single-sort), the unadjusted factor earns $-$0.99\% per month ($t = -4.46$), but the signal-adjusted factor earns only $-$0.09\% ($t = -0.51$). The implied bias is 90 basis points, or 91\% of the measured premium. The unadjusted CAPMB alpha of $-$0.77\% becomes 0.12\% after adjustment, with a bias of 89 basis points ($t = -8.20$). The return-adjusted approach yields comparable results: the premium falls to $-$0.17\% with a bias of 82 basis points.
The value factors show biases of 36--45 basis points. For \texttt{val\_ipr}, the unadjusted premium of 0.96\% per month falls to 0.51\% after signal adjustment. The bias of 45 basis points ($t = 6.93$) accounts for 47\% of the measured premium. For \texttt{val\_hz}, the premium falls from 0.81\% to 0.45\%, a bias of 36 basis points. 
Spread-based factors (\texttt{ytm}, \texttt{cs}, \texttt{bbtm}) display more modest biases of 15--18 basis points per month, though all are highly statistically distinguishable from zero ($t > 5$). The \texttt{dcs6} factor, which sorts on spread changes from \citet*{KPP2023}, has a negative premium ($-$1.02\% unadjusted, $-$0.50\% adjusted) and correspondingly negative bias ($-$52 basis points).

Panel~B examines within-firm factors, which control for issuer-level heterogeneity by forming long-short portfolios within each firm. \citet*{feldhutter2023} report that within-firm factors based on price-based signals, specifically \texttt{ytm}, \texttt{cs}, \texttt{bbtm}, \texttt{val\_ipr}, and \texttt{str}, exhibit the largest premia in their framework. After adjustment, only two factors retain statistically meaningful alphas: \texttt{val\_ipr} ($\alpha = 0.18\%$, $t = 4.25$ with adjusted signal; $\alpha = 0.20\%$, $t = 4.30$ with adjusted return) and \texttt{dcs6} ($\alpha = -0.22\%$, $t = -4.98$). The remaining factors (\texttt{ytm}, \texttt{cs}, \texttt{bbtm}, \texttt{val\_hz}, and \texttt{str}) all have adjusted alphas with $|t| < 2$. Within-firm premia and alphas are lower on average than their single-sort counterparts, and the gap widens after adjustment.
\begin{table}[!ht]
\caption{Latent implementation bias in price-based factors.}
\begin{spacing}{1}
{\footnotesize
The table reports premia and CAPMB alphas for seven price-based factors under standard and bias-adjusted approaches. Portfolios are value-weighted (using bond market capitalization) with excess returns over the one-month T-bill rate. Panel~A sorts bonds into deciles each month, going long P10 and short P1. Panel~B uses within-firm sorts: for each firm with at least two bonds, we rank bonds by signal, form a long-short portfolio, and aggregate across firms using market-value weights. $\mu$ is the monthly average return (\%), $\alpha$ is the CAPMB alpha (\%). $\Delta\mu$ and $\Delta\alpha$ measure the bias as the difference in premia and alphas between standard and adjusted approaches. Bias~(1)$-$(2) compares standard versus signal-adjusted. Bias~(1)$-$(3) compares standard versus return-adjusted. $t$-statistics (Newey-West, lags $= \lfloor T^{0.25} \rfloor$) in parentheses. Sample: 2002-09 to 2024-12, $T$=268.}
\end{spacing}
\vspace{-4mm}
\begin{center}
\label{tab:mmn_1}
\vspace{2mm}
\scalebox{0.75}{%
\begin{tabular}{l rr rr rr rr rr}
\toprule
 & \multicolumn{2}{c}{Unadjusted (1)} & \multicolumn{2}{c}{Adj.\ Signal (2)} & \multicolumn{2}{c}{Adj.\ Return (3)} & \multicolumn{2}{c}{Bias (1)$-$(2)} & \multicolumn{2}{c}{Bias (1)$-$(3)} \\
\cmidrule(lr){2-3} \cmidrule(lr){4-5} \cmidrule(lr){6-7} \cmidrule(lr){8-9} \cmidrule(lr){10-11}
Factor & $\mu$ & $\alpha$ & $\mu$ & $\alpha$ & $\mu$ & $\alpha$ & $\Delta\mu$ & $\Delta\alpha$ & $\Delta\mu$ & $\Delta\alpha$ \\
\midrule
\multicolumn{11}{c}{\textbf{Panel A:} Single-Sort} \\
\midrule
\texttt{ytm} & 1.20 & 0.69 & 1.05 & 0.54 & 0.92 & 0.44 & 0.15 & 0.15 & 0.28 & 0.25 \\
 & (2.97) & (2.50) & (2.63) & (1.96) & (2.64) & (1.80) & (6.54) & (6.06) & (3.90) & (4.33) \\
\addlinespace
\texttt{cs} & 1.13 & 0.66 & 0.96 & 0.50 & 0.82 & 0.38 & 0.17 & 0.17 & 0.31 & 0.28 \\
 & (2.89) & (2.39) & (2.52) & (1.82) & (2.45) & (1.55) & (6.23) & (6.17) & (4.14) & (4.54) \\
\addlinespace
\texttt{bbtm} & 0.93 & 0.54 & 0.75 & 0.35 & 0.66 & 0.29 & 0.18 & 0.18 & 0.27 & 0.25 \\
 & (2.19) & (1.69) & (1.80) & (1.12) & (1.80) & (1.04) & (6.20) & (5.24) & (3.39) & (3.56) \\
\addlinespace
\texttt{dcs6} & $-$1.02 & $-$0.79 & $-$0.50 & $-$0.29 & $-$0.46 & $-$0.25 & $-$0.52 & $-$0.51 & $-$0.56 & $-$0.54 \\
 & ($-$4.58) & ($-$4.55) & ($-$2.52) & ($-$1.79) & ($-$2.50) & ($-$1.71) & ($-$8.21) & ($-$8.12) & ($-$7.83) & ($-$7.73) \\
\addlinespace
\texttt{val\_ipr} & 0.96 & 0.76 & 0.51 & 0.31 & 0.51 & 0.33 & 0.45 & 0.45 & 0.45 & 0.43 \\
 & (4.82) & (5.24) & (2.68) & (2.12) & (3.05) & (2.61) & (6.93) & (6.42) & (7.12) & (7.52) \\
\addlinespace
\texttt{val\_hz} & 0.81 & 0.47 & 0.45 & 0.12 & 0.43 & 0.12 & 0.36 & 0.34 & 0.38 & 0.35 \\
 & (3.80) & (3.71) & (2.29) & (1.13) & (2.32) & (1.08) & (6.60) & (6.43) & (6.21) & (6.01) \\
\addlinespace
\texttt{str} & $-$0.99 & $-$0.77 & $-$0.09 & 0.12 & $-$0.17 & 0.03 & $-$0.90 & $-$0.89 & $-$0.82 & $-$0.80 \\
 & ($-$4.46) & ($-$3.58) & ($-$0.51) & (0.68) & ($-$0.98) & (0.18) & ($-$8.70) & ($-$8.20) & ($-$7.34) & ($-$7.30) \\
\addlinespace
\midrule
\multicolumn{11}{c}{\textbf{Panel B:} Within-Firm Sort} \\
\midrule
\texttt{ytm} & 0.56 & 0.31 & 0.30 & 0.06 & 0.26 & 0.04 & 0.26 & 0.25 & 0.30 & 0.27 \\
 & (4.42) & (4.08) & (2.85) & (1.44) & (2.68) & (0.94) & (5.42) & (5.54) & (5.60) & (5.58) \\
\addlinespace
\texttt{cs} & 0.63 & 0.47 & 0.26 & 0.11 & 0.24 & 0.10 & 0.38 & 0.36 & 0.40 & 0.37 \\
 & (5.57) & (4.43) & (3.10) & (1.72) & (3.11) & (1.62) & (6.34) & (6.36) & (6.42) & (6.11) \\
\addlinespace
\texttt{bbtm} & 0.38 & 0.28 & 0.19 & 0.10 & 0.17 & 0.09 & 0.19 & 0.18 & 0.21 & 0.19 \\
 & (4.02) & (4.11) & (2.50) & (1.86) & (2.32) & (1.64) & (6.05) & (6.58) & (5.64) & (5.56) \\
\addlinespace
\texttt{dcs6} & $-$0.70 & $-$0.65 & $-$0.26 & $-$0.22 & $-$0.25 & $-$0.22 & $-$0.44 & $-$0.43 & $-$0.45 & $-$0.43 \\
 & ($-$7.70) & ($-$8.08) & ($-$4.63) & ($-$4.98) & ($-$4.59) & ($-$4.92) & ($-$7.56) & ($-$7.44) & ($-$8.21) & ($-$8.25) \\
\addlinespace
\texttt{val\_ipr} & 0.62 & 0.57 & 0.22 & 0.18 & 0.22 & 0.20 & 0.40 & 0.39 & 0.39 & 0.37 \\
 & (6.02) & (7.27) & (3.50) & (4.25) & (3.58) & (4.30) & (7.41) & (7.47) & (7.38) & (7.85) \\
\addlinespace
\texttt{val\_hz} & 0.50 & 0.35 & 0.21 & 0.07 & 0.19 & 0.06 & 0.29 & 0.28 & 0.31 & 0.29 \\
 & (4.80) & (3.68) & (2.67) & (1.14) & (2.68) & (1.02) & (5.78) & (5.91) & (5.68) & (5.49) \\
\addlinespace
\texttt{str} & $-$0.65 & $-$0.61 & $-$0.10 & $-$0.07 & $-$0.11 & $-$0.10 & $-$0.55 & $-$0.54 & $-$0.54 & $-$0.52 \\
 & ($-$7.68) & ($-$7.12) & ($-$1.89) & ($-$1.36) & ($-$1.91) & ($-$1.69) & ($-$8.37) & ($-$8.51) & ($-$8.48) & ($-$8.77) \\
\addlinespace
\bottomrule
\end{tabular}
}
\end{center}
\end{table}

Cumulative returns in Fig.~\ref{fig:mmn_figure_1} contrast standard and adjusted approaches. For short-term reversal (Panels~A--B), \$1 invested in the unadjusted single-sort factor grows to \$12 by the end of 2024, but only to \$1.1 (adjusted signal) and \$1.4 (adjusted return); the within-firm factor falls from \$5.6 unadjusted to \$1.3 under both adjustments. Credit spread (Panels~C--D) retains more of its single-sort premium (\$9.3 adjusted signal, \$6.5 adjusted return, vs. \$14 unadjusted), but the within-firm factor drops from \$14 unadjusted to \$1.8 (adjusted signal) and \$2.0 (adjusted return).
\begin{figure}[h!]
\begin{center}
\includegraphics[width=\textwidth,height=0.85\textheight,keepaspectratio]{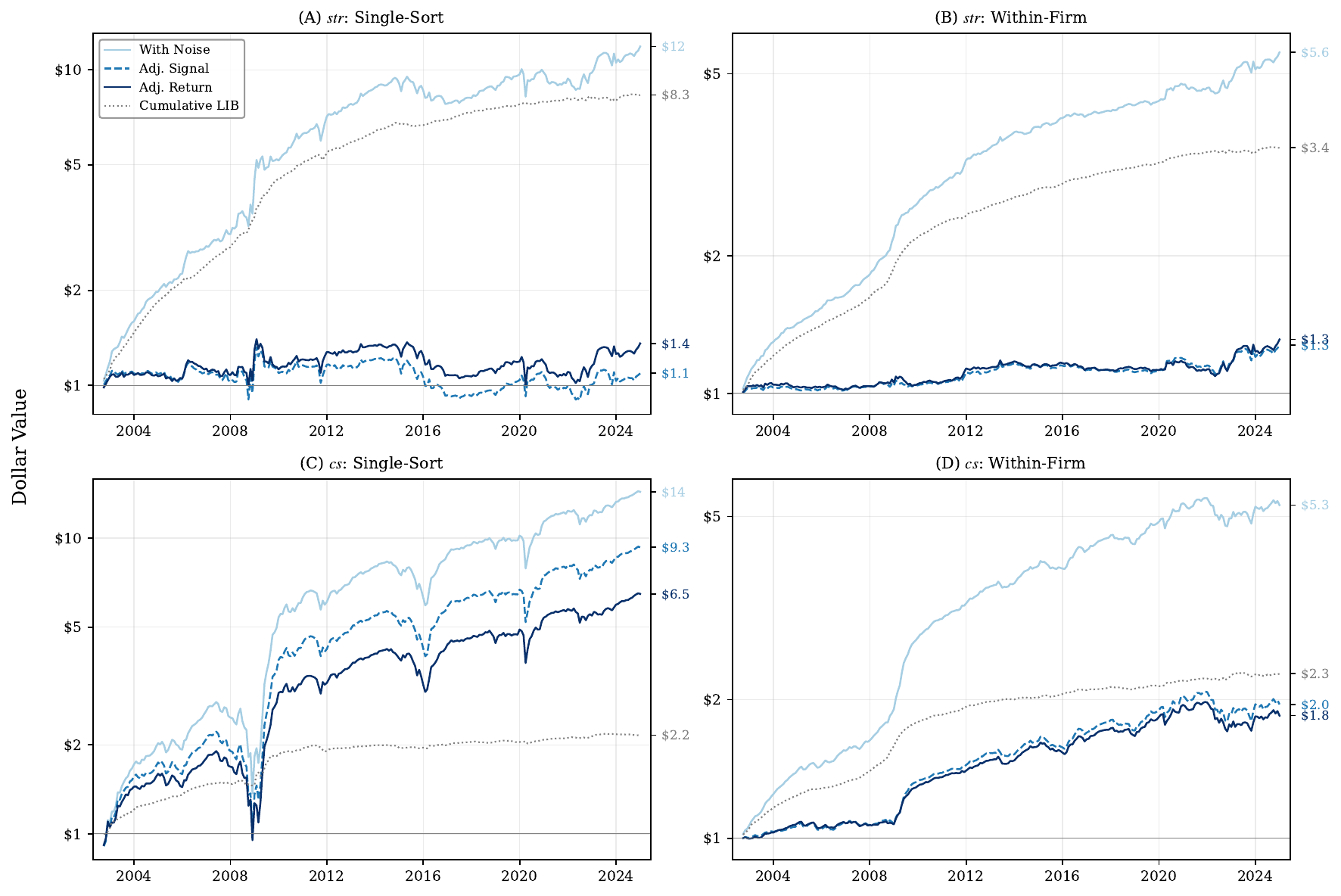}
\end{center}
\vspace{-4mm}
\caption{Cumulative factor returns under standard and adjusted approaches.}
\vspace{-2mm}
\begin{justify}
\begin{spacing}{1}
\footnotesize{
The figure shows the growth of \$1 invested in long-short factors under standard and adjusted approaches. Panels~A--B plot short-term reversal (\texttt{str}). Panels~C--D plot credit spread (\texttt{cs}). The left column reports single-sort factors, while the right column within-firm sorts. The dotted gray line tracks cumulative Latent Implementation Bias (LIB). The $y$-axis uses log scale. The \texttt{str} factor is sign-corrected to have a positive premium. Value-weighted portfolios. Sample: 2002-09 to 2024-12, $T$=268.
}
\end{spacing}
\end{justify}
\vspace{-12mm}
\label{fig:mmn_figure_1}
\end{figure}

Panels~A--B of Fig.~\ref{fig:mmn_figure_2} report average monthly biases with 95\% confidence intervals; all biases are statistically different from zero. The largest biases appear for \texttt{str} and \texttt{dcs6}, followed by the value factors; spread-based factors exhibit smaller but still substantial biases of 15--30 basis points. Panels~C--D decompose unadjusted returns into the implementable component and LIB. For reversal, LIB constitutes over 80\% of the unadjusted premium; value and spread-change factors show LIB shares of 45--55\%; yield and credit spread factors have LIB shares of 20--50\%. Within-firm factors exhibit larger LIB shares than their single-sort counterparts.

\begin{figure}[h!]
\begin{center}
\includegraphics[width=\textwidth,height=0.85\textheight,keepaspectratio]{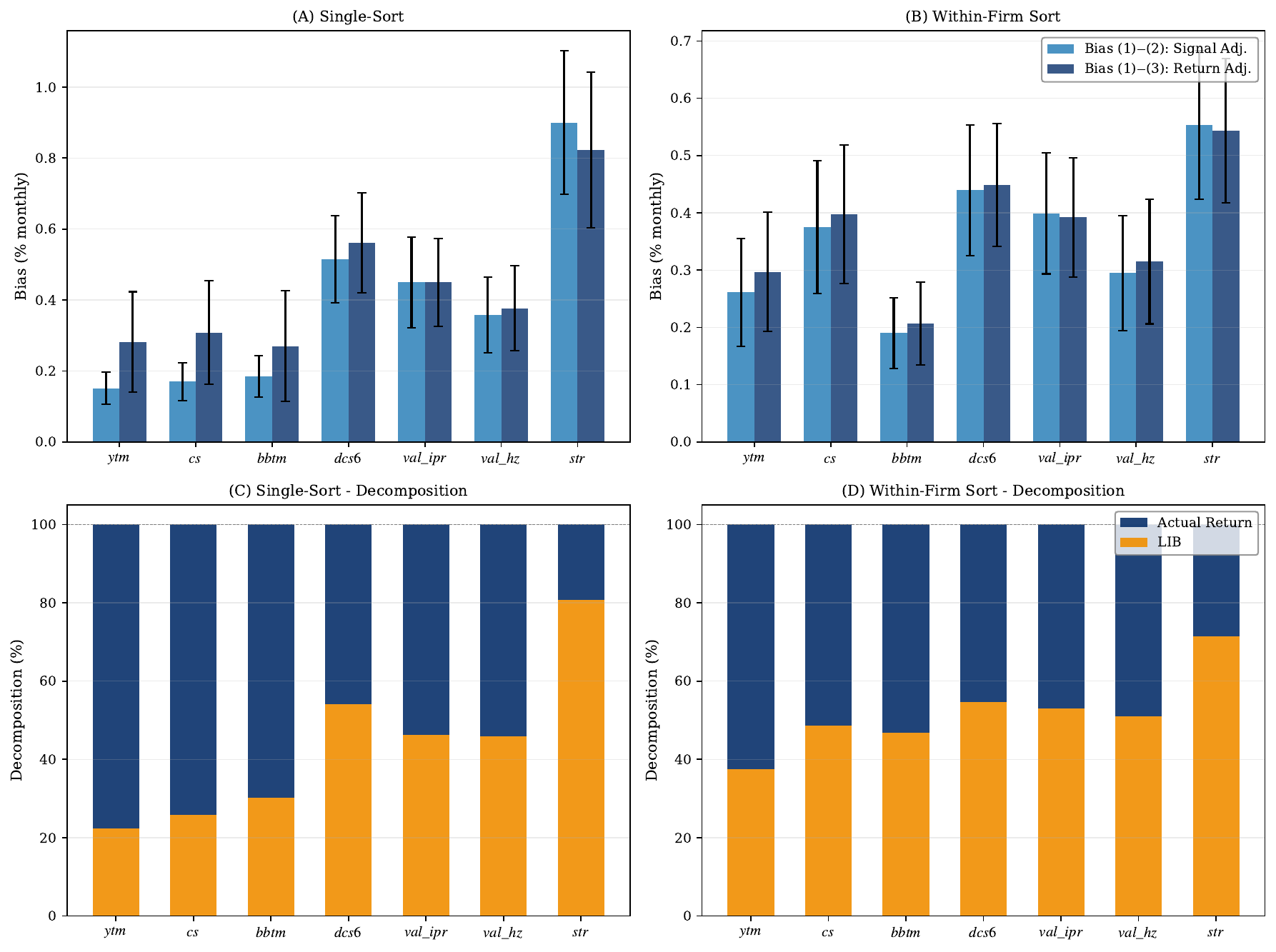}
\end{center}
\vspace{-4mm}
\caption{Latent implementation bias magnitude and return decomposition for price-based factors.}
\vspace{-2mm}
\begin{justify}
\begin{spacing}{1}
\footnotesize{
The figure displays the magnitude of Latent Implementation Bias (LIB) and its share of unadjusted returns for seven price-based factors. Panels~A--B plot the magnitude of the bias for single-sort and within-firm factors. Bias~(1)$-$(2) measures the difference between standard and signal-adjusted approaches. Bias~(1)$-$(3) measures the difference between standard and return-adjusted approaches. Bias is sign-corrected to be positive. Error bars show 95\% confidence intervals using Newey-West standard errors with lags $= \lfloor T^{0.25} \rfloor$. Panels~C--D decompose unadjusted returns into the LIB portion (orange, percentage of the unadjusted mean attributable to LIB) and the implementable component (dark blue). Stacked bars sum to 100\%. Sample: 2002-09 to 2024-12, $T$=268.
}
\end{spacing}
\end{justify}
\vspace{-12mm}
\label{fig:mmn_figure_2}
\end{figure}

To validate the LIB decomposition in Eq.~\eqref{eq:lib}, Table~\ref{tab:mmn_2} compares the month-end return, month-begin return, their difference, and the directly estimated LIB. If the decomposition holds, the difference $r^{\text{End}}_{i,t+1} - r^{\text{Bgn}}_{i,t+1}$ should approximately equal the directly estimated LIB term $\hat{\eta}_{\text{LIB}}$, and the two series should be nearly perfectly correlated. The decomposition holds tightly for single-sort factors: all seven factors show correlations of 0.99 between $(r^{\text{End}} - r^{\text{Bgn}})$ and LIB, with residuals of at most 0.02 percentage points. Within-firm sorts show slightly lower correlations (0.93--0.98) and larger residuals (up to 0.09 percentage points), but the decomposition remains economically tight.

\begin{table}[!ht]
\caption{Latent implementation bias validation for price-based factors.}
\begin{spacing}{1}
	{\footnotesize
The table tests the decomposition $r^{\text{End}} \approx \eta_{\text{LIB}} + r^{\text{Bgn}}$ for seven price-based factors. Panel~A reports single-sort factors. Panel~B reports within-firm factors. $\mu_{\text{End}}$ is the month-end average return (\%), $\mu_{\text{Bgn}}$ is the month-begin average return (\%), $\Delta\mu = \mu_{\text{End}} - \mu_{\text{Bgn}}$, $\mu_{\text{LIB}}$ is the average LIB (\%). $\hat{\rho}$ is the time series correlation between $(r^{\text{End}} - r^{\text{Bgn}})$ and LIB. Residual $= \Delta\mu - \mu_{\text{LIB}}$. $t$-statistics use Newey-West (lags $= \lfloor T^{0.25} \rfloor$). Sample: 2002-09 to 2024-12, $T$=268.}
\end{spacing}
\vspace{-4mm}
\begin{center}
\label{tab:mmn_2}
\vspace{2mm}
\scalebox{0.80}{%
\begin{tabular}{l rr rr rr r r}
\toprule
Factor & $\mu_{\text{End}}$ & $\mu_{\text{Bgn}}$ & $\Delta\mu$ & $t_{\Delta\mu}$ & $\mu_{\text{LIB}}$ & $t_{\mu_{\text{LIB}}}$ & $\hat{\rho}$ & Residual \\
\midrule
\multicolumn{9}{c}{\textbf{Panel A:} Single-Sort} \\
\midrule
\texttt{ytm} & 1.20 & 0.92 & 0.28 & (3.90) & 0.27 & (3.67) & 0.99 & 0.01 \\
\addlinespace
\texttt{cs} & 1.13 & 0.82 & 0.31 & (4.14) & 0.29 & (3.84) & 0.99 & 0.02 \\
\addlinespace
\texttt{bbtm} & 0.93 & 0.66 & 0.27 & (3.39) & 0.28 & (3.47) & 0.99 & $-$0.01 \\
\addlinespace
\texttt{dcs6} & $-$1.02 & $-$0.46 & $-$0.56 & ($-$7.83) & $-$0.55 & ($-$7.63) & 0.99 & $-$0.01 \\
\addlinespace
\texttt{val\_ipr} & 0.96 & 0.51 & 0.45 & (7.12) & 0.44 & (6.81) & 0.99 & 0.01 \\
\addlinespace
\texttt{val\_hz} & 0.81 & 0.43 & 0.38 & (6.21) & 0.37 & (5.87) & 0.99 & 0.01 \\
\addlinespace
\texttt{str} & $-$0.99 & $-$0.17 & $-$0.82 & ($-$7.34) & $-$0.80 & ($-$7.18) & 0.99 & $-$0.02 \\
\addlinespace
\midrule
\multicolumn{9}{c}{\textbf{Panel B:} Within-Firm Sort} \\
\midrule
\texttt{ytm} & 0.56 & 0.26 & 0.30 & (5.60) & 0.21 & (5.31) & 0.93 & 0.09 \\
\addlinespace
\texttt{cs} & 0.63 & 0.24 & 0.40 & (6.42) & 0.31 & (6.21) & 0.96 & 0.09 \\
\addlinespace
\texttt{bbtm} & 0.38 & 0.17 & 0.21 & (5.64) & 0.18 & (5.35) & 0.96 & 0.03 \\
\addlinespace
\texttt{dcs6} & $-$0.70 & $-$0.25 & $-$0.45 & ($-$8.21) & $-$0.38 & ($-$7.95) & 0.98 & $-$0.06 \\
\addlinespace
\texttt{val\_ipr} & 0.62 & 0.22 & 0.39 & (7.38) & 0.33 & (7.29) & 0.95 & 0.06 \\
\addlinespace
\texttt{val\_hz} & 0.50 & 0.19 & 0.31 & (5.68) & 0.26 & (5.76) & 0.95 & 0.06 \\
\addlinespace
\texttt{str} & $-$0.65 & $-$0.11 & $-$0.54 & ($-$8.48) & $-$0.46 & ($-$8.19) & 0.94 & $-$0.08 \\
\addlinespace
\bottomrule
\end{tabular}
}
\end{center}
\end{table}

One concern is whether the gap procedure drives differences in premia through an altered holding period or economic events during the gap window, rather than by removing LIB. Table~\ref{tab:lib_summary} in Appendix~\ref{app:lib-empirical} provides a discriminating test. We compute the month-end minus month-begin return difference for all 108 factors using unadjusted portfolio weights, comparing Approach~1 (month-end returns) with Approach~3 (month-begin returns). Among the 78 non-price-based factors, where Corollary~\ref{cor:nonprice} predicts no CEIV bias, the average absolute return difference is 0.03--0.10\% per month with average $|t|$-statistics near unity. Only 10--18\% of factors show a statistically significant gap, depending on the sort. Among the 30 price-based factors, 43--67\% show a significant gap depending on the sort, with average absolute differences of 0.13--0.30\% per month. For illiquidity factors, which are constructed from within-month averages of daily observations, the shared-price link between signal noise and return noise is dampened. The Internet Appendix confirms that LIB estimates for five illiquidity characteristics are economically small and mostly statistically indistinguishable from zero. CAPMB alphas for all five illiquidity factors range from 0.10\% to 0.20\% per month with $|t| < 1.65$, and within-firm alphas are indistinguishable from zero: the illiquidity factor premium is at best marginal.

Measurement error in transaction prices inflates price-based factor premia by 15 to 90 basis points per month. Non-price-based factors and within-month illiquidity averages are largely unaffected. Infrequent trading and wide bid-ask spreads also produce volatile bond-level returns: monthly returns below $-$50\% or above 50\% are not uncommon, particularly during financial distress. Such returns, though economically important, invite ad hoc filtering, such as winsorizing or trimming at sample-wide thresholds, which introduces a second form of bias that we examine next.


\section{Look-ahead bias}\label{sec:lab}

This section turns to ex-post return filtering, a second source of bias that affects all factors, including non-price-based ones. Historical cross-sectional backtests require some look ahead in sample definition. To compute a return for month $t+1,$ a bond must trade near both month-end $t$ and month-end $t+1,$ so that bonds without valid future prices must be excluded at portfolio formation. This \emph{required} look ahead is inherent to historical backtests and affects all researchers equally. The same issue arises in the Center for Research in Security Prices (CRSP) database, though less often because stocks trade daily. In TRACE, where many bonds trade infrequently, it is unavoidable.
\emph{Avoidable} look-ahead bias arises when return filtering thresholds are computed from the full sample rather than from data available at portfolio formation. We use ``ex-ante'' and ``ex-post'' to describe the information set used by the filtering rule rather than the return. Although all returns are realized, the distinction revolves around whether the threshold applied to those returns uses only information available at portfolio formation (ex-ante) or the sample including future months (ex-post).
Ex-post filtering is widespread in corporate bond research despite being rarely stated explicitly. For many factors, reported factor premia and alphas are statistically significant only when ex-post filters are applied, commonly rationalized by outlier removal. At month $t$, the threshold $\tau^{\text{ex-post}} = \text{Percentile}_q(\{r_{i,s}\}_{s=1}^{T})$ incorporates information from months $t{+}1, \ldots, T$ that has not yet been realized. Filtering with such thresholds creates an infeasible trading strategy, one that no investor could have implemented in real time.
We define the \emph{Look-Ahead Bias} (LAB) as the difference between the infeasible (reported) factor return and the feasible (achievable) return. \citet*{duarte2023too} document the same mechanism in equity options, where ex-post sample filters selectively remove the worst-performing observations and inflate reported Sharpe ratios from 0.5 to above 5. The \emph{base} (feasible) factor applies all transaction-level price filters (Dick-Nielsen corrections, decimal shift, bounce-back) and the signal gap procedure from Section~\ref{sec:mmn}, but does not apply any return-level winsorization or trimming. The \emph{winsorized} (infeasible) factor additionally caps returns at percentile thresholds computed from the full sample. Transaction-level price filters already address data recording errors, whereas return-level winsorization applied to clean data selectively removes extreme but genuine returns. The Internet Appendix documents the prevalence and time concentration of extreme returns, with 15 bond-month observations having month-end returns below $-$95\% and 385 exceeding 95\%. Consider a long-short factor formed by sorting bonds on signal $\hat{s}_{i,t}$, where the long leg holds bonds in the top portfolio, the short leg holds bonds in the bottom portfolio, and both legs are equal-weighted. Let $r_{t+1}^{\text{LS}} = r_{t+1}^{L} - r_{t+1}^{S}$ denote the feasible long-short return using unmodified returns. Under ex-post winsorizing, each bond return $r_{i,t+1}$ is replaced by $\tilde{r}_{i,t+1} = \max(\tau_L, \min(r_{i,t+1}, \tau_U))$, where $\tau_L$ and $\tau_U$ are the lower and upper thresholds. The winsorizing adjustment $\Delta_{i,t+1} \equiv \tilde{r}_{i,t+1} - r_{i,t+1}$ is positive when returns are floored (left tail) and negative when capped (right tail). The look-ahead bias is
\begin{equation}
\text{LAB}_{t+1} \equiv \tilde{r}_{t+1}^{\text{LS}} - r_{t+1}^{\text{LS}} = \underbrace{\frac{1}{N_t^L} \sum_{i \in \text{Long}} \Delta_{i,t+1}}_{\text{LAB}_{t+1}^{L}} - \underbrace{\frac{1}{N_t^S} \sum_{i \in \text{Short}} \Delta_{i,t+1}}_{\text{LAB}_{t+1}^{S}}.
\label{eq:lab_def}
\end{equation}
The bias decomposes into contributions from each leg, with $\text{LAB}^L$ capturing adjustments to long positions and $\text{LAB}^S$ capturing adjustments to short positions. A positive LAB indicates that the reported factor return overstates achievable performance.\footnote{LAB can be decomposed into a weight component (different portfolio composition under filtering) and a return-evaluation component (different return magnitudes for the same bonds). The return-evaluation component dominates: winsorization directly alters measured returns while portfolio composition changes are negligible in value-weighted portfolios with many bonds.}
The direction of LAB depends on where extreme returns concentrate. If the long leg contains more left-tail returns (e.g., high-risk bonds that crash during stress), left-tail winsorizing benefits the long leg disproportionately, inflating the factor return. If the short leg contains more right-tail returns, right-tail winsorizing protects the short leg from losses, inflating the factor return.
The magnitude of LAB is largest during financial distress. In episodes like the 2008--2009 Great Recession or COVID-19 in early 2020, return distributions become highly asymmetric. One leg of the portfolio experiences a disproportionate share of extreme returns, and the winsorizing adjustment becomes unbalanced.
Fig.~\ref{fig:lookahead_bias} contrasts ex-post and ex-ante filtering. In Panel~A, the threshold $\tau^{\text{ex-post}}$ is computed using returns from the full sample, including months $t{+}1, \ldots, T$ that lie in the investor's future. In Panel~B, the ex-ante threshold $\tau^{\text{ex-ante}}_t = \text{Percentile}_q(\{r_{i,s}\}_{s=1}^{t})$ uses only returns through month $t$, making it computable in real time. Ex-ante filtering uses only information available at portfolio formation and therefore yields a feasible trading strategy.

\begin{figure}[h!]
\begin{center}
\vspace{.25cm}
\scalebox{0.95}{%
\begin{tikzpicture}[
    >=Stealth,
    month/.style={draw, minimum width=2.4cm, minimum height=0.6cm, font=\small, anchor=east},
    timepoint/.style={circle, fill=black, inner sep=1.5pt},
    lbl/.style={font=\footnotesize},
    panel/.style={font=\bfseries\small},
    paneldesc/.style={font=\small},
    warn/.style={black},
    safe/.style={black}
]

\node[panel, anchor=west] at (-6.5, 3.2) {Panel A:};
\node[paneldesc, anchor=west] at (-4.8, 3.2) {Ex-Post Filter (Look-Ahead Bias)};

\draw[thick, ->] (-6.5, 1.0) -- (8.5, 1.0) node[right, font=\footnotesize] {time};

\node[timepoint] (p0a) at (-3.5, 1.0) {};
\node[timepoint] (p1a) at (-0.5, 1.0) {};
\node[timepoint] (p2a) at (2.5, 1.0) {};

\node[month] at (-3.5, 2.4) {Month $t{-}1$};
\node[month] at (-0.5, 2.4) {Month $t$};
\node[month] at (2.5, 2.4) {Month $t{+}1$};
\node[month] at (5.5, 2.4) {...};
\node[month] at (8.0, 2.4) {Month $T$};

\node[lbl, above=2pt] at (p0a) {$P_{t-1}^{\text{end}}$};
\node[lbl, above=2pt] at (p1a) {$P_t^{\text{end}}$};
\node[lbl, above=2pt] at (p2a) {$P_{t+1}^{\text{end}}$};

\node[lbl, below=6pt] at (p1a) {\textbf{Formation}};

\draw[decorate, decoration={brace, amplitude=4pt, raise=2pt}]
    (p1a.north) -- (p2a.north) node[midway, above=8pt, lbl] {$r_{t+1}^{\text{end}}$};

\draw[thick, warn, ->] (-6.0, 0.3) -- (8.0, 0.3);
\node[lbl, warn, anchor=west] at (-6.0, -0.1) {$\tau^{\text{ex-post}}$ uses \textbf{full sample}};
\node[lbl, warn, anchor=west] at (-6.0, -0.5) {(includes months $t{+}1, \ldots, T$)};

\node[panel, anchor=west] at (-6.5, -2.0) {Panel B:};
\node[paneldesc, anchor=west] at (-4.8, -2.0) {Ex-Ante Filter (Feasible)};

\draw[thick, ->] (-6.5, -4.2) -- (6.5, -4.2) node[right, font=\footnotesize] {time};

\node[timepoint] (p1b) at (1.0, -4.2) {};
\node[timepoint] (p2b) at (4.0, -4.2) {};

\node[month] at (-4.5, -2.8) {Month $1$};
\node[month] at (-2.0, -2.8) {...};
\node[month] at (1.0, -2.8) {Month $t$};
\node[month] at (4.0, -2.8) {Month $t{+}1$};

\node[lbl, above=2pt] at (p1b) {$P_t^{\text{end}}$};
\node[lbl, above=2pt] at (p2b) {$P_{t+1}^{\text{end}}$};

\node[lbl, below=6pt] at (p1b) {\textbf{Formation}};

\draw[decorate, decoration={brace, amplitude=4pt, raise=2pt}]
    (p1b.north) -- (p2b.north) node[midway, above=8pt, lbl] {$r_{t+1}^{\text{end}}$};

\draw[thick, safe, <-] (-6.0, -4.9) -- (0.7, -4.9);
\node[lbl, safe, anchor=west] at (-6.0, -5.3) {$\tau^{\text{ex-ante}}_t$ uses \textbf{past only}};
\node[lbl, safe, anchor=west] at (-6.0, -5.7) {(includes months $1, \ldots, t$)};

\end{tikzpicture}
} 
\end{center}
\caption{Comparison of ex-post (full-sample) and ex-ante (rolling-window) return filtering.} 
\begin{justify}
\begin{spacing}{1}
\footnotesize{
The figure contrasts ex-post and ex-ante return filtering. Panel~A depicts ex-post filtering, where thresholds $\tau^{\text{ex-post}}$ are computed from the full sample including future months $t{+}1, \ldots, T$, embedding future information into portfolio construction. Panel~B depicts the ex-ante alternative, where thresholds $\tau^{\text{ex-ante}}_t$ use only current data available at portfolio formation.
}
\end{spacing}
\end{justify}
\label{fig:lookahead_bias}
\end{figure}
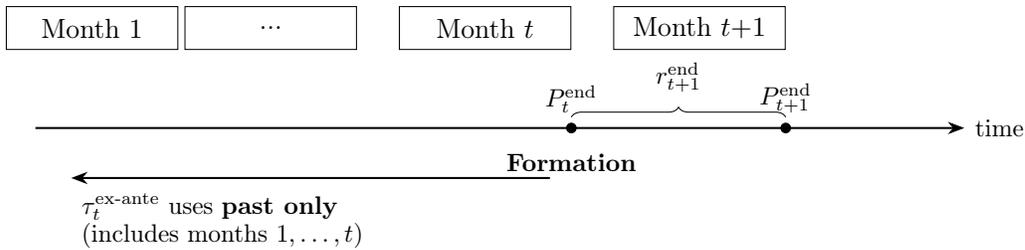

\subsection{Factors sensitive to ex-post return filtering}
Table~\ref{tab:lab_affected_factors} reports factors from influential corporate bond studies (including momentum, idiosyncratic volatility, downside risk, and macroeconomic exposures) whose statistical significance depends on the choice of asymmetric return filtering.
\begin{table}[!ht]
\caption{Corporate bond factors sensitive to ex-post return filtering.}
\begin{spacing}{1}
	{\footnotesize 
The table lists factors from published corporate bond studies whose statistical significance depends on asymmetric return filtering. The ``Wins.\ Threshold'' column reports the representative percentile threshold at which significance changes, where $0.50^{\text{th}}$ indicates left-tail winsorization at the 0.5th percentile and $99.50^{\text{th}}$ indicates right-tail winsorization. \citet*{jostova2013momentum} use the $99.50^{\text{th}}$ percentile (Footnote 16 of their paper). The remaining papers do not state their thresholds. Sample: 2002-08 to 2024-12, $T$=269.
}
\end{spacing}
\vspace{-4mm}
\begin{center}
\label{tab:lab_affected_factors}
\vspace{2mm}
\scalebox{0.75}{%
\begin{tabular}{l p{5.5cm} l l}
\toprule
Mnemonic & Factor Description & Journal & Wins.\ Threshold \\
\midrule
        \texttt{b\_dunc} & Macro uncertainty (1-month $\Delta$) beta & JFQA, \citet*{bali2021macroeconomic} & $0.50^{\text{th}}$ \\
        \texttt{b\_dunc3} & Macro uncertainty (3-month $\Delta$) beta & JFQA, \citet*{bali2023macroeconomic} & $0.50^{\text{th}}$ \\
        \texttt{b\_dunc6} & Macro uncertainty (6-month $\Delta$) beta & JFQA, \citet*{bali2023macroeconomic} & $0.50^{\text{th}}$ \\
        \texttt{b\_unc} & Macro uncertainty level beta & JFQA, \citet*{bali2023macroeconomic} & $0.50^{\text{th}}$ \\
        \texttt{mom3\_1} & 3-month momentum & RFS, \citet*{jostova2013momentum} & $99.50^{\text{th}}$ \\
        \texttt{mom6\_1} & 6-month momentum & RFS, \citet*{jostova2013momentum} & $99.50^{\text{th}}$ \\
        \texttt{mom12\_1} & 12-month momentum & RFS, \citet*{jostova2013momentum} & $99.50^{\text{th}}$ \\
        \texttt{ltr48\_12} & Long-term reversal (48-12) & JFE, \citet*{bali2021long} & $0.50^{\text{th}}$ \\
        \texttt{ltr30\_6} & Long-term reversal (30-6) & WP, \citet*{subrahmanyam2023corporatebonddata} & $0.50^{\text{th}}$ \\
        \texttt{ivol\_bbw} & Idiosyncratic volatility (BBW) & JFE, \citet*{bai2021there} & $0.50^{\text{th}}$ \\
        \texttt{ivol\_vp} & Idiosyncratic volatility (VIX-PSB) & JFE, \citet*{ChungWangWu_2019} & $0.50^{\text{th}}$ \\
        \texttt{var\_95} & Historical 95\% value-at-risk & JFE, \citet*{BaiBaliWen_2019} & $0.50^{\text{th}}$ \\
        \texttt{es\_90} & Historical 90\% expected shortfall & JFE, \citet*{BaiBaliWen_2019} & $0.50^{\text{th}}$ \\
        \texttt{b\_dvix\_vp} & VIX innovation beta (PSB) & JFE, \citet*{ChungWangWu_2019} & $0.50^{\text{th}}$ \\
        \texttt{b\_psb\_m} & Pastor-Stambaugh beta (multi-factor) & JFE, \citet*{LinWangWu_2011} & $0.50^{\text{th}}$ \\
        \texttt{b\_amd\_m} & Amihud beta (multi-factor) & JFE, \citet*{LinWangWu_2011} & $0.50^{\text{th}}$ \\
\bottomrule
\end{tabular}
}
\end{center}
\end{table}
Table~\ref{tab:lab_affected_factors} shows that factors based on characteristics -- volatility, Value at Risk (VaR), Expected Shortfall (ES), macro uncertainty -- require left-tail winsorization to produce significant alphas, while momentum factors require right-tail winsorization. This pattern is consistent with the LAB decomposition in Eq.~\eqref{eq:lab_def}, whereby factors that sort on volatility, VaR, or macro uncertainty place bonds with large tail exposures in the long leg where left-tail winsorizing provides protection, while momentum factors place past losers in the short leg where right-tail winsorizing caps rebounds for bonds sold short. The winsorization threshold that ``works'' is determined by the sorting characteristic, specifically by whether the tail in which extreme returns concentrate corresponds to the long leg holding bonds with large downside exposure (left tail) or the short leg holding bonds prone to rebounds (right tail).
Table~\ref{tab:lab_ls_1} quantifies LAB for each factor. We sort bonds into value-weighted decile portfolios, going long P10 and short P1. Panel~A reports factors sensitive to left-tail winsorization and Panel~B reports factors sensitive to right-tail winsorization. For each factor, we compare the winsorized (infeasible) and base (feasible) premia and CAPMB alphas.

\begin{table}[!ht]
\caption{Look-ahead bias by factor.}
\begin{spacing}{1}
{\footnotesize
The table compares premia and CAPMB alphas with ($\tilde \mu$ and $\tilde \alpha$) and without ($\mu$ and $\alpha$) asymmetric ex-post return winsorization for the factors in Table~\ref{tab:lab_affected_factors}. Factors are sorted into value-weighted decile portfolios, long P10 and short P1. Bias is the difference between the winsorized (infeasible) and base (feasible) factor statistics. Panel~A if for factors sensitive to left-tail winsorization (0.50\textsuperscript{th} percentile). Panel~B is for factors sensitive to right-tail winsorization (99.50\textsuperscript{th} percentile). Factors have not been sign-corrected to have positive means. $t$-statistics use Newey-West standard errors with lags $= \lfloor T^{0.25} \rfloor$. Sample: 2002-08 to 2024-12, $T$=269.}
\end{spacing}
\vspace{-4mm}
\begin{center}
\label{tab:lab_ls_1}
\vspace{2mm}
\scalebox{0.78}{%
\begin{tabular}{l ccc ccc}
\toprule
 & \multicolumn{3}{c}{Mean Returns (\%)}  & \multicolumn{3}{c}{CAPMB Alpha (\%)} \\
\cmidrule(lr){2-4} \cmidrule(lr){5-7}
 & $\tilde{\mu}_{\text{LS}}$ & $\mu_{\text{LS}}$ & Bias & $\tilde{\alpha}_{\text{LS}}$ & $\alpha_{\text{LS}}$ & Bias \\
\midrule
\multicolumn{7}{l}{\textbf{Panel A:} Left-tail winsorized factors (0.50\textsuperscript{th} percentile)} \\
\midrule
        \texttt{b\_dunc} & $-$0.17 & $-$0.04 & $-$0.13 & $-$0.22 & $-$0.04 & $-$0.18 \\
         & ($-$1.03) & ($-$0.22) & ($-$2.14) & ($-$1.38) & ($-$0.22) & ($-$2.20) \\
        \texttt{b\_dunc3} & $-$0.36 & $-$0.20 & $-$0.16 & $-$0.32 & $-$0.11 & $-$0.21 \\
         & ($-$1.91) & ($-$1.04) & ($-$1.79) & ($-$1.87) & ($-$0.57) & ($-$1.88) \\
        \texttt{b\_unc} & $-$0.32 & $-$0.18 & $-$0.13 & $-$0.20 & $-$0.02 & $-$0.18 \\
         & ($-$1.24) & ($-$0.78) & ($-$1.58) & ($-$0.91) & ($-$0.09) & ($-$1.68) \\
        \texttt{ltr48\_12} & $-$0.40 & $-$0.32 & $-$0.08 & $-$0.33 & $-$0.23 & $-$0.10 \\
         & ($-$2.93) & ($-$2.21) & ($-$1.90) & ($-$2.57) & ($-$1.56) & ($-$1.65) \\
        \texttt{ltr30\_6} & $-$0.58 & $-$0.45 & $-$0.14 & $-$0.43 & $-$0.26 & $-$0.17 \\
         & ($-$2.76) & ($-$1.98) & ($-$2.30) & ($-$2.37) & ($-$1.26) & ($-$2.32) \\
        \texttt{ivol\_bbw} & 0.87 & 0.57 & 0.29 & 0.53 & 0.18 & 0.35 \\
         & (2.84) & (1.68) & (2.75) & (2.58) & (0.72) & (2.83) \\
        \texttt{ivol\_vp} & 0.78 & 0.52 & 0.26 & 0.40 & 0.08 & 0.31 \\
         & (2.99) & (1.76) & (2.39) & (2.48) & (0.44) & (2.49) \\
        \texttt{b\_dvix\_vp} & $-$0.02 & 0.02 & $-$0.04 & $-$0.01 & 0.04 & $-$0.05 \\
         & ($-$0.19) & (0.14) & ($-$1.36) & ($-$0.09) & (0.38) & ($-$1.60) \\
        \texttt{b\_psb\_m} & 0.07 & 0.15 & $-$0.08 & 0.12 & 0.23 & $-$0.11 \\
         & (0.43) & (0.80) & ($-$2.36) & (0.85) & (1.35) & ($-$2.47) \\
        \texttt{b\_amd\_m} & $-$0.29 & $-$0.24 & $-$0.05 & $-$0.32 & $-$0.26 & $-$0.06 \\
         & ($-$1.99) & ($-$1.68) & ($-$2.93) & ($-$2.07) & ($-$1.74) & ($-$2.74) \\
        \texttt{var\_95} & 1.05 & 0.78 & 0.26 & 0.56 & 0.24 & 0.32 \\
         & (3.50) & (2.65) & (2.69) & (3.01) & (1.42) & (2.71) \\
        \texttt{es\_90} & 1.26 & 0.91 & 0.35 & 0.74 & 0.31 & 0.43 \\
         & (3.53) & (2.59) & (2.38) & (3.10) & (1.42) & (2.45) \\
\midrule
\multicolumn{7}{l}{\textbf{Panel B:} Right-tail winsorized factors (99.50\textsuperscript{th} percentile)} \\
\midrule
        \texttt{mom3\_1} & 0.33 & 0.17 & 0.16 & 0.45 & 0.31 & 0.13 \\
         & (1.78) & (0.90) & (2.22) & (2.26) & (1.65) & (2.10) \\
        \texttt{mom6\_1} & 0.30 & 0.00 & 0.30 & 0.46 & 0.22 & 0.24 \\
         & (1.45) & (0.01) & (2.03) & (2.39) & (1.08) & (2.05) \\
        \texttt{mom12\_1} & 0.26 & $-$0.13 & 0.40 & 0.46 & 0.14 & 0.31 \\
         & (0.92) & ($-$0.39) & (2.21) & (1.60) & (0.51) & (2.33) \\
\bottomrule
\end{tabular}
}
\end{center}
\end{table}

The idiosyncratic volatility factors exhibit the largest biases among factors sensitive to left-tail winsorization. For \texttt{ivol\_bbw} \citep*{bai2021there}, the winsorized alpha of 0.53\% per month ($t = 2.58$) drops to 0.18\% ($t = 0.72$) without winsorization. The implied bias of 35 basis points accounts for 66\% of the measured alpha. \texttt{ivol\_vp} \citep*{ChungWangWu_2019} drops from 0.40\% ($t = 2.48$) to 0.08\% ($t = 0.44$), a 31-basis-point bias. Both factors sort high-volatility bonds into the long leg, where left-tail winsorization provides disproportionate protection during market stress.
The macroeconomic uncertainty factor \texttt{b\_dunc3} \citep*{bali2023macroeconomic}, which sorts on exposure to 3-month changes in the macroeconomic uncertainty index, has a winsorized alpha of $-$0.32\% ($t = -1.87$) that shrinks to $-$0.11\% ($t = -0.57$) in the feasible specification (a bias of 21 basis points). Without winsorization, the \texttt{ltr48\_12} \citep*{bali2021long} alpha shrinks from $-$0.33\% ($t = -2.57$) to $-$0.23\% ($t = -1.56$). In both cases, statistical significance depends on the infeasible filter.\footnote{Several factors fail to produce detectable premia or alphas even with ex-post winsorization: \texttt{b\_dunc}, \texttt{b\_unc}, \texttt{b\_dvix\_vp}, and \texttt{b\_psb\_m} have $|t| < 1.96$ for both winsorized premia and alphas.}
Momentum factors in Panel~B require right-tail winsorization. The 6-month momentum factor \texttt{mom6\_1} \citep*{jostova2013momentum} has a winsorized premium of 0.30\% but a base premium of 0.00\%. The alpha drops from 0.46\% ($t = 2.39$) to 0.22\% ($t = 1.08$). The 12-month momentum factor is more extreme, with the winsorized premium reaching 0.26\% compared to a base premium of $-$0.13\%, such that the 40-basis-point bias converts a negative premium into a positive one. Right-tail winsorization caps the rebounds of past losers in the short leg, artificially inflating momentum returns. The Internet Appendix decomposes these biases into contributions from the long and short legs. Even with winsorization, only \texttt{mom3\_1} achieves a marginally significant premium ($t = 1.78$); \texttt{mom6\_1} and \texttt{mom12\_1} remain insignificant. The $t$-statistics on the winsorized CAPMB alphas for \texttt{mom3\_1} ($t = 2.26$) and \texttt{mom6\_1} ($t = 2.39$) exceed conventional thresholds, but are an artifact of the infeasible specifications and turn statistically insignificant without ex-post filtering.
Fig.~\ref{fig:dua_figure_19} illustrates the danger of asymmetric ex-post trimming using the 6-month momentum factor of \citet*{jostova2013momentum} as a case study. Panels~A and~C apply ex-ante trimming: bonds with extreme returns are excluded from signal computation and portfolio formation using only historical information. Under ex-ante trimming, the CAPMB alpha is statistically indistinguishable from zero across all trimming thresholds in both tails. No feasible filter generates a detectable momentum alpha.

\begin{figure}[h!]
\centering
\scalebox{0.85}{\includegraphics[width=1.18\textwidth,height=0.85\textheight,keepaspectratio,trim=0 0 0 35pt,clip]{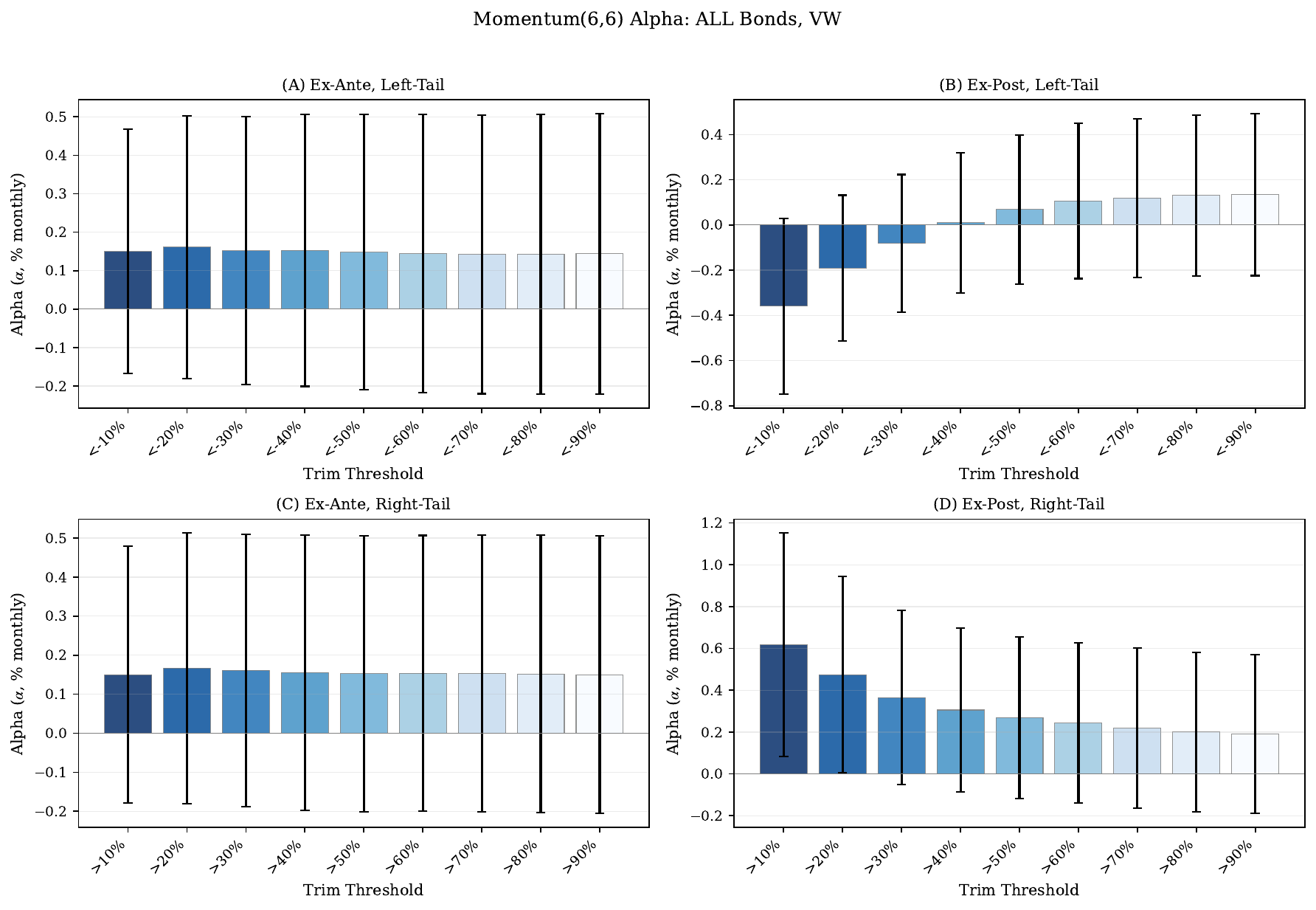}}
\caption{Sensitivity of momentum alpha to return filtering thresholds.}
\vspace{-2mm}
\begin{justify}
\begin{spacing}{1}
\footnotesize{
The figure plots CAPMB alphas for the six-month momentum factor (\texttt{mom6\_1}), formed with a staggered six-month holding period following \citet*{jostova2013momentum}, across a range of return filtering thresholds. Panels~A and~C apply ex-ante filtering, excluding returns below the threshold only up to portfolio formation and excluding bonds with returns below the threshold over the prior month. Panels~B and~D apply ex-post filtering, adding full-sample return exclusions that embed look-ahead bias. Bars show $\alpha$ (in \% monthly) with 95\% confidence intervals using Newey-West standard errors with lags $= \lfloor T^{0.25} \rfloor$. The $x$-axis shows the trim threshold ($<$$-$$\tau$\% for left-tail, $>$$\tau$\% for right-tail). Value-weighted portfolios. Sample: 2003-03 to 2024-12, $T$=262.}
\end{spacing}
\end{justify}
\vspace{-12mm}
\label{fig:dua_figure_19}
\end{figure}

In Panel~D, right-tail ex-post trimming at a threshold of 10\% per month generates a large alpha of approximately 0.60\% that is statistically distinguishable from zero. The alpha declines monotonically as the threshold increases, requiring aggressive trimming of moderate positive returns to manufacture the premium. Panel~B shows the mirror image for left-tail trimming, where at aggressive thresholds ($<$$-$10\%), the alpha is large and \emph{negative}, the opposite of a momentum factor with a positive alpha. The alpha rises monotonically and crosses zero only at thresholds around $-$40\%. This pattern illustrates the arbitrariness of ex-post filtering: a researcher who chooses right-tail trimming finds that momentum appears to work, whereas a researcher who chooses left-tail trimming finds the opposite. Neither result reflects an implementable trading strategy.
Fig.~\ref{fig:lab_figure_1} documents the time variation in LAB. Panels~A and~B plot monthly LAB for factors sensitive to left-tail and right-tail winsorization. The bias exhibits strong positive spikes during periods of market distress, specifically the Great Recession (2008--2009), the European debt crisis (2011--2012), and the COVID-19 shock (March 2020). For left-tail factors, these spikes arise because winsorization dampens large negative returns concentrated in the long leg during stress. For right-tail momentum factors, winsorization truncates the rebounds of past losers in the short leg. Panels~C and~D confirm this pattern, with scatter plots of LAB against the VIX level showing a strong positive convex relationship. The magnitude of LAB rises sharply when market volatility is elevated.

\begin{figure}[h!]
\centering
\includegraphics[width=\textwidth,height=0.85\textheight,keepaspectratio]{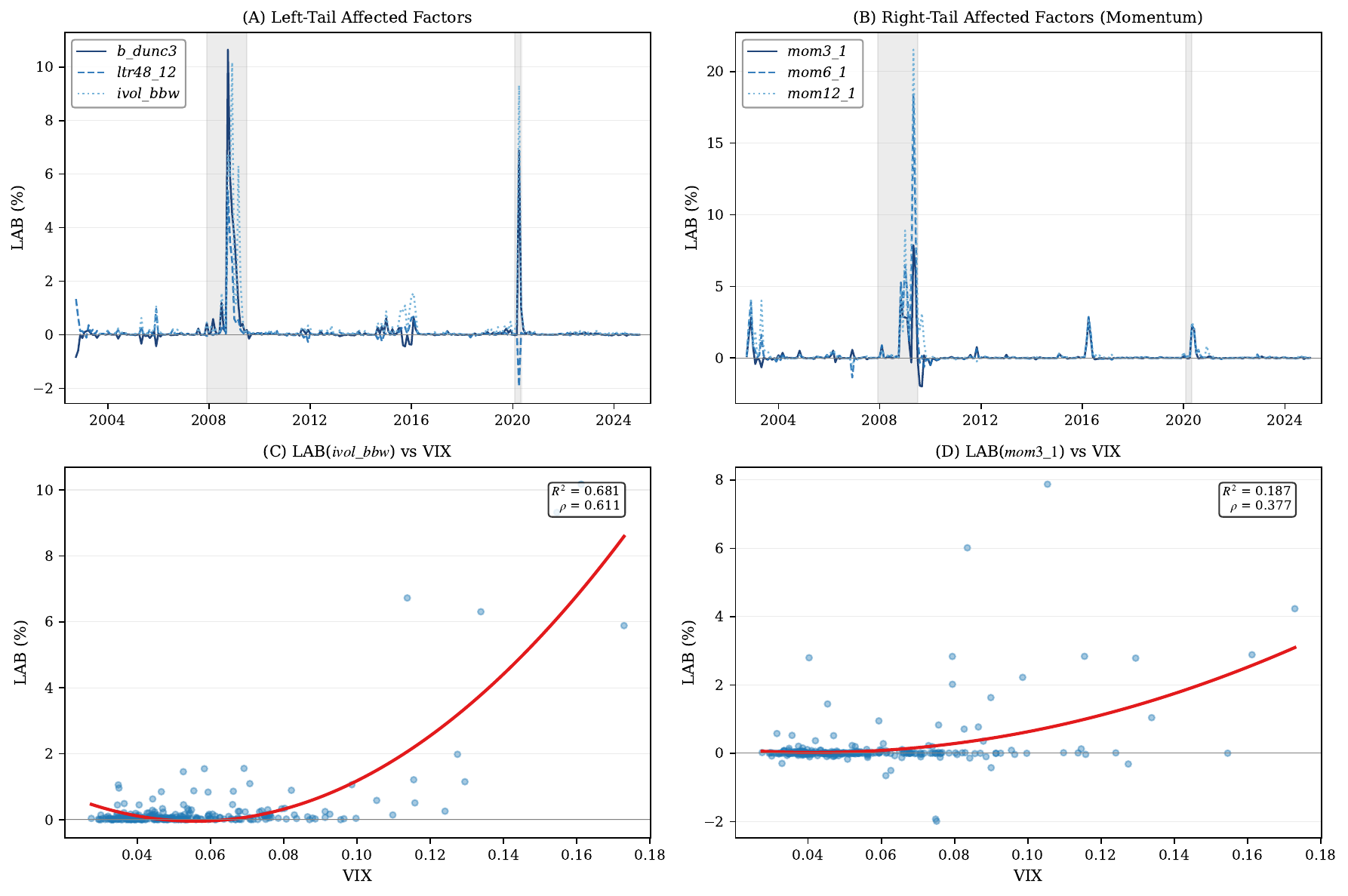}
\caption{Time variation in look-ahead bias.}
\vspace{-2mm}
\begin{justify}
\begin{spacing}{1}
\footnotesize{
The figure plots the time series of look-ahead bias and its relationship with market volatility for selected factors. Panels~A--B plot monthly LAB for left-tail factors (\texttt{b\_dunc3}, \texttt{ltr48\_12}, \texttt{ivol\_bbw}) and right-tail momentum factors (\texttt{mom3\_1}, \texttt{mom6\_1}, \texttt{mom12\_1}). The LAB for \texttt{b\_dunc3} and \texttt{ltr48\_12} is sign-corrected. Panels~C--D display scatter plots of LAB versus the VIX level with quadratic fit. Shaded regions in Panels~A--B indicate NBER recession periods (Great Recession: 2007:12--2009:06; COVID-19: 2020:02--2020:04). Sample: 2002-08 to 2024-12.}
\end{spacing}
\end{justify}
\vspace{-12mm}
\label{fig:lab_figure_1}
\end{figure}

Fig.~\ref{fig:lab_figure_2} shows cumulative returns for selected factors whose measured premia are sensitive to ex-post winsorization. The solid line plots the infeasible (winsorized) factor; the dashed line plots the feasible (unwinsorized) factor; the dotted line plots the cumulative LAB. For long-term reversal (\texttt{ltr48\_12}), idiosyncratic volatility (\texttt{ivol\_bbw}), and macroeconomic uncertainty (\texttt{b\_dunc3}), the cumulative LAB rises sharply during the Great Recession and COVID-19, precisely when the infeasible factor appears to outperform. The gap between infeasible and feasible cumulative returns widens during crises, creating the illusion that these factors provide hedging benefits. For momentum (\texttt{mom3\_1}), the same pattern emerges: the infeasible factor diverges from the feasible factor during stress episodes. In all cases, the ``apparent outperformance'' of the winsorized factor is attributable to the cumulative bias.

\begin{figure}[h!]
\centering
\includegraphics[width=\textwidth,height=0.85\textheight,keepaspectratio]{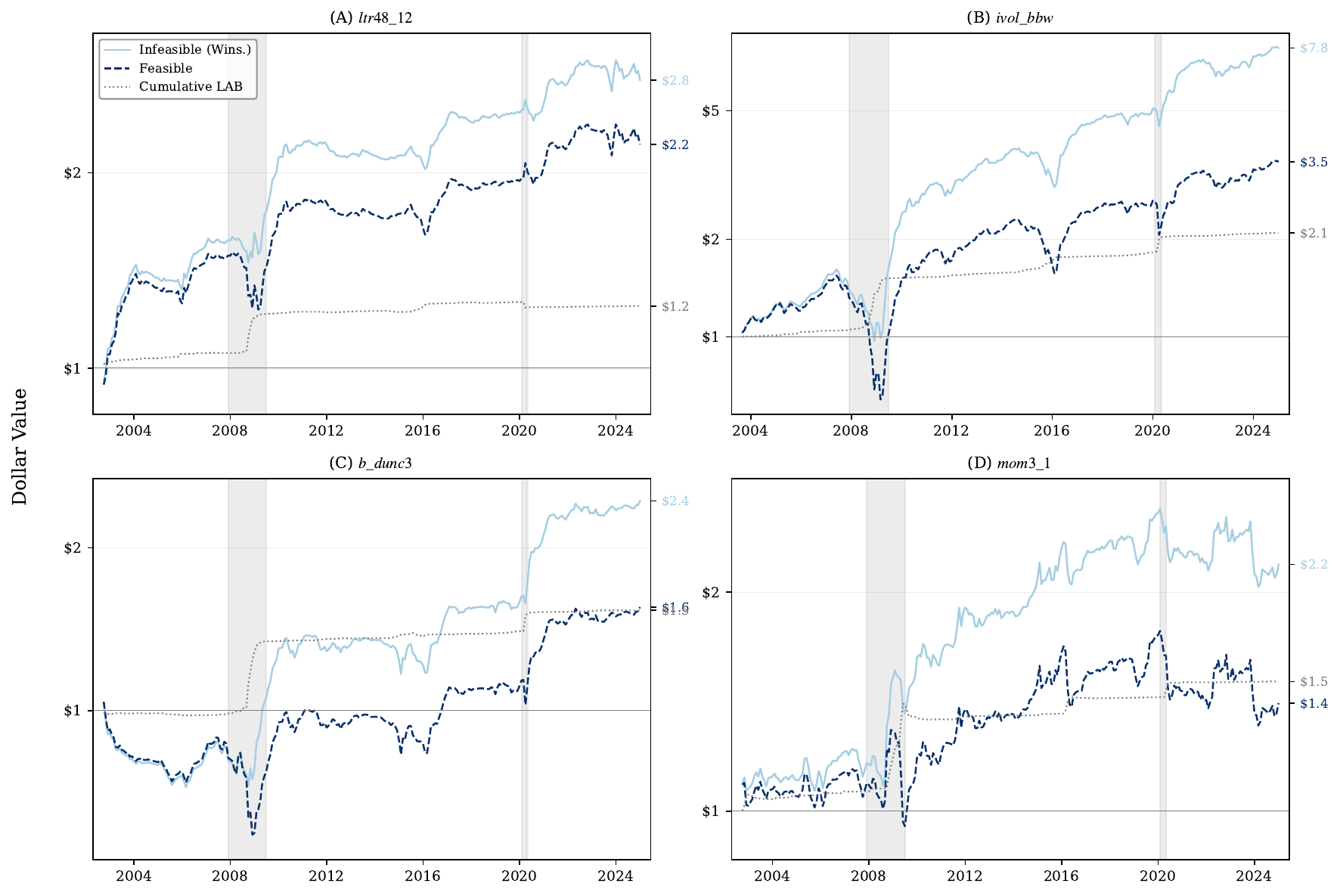}
\caption{Cumulative factor returns with and without ex-post winsorization.}
\vspace{-2mm}
\begin{justify}
\begin{spacing}{1}
\footnotesize{
The figure shows the growth of \$1 invested in four factors sensitive to ex-post return winsorization, under infeasible (with ex-post winsorization, solid light blue) and feasible (without winsorization, dashed dark blue) approaches. The dotted gray line tracks cumulative look-ahead bias, defined as the difference between the infeasible and feasible factor returns. The various panels plot \texttt{ltr48\_12} (long-term reversal), \texttt{ivol\_bbw} (idiosyncratic volatility), \texttt{b\_dunc3} (macro uncertainty beta), and \texttt{mom3\_1} (3-month momentum). Returns for \texttt{ltr48\_12} and \texttt{b\_dunc3} are sign-corrected. Shaded regions indicate NBER recession periods. The $y$-axis uses log scale. Sample: 2002-08 to 2024-12.}
\end{spacing}
\end{justify}
\vspace{-12mm}
\label{fig:lab_figure_2}
\end{figure}

Correcting for measurement error and look-ahead biases eliminates the two largest sources of inflated factor premia and alphas documented in the corporate bond literature. But estimated premia remain sensitive to data-filtering and portfolio-construction choices that are neither standardized nor typically reported. Section~\ref{sec:uncertainty} quantifies this residual fragility.

\section{Data and methodological uncertainty}\label{sec:uncertainty}
Even after correcting for measurement error (LIB) and ex-post filtering (LAB), corporate bond factor premia remain fragile. Variation from data-processing and portfolio-construction choices often exceeds both the point estimate and its sampling uncertainty. Sensitivity to specification choices is a well-documented concern in econometrics \citep*{leamer1983lets}, and recent work documents analogous fragility in equity factor research: \citet*{hou2020replicating} show that many equity factors fail to replicate under alternative data filters, and \citet*{jensen2023there} find that replicability depends on the specific methodological path. The corporate bond setting amplifies these concerns because wide bid-ask spreads, infrequent trading, and a short TRACE sample (2002--2024) magnify the influence of each filtering and aggregation decision.
Uncertainty in estimated factor premia arises at multiple stages of the empirical asset pricing pipeline. We distinguish between two sources of this variation. \emph{Data uncertainty} arises at portfolio formation and concerns which bond--month observations enter the admissible investment universe. Filtering choices such as return trimming thresholds, price range restrictions, and bid-ask bounce exclusions determine which bonds are available to trade at portfolio formation. \emph{Methodological uncertainty} arises after the investment universe has been fixed and concerns how bonds are allocated into portfolios to form factors. Portfolio granularity, breakpoint definitions, weighting schemes, and subsample restrictions all affect measured premia even when the underlying data are identical.
To quantify uncertainty arising from discretionary research choices, we follow \citet{Menkveld_etal_2024} and \citet*{Walter_etal_2024} and measure \emph{non-standard errors} (NSE).
Non-standard errors capture the variation in estimated factor premia attributable to alternative research and data-processing choices rather than sampling variation.
While standard errors quantify uncertainty from finite samples, non-standard errors reflect uncertainty stemming from methodological decisions that are typically unreported. We define the non-standard error for factor $v$ as the interquartile range of estimated premia across all estimation paths:
\begin{equation}
\text{NSE}^v = Q_{0.75}(\hat{\mu}^v) - Q_{0.25}(\hat{\mu}^v),
\end{equation}
where $Q_\alpha(\hat{\mu}^v)$ denotes the $\alpha$-quantile of the distribution of estimated premia $\hat{\mu}^v$ across paths. The ratio $\text{NSE}/\text{SE}$ measures the relative magnitude of research-design uncertainty compared to sampling uncertainty. A ratio exceeding unity indicates that researcher degrees of freedom contribute more to result variation than statistical noise. \citet*{soebhag2024non} raise a similar concern for equity portfolio sorts, showing that seemingly innocuous aggregation choices alter inference about factor premia. In our setting, the average ratio reaches 1.15 for data uncertainty and 1.45 for methodological uncertainty, meaning two researchers applying different filters or portfolio constructions to the same corporate bond data can reach potentially opposite conclusions about whether a given factor premium is statistically significant.
We apply this framework separately to data uncertainty and methodological uncertainty, reflecting their distinct stages in the research pipeline.

\subsection{Data uncertainty}
We study data uncertainty by varying ex-ante filtering rules that define the admissible corporate bond investment universe while holding portfolio construction fixed, with all factors using decile sorts, long the top decile (P10) and short the bottom decile (P1). All filters use only information available at portfolio formation month $t$, ensuring no look-ahead bias (see Fig.~\ref{fig:return_timeline} for the timeline). Filtering is applied independently each month, so a bond excluded from portfolio formation at month $t$ (because its characteristics violate the threshold computed from data through month $t$) is not held from $t$ to $t+1$, but may re-enter the portfolio in any future month $t+k$ if it satisfies the threshold at that date. For each cross-sectional sorting variable, we evaluate 108 alternative filter configurations spanning three categories:
\emph{Return Trimming (48 configurations).}
We exclude observations with monthly returns exceeding threshold $\tau,$ where $\tau \in \{0.20, 0.25, 0.30,\ldots, 0.95\}$ (16 breakpoints). Filters are applied to the left tail only (exclude $r < -\tau$), right tail only (exclude $r > \tau$), or both tails symmetrically, yielding $16 \times 3 = 48$ configurations \citep*[see, e.g.,][]{jostova2013momentum}.
\emph{Price Range Filters (30 configurations).}
We restrict the sample to bonds trading within price bounds expressed as a percentage of par value (\$1,000). The lower bound $P_L \in \{2, 4, 6,\ldots, 20\}$\% (i.e., \$20 to \$200) excludes deep-discount bonds, while the upper bound $P_U \in \{150, 165, 180,\ldots, 285\}$\% (i.e., \$1,500 to \$2,850) excludes extreme premium bonds. Each bound comprises 10 breakpoints. Filters are applied as lower bound only (10 configurations), upper bound only (10 configurations), or combined with matched pairs $[P_L, P_U] \in \{[2, 285], [4, 270], \ldots, [20, 150]\}$ (10 configurations), yielding 30 configurations total \citep*[see, e.g.,][]{BaiBaliWen_2019, bai2021there, bali2021long, bali2023macroeconomic, ChungWangWu_2019}.
\emph{Bid-Ask Bounce Filters (30 configurations).}
We exclude observations where the absolute price change from the previous trade exceeds $\phi,$ where $\phi \in \{0.01, 0.02,\ldots, 0.10\}$ (10 breakpoints). Filters are applied to negative changes only (left tail), positive changes only (right tail), or both directions symmetrically, yielding $10 \times 3 = 30$ configurations \citep*[see, e.g.,][]{bessembinder2008measuring, chordia2017capital}.
Combining the three filter categories yields $48 + 30 + 30 = 108$ filter paths per signal. With equal-weighted (EW) and value-weighted (VW) portfolios across the full sample, investment grade (IG), and non-investment grade (NIG) subsamples, we obtain $108 \times 2 \times 3 = 648$ alternative estimation paths per factor, or 69,984 paths across all 108 signals. For each path, we compute (i) the average long-short return (premium), (ii) its Newey-West $t$-statistic with lags $\lfloor T^{0.25} \rfloor$, and (iii) the CAPMB alpha from regressing factor returns on the bond market factor.

Table~\ref{tab:nse_by_cluster} summarizes non-standard errors and the NSE/SE ratio by factor cluster across all 69,984 filter paths. The average non-standard error for premia is 0.35\% per month, compared to an average premium of 0.33\% per month. The interquartile range of filter-induced variation exceeds the point estimate. For CAPMB alphas, the average NSE of 0.27\% per month is nearly twice as large as the average alpha of 0.15\% per month.

\begin{table}[!ht]
\caption{Non-standard errors and data uncertainty by factor cluster.}
\begin{spacing}{1}
{\footnotesize
The table reports non-standard errors from data-filtering variation across 69,984 filter paths (648 per factor). For each cluster, we report the mean ($\bar{\mu}$, $\bar{\alpha}$), median (Med), and interquartile range (NSE $= Q_{0.75} - Q_{0.25}$) of factor premia and alphas across all filter paths. Ratio measures methodological uncertainty relative to sampling uncertainty, defined as $\text{Ratio} = \text{std}(r^v_p) / \overline{\sigma^v_p}$, where $\text{std}(r^v_p)$ is the standard deviation of the estimated premium or alpha across paths and $\overline{\sigma^v_p}$ is the corresponding average Newey-West standard error. $N_{\text{paths}}$ is the number of filter configurations analyzed. Values in percent per month. Sample: 2002-09 to 2024-12, $T$=268.}
\end{spacing}
\vspace{-4mm}
\begin{center}
\label{tab:nse_by_cluster}
\vspace{2mm}
\scalebox{0.85}{%
\begin{tabular}{l cccc cccc r}
\toprule
 & \multicolumn{4}{c}{Premium ($\mu$)} & \multicolumn{4}{c}{Alpha ($\alpha$)} & \\
\cmidrule(lr){2-5} \cmidrule(lr){6-9}
Cluster & $\bar{\mu}$ & Med & NSE & Ratio & $\bar{\alpha}$ & Med & NSE & Ratio & $N_{\text{paths}}$ \\
\midrule
Spreads, Yields, Size & 0.45 & 0.32 & 0.48 & 1.43 & 0.16 & 0.14 & 0.31 & 1.34 & 5,832 \\
Value & 0.43 & 0.38 & 0.30 & 1.43 & 0.27 & 0.22 & 0.22 & 0.97 & 3,240 \\
Momentum \& Reversal & 0.24 & 0.17 & 0.25 & 0.93 & 0.16 & 0.15 & 0.35 & 1.12 & 13,608 \\
Illiquidity & 0.34 & 0.25 & 0.34 & 1.34 & 0.12 & 0.10 & 0.16 & 0.91 & 8,424 \\
Volatility \& Risk & 0.44 & 0.39 & 0.49 & 1.36 & 0.09 & 0.04 & 0.19 & 0.85 & 10,368 \\
Market Risk & 0.30 & 0.32 & 0.35 & 0.80 & 0.09 & 0.10 & 0.37 & 1.01 & 5,832 \\
Credit \& Default Risk & 0.56 & 0.53 & 0.42 & 0.79 & 0.16 & 0.15 & 0.20 & 0.65 & 2,592 \\
Volatility \& Liquidity Risk & 0.21 & 0.18 & 0.26 & 0.87 & 0.15 & 0.11 & 0.34 & 1.03 & 8,424 \\
Macro \& Other Risk & 0.29 & 0.25 & 0.27 & 0.85 & 0.19 & 0.19 & 0.18 & 0.75 & 11,664 \\
\midrule
All & 0.33 & 0.26 & 0.35 & 1.15 & 0.15 & 0.13 & 0.27 & 1.01 & 69,984 \\
\bottomrule
\end{tabular}
}
\end{center}
\end{table}

The NSE varies across clusters but exceeds the average premium for every cluster except Value, Credit \& Default Risk, and Macro \& Other Risk. Spreads, Yields, Size (NSE 0.48\%, mean premium 0.45\%) and Volatility \& Risk (NSE 0.49\%, mean premium 0.44\%) show the widest filter-induced dispersion. The NSE/SE ratio exceeds unity for four of nine clusters, with Spreads, Yields, Size (1.43), Value (1.43), and Volatility \& Risk (1.36) showing the largest ratios.
In the Internet Appendix, we analyze  the filter types that generate significant alphas where the baseline does not. A filter path is classified as an improvement if the associated CAPMB alpha $t$-statistic exceeds 1.96, exceeds the baseline, and produces a larger alpha. Most clusters show improvement rates of 0--3\%: Credit \& Default Risk shows 0\% across all configurations, Volatility \& Risk 0--2\%, Market Risk 0--3\%, and Macro \& Other Risk 0--1\%. The Value cluster is the exception, with 21\% of left-tail trimming paths and 15\% of symmetric trimming paths producing significant improvements.
The Internet Appendix displays premium $t$-statistics for the top four factors within each cluster across all 648 filter paths.
We do not take a stand on which filters are economically justified, nor do we test every combination, as applying filters jointly across categories would produce approximately 4.2 million feasible configurations per signal. Even when each filter is tested individually, the NSE exceeds the average premium (0.35\% vs 0.33\% per month) and the average alpha (0.27\% vs 0.15\%), and most filter paths reduce rather than increase measured premia. Applying ad hoc data filters to transaction-level-cleaned data is more likely to destroy factor premia and alphas than to reveal them. The few configurations that do produce significant alphas are concentrated in a single cluster (Value) and a single filter type (left-tail return exclusions at portfolio formation). Value is also the only cluster that produces significant premia under the baseline (unfiltered) specification, so the improvement reflects amplification of an already detectable signal.

\subsection{Methodological uncertainty}
While data uncertainty originates from filtering bonds that enter the investable set, methodological uncertainty stems from how those bonds are sorted into portfolios to form factors. Five dimensions span the design space: portfolio granularity (tercile, quintile, decile), breakpoint universe (all bonds, investment grade only, large bonds only), weighting scheme (EW, VW), rating subsample (all, IG, NIG), and maturity bucket (all, short, intermediate, long).\footnote{Breakpoint variation changes a bond's relative rank given a reference distribution, while rating subsampling changes the set of bonds included in the sort. For example, ``investment grade only'' breakpoints compute quantile boundaries from the cross-section of a signal within the IG universe, then allocate all bonds (including high-yield) to portfolios based on those boundaries. The equity analogue is Fama and French's use of NYSE-only size breakpoints to sort all stocks.} Researchers routinely treat these dimensions as secondary design choices \citep*[see also][]{Walter_etal_2024}. Yet as \citet*{linnainmaa2018history} show for equity factors, the mapping from a signal to a portfolio payoff can be as consequential as the signal.
Combining these dimensions yields $2 \times 3 \times 3 \times 3 \times 4 = 216$ candidate portfolio constructions per signal. Not all are admissible. Breakpoints computed on IG bonds cannot sort NIG bonds (the universes do not overlap). This eliminates 24 specifications per signal. When portfolios are formed, e.g., exclusively within IG bonds, IG-based breakpoints coincide with full-universe breakpoints. This eliminates a further 24 redundant specifications. The effective grid contains 168 economically distinct portfolio constructions per signal, yielding $168 \times 108 = 18{,}144$ factor return series. Sixteen specifications (0.088\%) produce months with empty long or short legs, all in decile sorts within the high-yield, long-maturity segment.\footnote{The Internet Appendix confirms that even the most restrictive configurations contain a sufficient number of bonds in the long and short portfolios.} Excluding these leaves 18,128 well-defined factor return series.\footnote{All the main results are robust to excluding decile sorts from the specification grid.}
Table~\ref{tab:mu_nse_by_cluster} reports non-standard errors across these 168 portfolio construction paths. The average NSE for premia is 0.31\% per month, compared to an average premium of 0.22\%, and the interquartile range (IQR) of methodology-induced variation exceeds the point estimate. For CAPMB alphas, the average NSE of 0.20\% per month exceeds the average alpha of 0.11\%. The premium-related NSE/SE ratio averages 1.45 and exceeds unity for all nine factor clusters, ranging from 1.02 (Market Risk) to 2.00 (Spreads, Yields, Size).
\begin{table}[!ht]
\caption{Non-standard errors and methodological uncertainty by factor cluster.}
\label{tab:mu_nse_by_cluster}
\begin{spacing}{1}
{\footnotesize
The table reports non-standard errors from portfolio-construction variation across 18,128 non-degenerate methodology paths (up to 168 per factor). For each cluster, we report the mean ($\bar{\mu}$, $\bar{\alpha}$), median (Med), and interquartile range (NSE $= Q_{0.75} - Q_{0.25}$) of factor premia and alphas across all methodology paths. Ratio measures methodological uncertainty relative to sampling uncertainty, defined as $\text{Ratio} = \text{std}(r^v_p) / \overline{\sigma^v_p}$, where $\text{std}(r^v_p)$ is the standard deviation of the estimated premium or alpha across paths and $\overline{\sigma^v_p}$ is the corresponding average Newey-West standard error. $N_{\text{paths}}$ is the number of methodological configurations analyzed. Values in percent per month. Sample: 2002-09 to 2024-12, $T$=268.}
\end{spacing}
\vspace{-4mm}
\begin{center}

\vspace{2mm}
\scalebox{0.85}{%
\begin{tabular}{l cccc cccc r}
\toprule
 & \multicolumn{4}{c}{Premium ($\mu$)} & \multicolumn{4}{c}{Alpha ($\alpha$)} & \\
\cmidrule(lr){2-5} \cmidrule(lr){6-9}
Cluster & $\bar{\mu}$ & Med & NSE & Ratio & $\bar{\alpha}$ & Med & NSE & Ratio & $N_{\text{paths}}$ \\
\midrule
Spreads, Yields, Size & 0.25 & 0.22 & 0.46 & 2.00 & 0.09 & 0.09 & 0.34 & 1.57 & 1,508 \\
Value & 0.29 & 0.25 & 0.27 & 1.62 & 0.19 & 0.17 & 0.20 & 1.13 & 838 \\
Momentum \& Reversal & 0.13 & 0.09 & 0.24 & 1.28 & 0.08 & 0.06 & 0.23 & 1.21 & 3,520 \\
Illiquidity & 0.22 & 0.16 & 0.27 & 1.53 & 0.10 & 0.08 & 0.14 & 1.02 & 2,184 \\
Volatility \& Risk & 0.32 & 0.26 & 0.43 & 1.66 & 0.10 & 0.07 & 0.19 & 1.07 & 2,688 \\
Market Risk & 0.21 & 0.20 & 0.27 & 1.02 & 0.09 & 0.08 & 0.24 & 1.15 & 1,512 \\
Credit \& Default Risk & 0.43 & 0.38 & 0.35 & 1.11 & 0.16 & 0.13 & 0.20 & 0.79 & 672 \\
Volatility \& Liquidity Risk & 0.15 & 0.12 & 0.20 & 1.10 & 0.11 & 0.09 & 0.23 & 1.10 & 2,184 \\
Macro \& Other Risk & 0.21 & 0.19 & 0.26 & 1.12 & 0.14 & 0.13 & 0.16 & 0.89 & 3,022 \\
\midrule
All & 0.22 & 0.17 & 0.31 & 1.45 & 0.11 & 0.09 & 0.20 & 1.14 & 18,128 \\
\bottomrule
\end{tabular}
}
\end{center}
\end{table}

The Internet Appendix displays CAPMB alpha $t$-statistics for the top four factors within each cluster across all 168 portfolio constructions. For many signals, the median $t$-statistic lies below 1.96, and the interquartile range spans both significant and insignificant values.
Across these 36 factors (four top factors $\times$ nine clusters), 12 produce both positive and negative long--short premia under alternative portfolio constructions, and 29 produce both positive and negative CAPMB alphas. The Internet Appendix decomposes improvement rates across portfolio construction dimensions. Out of all methodological variations, 94\% of them do not improve the CAPMB alpha relative to the baseline. Value is the largest improver (16\% overall, 36\% in long-maturity constructions), and long maturity is the primary driver across clusters (9.5\% vs 1.9\% for intermediate maturity). As with data uncertainty, Value is the only cluster with a detectable baseline premium.
The Internet Appendix displays premium $t$-statistics across all constructions for the top four factors per cluster.
Given that 94\% of the specifications fail to improve the CAPMB alpha relative to the baseline, factor construction choices are better guided by economic reasoning rather than by searching across the specification space.

\section{Conclusion}\label{sec:conclusion}

Corporate bond factor research faces a replication crisis. Across a `factor zoo' of 108 signals spanning nine thematic clusters, most previously documented factor alphas reflect two biases rather than genuine risk-adjusted performance. Latent Implementation Bias (LIB) inflates measured factor premia for price-based factors (yields, credit spreads, value, and short-term reversal) by up to 91\%. LIB arises from two sources: the same noisy transaction price enters both the sorting signal and the return denominator, creating a correlated 
errors-in-variables (CEIV) problem, and the observed transaction price is not executable in over-the-counter (OTC) markets. Look-Ahead Bias (LAB) inflates measured factor premia for momentum, long-term reversal, idiosyncratic volatility, downside risk, and macroeconomic exposure factors by 57--100\% through ex-post return filtering that embeds future information into portfolio construction. A third source of fragility persists even after both biases are corrected. Across 648 data-filtering configurations using only ex-ante information, non-standard errors average 0.35\% per month, exceeding the average premium of 0.33\%. Across 168 methodological configurations holding the investment universe fixed, the NSE/SE ratio averages 1.45 and exceeds unity for all nine factor clusters. Twelve of 36 top factors flip sign depending on the portfolio construction chosen. As emphasized in the Internet Appendix, across 432 factor-specification combinations, only 26 (6.0\%) bond CAPM alphas survive after false discovery rate (FDR) correction. The survivors concentrate among credit-spread-based value factors. Applying FDR to mean return $p$-values produces consistent results, with 119 of 432 nominally significant but only 22 surviving correction. Across five illiquidity factors, no CAPMB alpha is distinguishable from zero.

Our contribution extends beyond identifying these replication failures. We provide open source data infrastructure, software, and a reproducible framework for corporate bond asset pricing research. The companion website \href{https://openbondassetpricing.com/}{Open Bond Asset Pricing} hosts daily bond data constructed through a two-stage TRACE cleaning pipeline, pre-formed factors for all 108 signals, and data reports documenting every cleaning step. The data are also available through \href{https://wrds-www.wharton.upenn.edu/pages/get-data/contributed-data-forms/corporate-bond-data-dickerson-monthly/}{WRDS Contributed Data}. The companion software \href{https://pypi.org/project/PyBondLab/}{\texttt{PyBondLab}} provides the infrastructure for reproducible factor construction, with built-in signal gap procedures, ex-ante filtering, and flexible portfolio configurations. The replication code is publicly available on \href{https://github.com/Alexander-M-Dickerson/trace-data-pipeline}{GitHub}.

These resources define a path forward. Credible corporate bond factor research requires three elements: breaking the shared-price link through gap procedures that ensure the signal and the return denominator use different prices, filtering only with information available at portfolio formation, and reporting economically motivated factor specifications that limit the effects of data mining. Finally, given the nontrivial trade execution costs in OTC corporate bond markets, a natural extension is to evaluate whether the surviving factors remain profitable after accounting for corporate bond trade execution costs.

\clearpage
{\footnotesize
\singlespacing 
\setlength{\bibsep}{1pt}

}
\pagebreak


\clearpage
\newpage
\appendix

\begin{center}
 {\LARGE \textbf{Appendix}}   
\end{center}

\renewcommand{\theequation}{A.\arabic{equation}}%
\renewcommand{\thefigure}{A.\arabic{figure}} \setcounter{figure}{0}
\renewcommand{\thetable}{A.\arabic{table}} \setcounter{table}{0}
\renewcommand{\theHfigure}{A.\arabic{figure}}
\renewcommand{\theHtable}{A.\arabic{table}}
\renewcommand{\theHequation}{A.\arabic{equation}}

\vspace{-.25cm}
\section{Data filters and monthly returns} \label{app:filters}

This appendix provides detailed specifications for the transaction-level and daily-level filters described in Section~\ref{sec:osbap}, followed by the formal definitions of monthly returns and the treatment of defaulted bonds. All thresholds are expressed in percentage-of-par units (\$1{,}000 face value), so a threshold of 35 corresponds to \$350. Section~\ref{sec:uncertainty} examines the sensitivity of our results to these parameter choices across 648 filter configurations. All filters are implemented in Python and available at \href{https://github.com/Alexander-M-Dickerson/trace-data-pipeline}{\nolinkurl{github.com/Alexander-M-Dickerson/trace-data-pipeline}}.

\subsection{Decimal shift corrector}

Corporate bond prices in TRACE are expressed as a percentage of par value, where par equals \$1,000 for all bonds in our sample. A price of 100 represents \$1,000 (100\% of par); a price of 98.5 represents \$985. Prices occasionally exhibit multiplicative errors arising from incorrect decimal placement during data entry (for example, a price of 98.50 recorded as 985.0 or 9.85). The decimal shift corrector identifies and corrects such errors while preserving legitimate price movements. Further details and worked examples are available at \href{https://github.com/Alexander-M-Dickerson/trace-data-pipeline/blob/main/stage0/README_decimal_shift_corrector.md}{\nolinkurl{github.com/.../README_decimal_shift_corrector.md}}.

\paragraph{Anchor construction.}
For each transaction $i$ with observed price $P_i$, we construct a \emph{rolling unique-median anchor} $A_i$: the median of distinct prices within a centered window of $2w+1$ observations (default $w=5$, yielding 11 transactions). When centered windows are unavailable near sample boundaries, the filter falls back to forward-looking or backward-looking windows, and if both are unavailable, to the median of the bond's full price series.

\paragraph{Error measurement.}
The \emph{raw relative error} measures how far the observed price deviates from the anchor:
\begin{equation}
\varepsilon_i^{\text{raw}} = \frac{|P_i - A_i|}{A_i}.
\end{equation}
A transaction qualifies for correction testing only if this raw error exceeds the gate threshold $\tau_{\text{bad}} = 0.05$ (5\%), establishing that a meaningful discrepancy exists.

\paragraph{Candidate correction.}
We test four multiplicative correction factors: $\mathcal{F} = \{0.1, 0.01, 10, 100\}$, corresponding to common decimal-shift scenarios. For each candidate factor $f \in \mathcal{F}$, the corrected price is $P_i^{(f)} = f \cdot P_i$, and the \emph{corrected relative error} is
\begin{equation}
\varepsilon_i^{(f)} = \frac{|P_i^{(f)} - A_i|}{A_i}.
\end{equation}

\paragraph{Acceptance gates.}
A correction is accepted only if all four conditions hold:
\begin{enumerate}
\item \textbf{Raw error gate:} $\varepsilon_i^{\text{raw}} > \tau_{\text{bad}} = 0.05$ (5\%), confirming a meaningful discrepancy.
\item \textbf{Alignment gate:} At least one of: (a)~$\varepsilon_i^{(f)} \leq \tau_{\text{pct}} = 0.02$ (2\%); (b)~the absolute error $|P_i^{(f)} - A_i| \leq \tau_{\text{abs}} = 8.0$ price points; or (c)~both $A_i$ and $P_i^{(f)}$ fall within $\pm 15$ points of par (par-proximity rule).
\item \textbf{Improvement requirement:} $\varepsilon_i^{(f)} \leq \gamma \cdot \varepsilon_i^{\text{raw}}$, where $\gamma = 0.20$; the corrected error must be at most 20\% of the raw error.
\item \textbf{Plausibility check:} The corrected price lies within $[5, 300]$ (5\% to 300\% of par).
\end{enumerate}
When multiple factors satisfy all gates, we select the factor minimizing $\varepsilon_i^{(f)}$.

Table~\ref{tab:decimal_params} summarizes the default parameter values.

\begin{table}[h!]
\centering
\caption{Decimal shift corrector parameters.}
\label{tab:decimal_params}
\begin{tabular}{llp{7cm}}
\toprule
Parameter & Default & Description \\
\midrule
$w$ & 5 & Half-window size (effective window = $2w+1 = 11$) \\
$\tau_{\text{pct}}$ & 0.02 & Relative error acceptance threshold \\
$\tau_{\text{abs}}$ & 8.0 & Absolute error acceptance threshold (price points) \\
$\tau_{\text{bad}}$ & 0.05 & Minimum raw error to consider correction \\
$\gamma$ & 0.20 & Required proportional improvement \\
$[\underline{P}, \bar{P}]$ & $[5, 300]$ & Plausible corrected price range \\
par-proximity band & 15.0 & Par-proximity tolerance (points from par) \\
\bottomrule
\end{tabular}
\end{table}

\subsection{Bounce-back filter}

The bounce-back filter detects transient price spikes (large deviations from baseline that quickly revert) which typically indicate data entry errors rather than genuine market movements. Further details and worked examples are available at \href{https://github.com/Alexander-M-Dickerson/trace-data-pipeline/blob/main/stage0/README_bounce_back_filter.md}{\nolinkurl{github.com/.../README_bounce_back_filter.md}}.
The detection threshold is $\tau = 35$ price points (35\% of par, or \$350 on a \$1{,}000 face-value bond).

\paragraph{Anchor construction.}
For each transaction $i$ with observed price $P_i$, we compute a trailing baseline $B_i$ as the median of unique prices from the prior $w$ trades (default $w=5$):
\begin{equation}
B_i = \text{median}\bigl(\text{unique}(\{P_{i-w}, \ldots, P_{i-1}\})\bigr).
\end{equation}
Taking unique prices removes repeated-print bias from consecutive identical trades.

\paragraph{Candidate detection.}
An observation becomes a bounce-back candidate if any of the following three conditions holds:
\begin{enumerate}
\item \textbf{Large one-step jump:} $|\Delta P_i| \geq \tau - \delta$, where $\Delta P_i = P_i - P_{i-1}$ and $\delta = 1$ is a small slack tolerance.
\item \textbf{Large displacement from baseline:} $|P_i - B_i| \geq \tau - \delta$.
\item \textbf{Par-spike heuristic:} $P_i = 100$ (par) and $|P_i - B_i| \geq \alpha \cdot \tau$, where $\alpha = 0.25$.
\end{enumerate}

\paragraph{Reversion detection.}
Once a candidate is identified, the filter scans forward up to $L=5$ rows seeking evidence of reversion via either:
\begin{itemize}
\item \textbf{Path A (opposite move):} A subsequent price change of opposite sign with magnitude $\geq \tau - \delta$.
\item \textbf{Path B (return to anchor):} A subsequent price within $\alpha \cdot \tau$ of the original baseline $B_i$.
\end{itemize}

\paragraph{Flagging logic.}
Upon detecting reversion, the filter: (i)~reassigns the flag to the preceding row if it shows greater displacement from baseline ($\geq 5$ points); (ii)~extends flags forward within a five-row span while prices remain displaced; and (iii)~flags persistent par blocks (runs of $\ell_{\min} \geq 3$ consecutive trades at par) with a two-row cooldown to suppress cascading flags.

Table~\ref{tab:bounce_params} summarizes the default parameter values.

\begin{table}[h!]
\centering
\caption{Bounce-back filter parameters.}
\label{tab:bounce_params}
\begin{tabular}{llp{7cm}}
\toprule
Parameter & Default & Description \\
\midrule
$\tau$ & 35.0 & Minimum price change to trigger detection (\$350) \\
$L$ & 5 & Forward rows scanned for reversion \\
$w$ & 5 & Trailing window for anchor calculation \\
$\alpha$ & 0.25 & Tolerance fraction for return-to-anchor criterion \\
$\ell_{\min}$ & 3 & Minimum consecutive par trades to flag \\
cooldown & 2 & Rows skipped after flagging par blocks \\
\bottomrule
\end{tabular}
\end{table}

\subsection{Distressed bond filters}

After aggregating transactions to daily frequency, a second filtering stage targets anomalies in distressed bond prices. Four sub-filters address distinct error patterns, each using lookback and lookforward windows of $L=5$ days and ratio-based thresholds that adapt to varying price scales. All prices are expressed as a percentage of par (\$1,000 face value). Further details and worked examples are available at \href{https://github.com/Alexander-M-Dickerson/trace-data-pipeline/blob/main/stage1/README_distressed_filter.md}{\nolinkurl{github.com/.../README_distressed_filter.md}}.

\paragraph{Filter 1: Anomaly detection (downward outliers).}
This filter identifies isolated ultra-low prices far below surrounding observations. A price $P_t$ on day $t$ becomes a candidate if either: (i) $P_t < \tau_{\text{low}} = 0.10$ (i.e., below 0.10\% of par, or \$1), or (ii) $P_t$ matches a suspicious round number in $\{0.001, 0.01, 0.05, 0.10, 0.25, 0.50, 1.00\}$. To assess whether the candidate is anomalous, we collect all prices from the surrounding $\pm L$-day window that exceed $P_t$, and compute their median $M_{\text{surr}}$. If the surrounding prices are much higher than the candidate (specifically, if $M_{\text{surr}} / P_t \geq \rho_{\text{anomaly}} = 3.0$), then the candidate is flagged as a likely data error. 
\paragraph{Filter 2: Spike detection (upward outliers).}
This filter detects temporary upward price spikes that quickly revert. A price $P_t$ becomes a candidate if either: (i) $P_t > \tau_{\text{high}} = 5.0$ (i.e., above 5\% of par, or \$50), or (ii) $P_t$ is a round number exceeding 0.50 (\$5). To confirm a spike, we compute the pre-spike median $M_{\text{pre}}$ from the $L$ days preceding day $t$. The spike is flagged only if two conditions hold. First, the spike magnitude satisfies $P_t / M_{\text{pre}} \geq \rho_{\text{spike}} = 3.0$ (the price is at least $3\times$ the pre-spike level). Second, recovery occurs within the subsequent $L$ days: at least one price falls to $P_j \leq \rho_{\text{recovery}} \cdot M_{\text{pre}}$, where $\rho_{\text{recovery}} = 2.0$.

\paragraph{Filter 3: Plateau detection.}
This filter identifies consecutive days with identical ultra-low or round prices, often indicating stale or placeholder values. A price $P_t$ opens a candidate if either: (i) $P_t < \tau_{\text{plateau}} = 0.15$ (below 0.15\% of par, or \$1.50), or (ii) $P_t$ is a round number. The algorithm identifies runs of $\ell \geq \ell_{\min} = 2$ consecutive days with exactly the same price. Let $P_{\text{pre}}$ denote the price on the day before the plateau begins, and $P_{\text{post}}$ denote the price on the day after the plateau ends. The plateau is flagged if any of three conditions holds: (i) $P_{\text{pre}} / P_t \geq \rho_{\text{plateau}} = 3.0$; (ii) $P_{\text{post}} / P_t \geq \rho_{\text{plateau}}$; or (iii) $P_t$ is a round number. The intuition is that genuine low prices should persist, whereas placeholder values typically appear briefly before reverting.

\paragraph{Filter 4: Intraday inconsistency.}
This filter flags large discrepancies between different intraday price measures. For each day $t$, let $\mathcal{I}_t$ denote the set of available intraday prices (high, low) and let $\bar{P}_t = |\mathcal{I}_t|^{-1} \sum_{p \in \mathcal{I}_t} p$ be their mean. An observation is flagged if both conditions hold: (i) at least one price in $\mathcal{I}_t$ falls below $\tau_{\mathrm{intraday}} = 20.0$ (20\% of par, or \$200), and (ii) the intraday range exceeds $\gamma_{\mathrm{range}} = 0.75$ (75\%) of the mean:
\[
\frac{\max(\mathcal{I}_t) - \min(\mathcal{I}_t)}{\bar{P}_t} > 0.75.
\]
The first condition ensures the filter activates only for days with at least one low price, avoiding false positives on normally-priced bonds with typical intraday volatility.
Table~\ref{tab:distressed_params} summarizes the default parameter values.

\begin{table}[h!]
\centering
\caption{Distressed bond filter parameters.}
\label{tab:distressed_params}
\begin{tabular}{llp{7cm}}
\toprule
Parameter & Default & Description \\
\midrule
$\tau_{\text{low}}$ & 0.10\% & Ultra-low threshold (Filter 1) \\
$\tau_{\text{high}}$ & 5.0\% & High spike threshold (Filter 2) \\
$\tau_{\text{plateau}}$ & 0.15\% & Plateau opening threshold (Filter 3) \\
$\rho_{\text{anomaly}}$ & 3.0 & Ratio for downward outliers (Filter 1) \\
$\rho_{\text{spike}}$ & 3.0 & Spike magnitude ratio (Filter 2) \\
$\rho_{\text{recovery}}$ & 2.0 & Recovery ratio (Filter 2) \\
$\rho_{\text{plateau}}$ & 3.0 & Pre/post displacement ratio (Filter 3) \\
$\gamma_{\mathrm{range}}$ & 0.75 & Intraday range threshold (Filter 4) \\
$\tau_{\mathrm{intraday}}$ & 20.0\% & Price threshold for Filter 4 activation \\
$L$ & 5 & Lookback/lookforward window (days) \\
$\ell_{\min}$ & 2 & Minimum plateau length (days) \\
\bottomrule
\end{tabular}
\end{table}

\subsection{Monthly return computation}\label{app:returns}

We compute three return series for each bond-month, as referenced in Section~\ref{sec:data-monthly}.

\paragraph{Month-end return.}
The month-end return measures total performance from the end of month $t$ to the end of month $t+1$:
\begin{equation}\label{eq:ret-end}
r_{i,t+1}^{\text{End}} = \frac{P_{i,t+1}^{\text{end}} + AI_{i,t+1}^{\text{end}} + C_{i,t+1}}{P_{i,t}^{\text{end}} + AI_{i,t}^{\text{end}}} - 1,
\end{equation}
where $P_{i,t}$ is the volume-weighted clean price, $AI_{i,t}$ is accrued interest, and $C_{i,t+1}$ is the coupon payment (if any) during month $t+1$. A return is valid only if the bond trades within the last five business days of both month $t$ and month $t+1$ (NYSE calendar). This measure is the standard return used for comparability with prior literature.

\paragraph{Month-begin return.}
The month-begin return measures performance within a single month, from the first available trade to the last:
\begin{equation}\label{eq:ret-bgn}
r_{i,t+1}^{\text{Bgn}} = \frac{P_{i,t+1}^{\text{end}} + AI_{i,t+1}^{\text{end}} + C_{i,t+1}}{P_{i,t+1}^{\text{bgn}} + AI_{i,t+1}^{\text{bgn}}} - 1,
\end{equation}
where superscripts denote prices from the last and first five business days of month $t+1$, respectively. This return is valid only if the bond trades in both windows. Because the entry price post-dates signal observation, the month-begin return measures \emph{implementable} performance: it is the return a trader could earn after observing a signal at month-end $t$ and executing at the earliest available price in month $t+1$.

\paragraph{Return decomposition.}
The month-end return decomposes approximately as
\begin{equation}\label{eq:ret-decomp}
r_{i,t+1}^{\text{End}} \approx \underbrace{\frac{P_{i,t+1}^{\text{bgn}}}{P_{i,t}} - 1}_{\text{LIB}_{i,t+1}} \;+\; r_{i,t+1}^{\text{Bgn}},
\end{equation}
where $\text{LIB}_{i,t+1}$ is the clean-price return between the signal observation price and the execution price. Section~\ref{sec:mmn} discusses the economic interpretation of this decomposition.

\paragraph{Excess and duration-adjusted returns.}
We form excess returns by subtracting the Fama-French risk-free rate: $r_{i,t+1}^{x} = r_{i,t+1} - r_{t+1}^{f}$. Duration-adjusted excess returns instead subtract a duration-matched U.S.\ Treasury return:
\begin{equation}\label{eq:ret-dur}
r_{i,t+1}^{x} = r_{i,t+1} - r_{i,t+1}^{\text{Tsy}},
\end{equation}
where $r_{i,t+1}^{\text{Tsy}}$ is obtained by linearly interpolating key-rate U.S.\ Treasury bond returns from WRDS using each bond's modified duration, following \citet*{andreani2023computing}.

\subsection{Default handling}\label{app:defaults}

When a bond enters or trades under default, we assume coupon payments cease and adjust the return formula accordingly. We identify a bond as in default if its S\&P rating equals 22 (D) or its Moody's rating equals 21 (D).

\paragraph{Default-event return.}
When a bond transitions \emph{into} default at time $t+1$ (rated non-default at $t$, default at $t+1$), we set the coupon to zero. The return uses the clean price at default in the numerator and the dirty price before default in the denominator:
\begin{equation}\label{eq:ret-default}
r_{i,t+1}^{\text{def}} = \frac{P_{i,t+1}^{\text{clean}}}{P_{i,t} + AI_{i,t}} - 1.
\end{equation}

\paragraph{Trading-in-default return.}
When a bond remains in default at both $t$ and $t+1$, no coupon accrues. We compute the return on a flat basis using clean prices only:
\begin{equation}\label{eq:ret-flat}
r_{i,t+1}^{\text{flat}} = \frac{P_{i,t+1}^{\text{clean}}}{P_{i,t}^{\text{clean}}} - 1.
\end{equation}
This return is capped at the standard return~\eqref{eq:ret-end} to prevent cases where ignoring accrued interest produces an artificially higher return.

\section{Sorting-induced bias: Formal treatment}\label{app:sorting-bias-formal}
\renewcommand{\thetable}{B.\arabic{table}} \setcounter{table}{0}

A correlated errors-in-variables (CEIV) bias arises in long--short factor returns constructed from price-based signals whenever the observed sorting signal and the observed return load on a common disturbance. In this case, portfolio assignment becomes informative about the return measurement error, and the bias term in the long--short decomposition does not cancel out.

Our general result only requires that price measurement error exhibit cross-sectional variation, that this disturbance enter the observed signal in a non-degenerate way, and that the return measurement error contain a component driven by the same disturbance. In particular, when the signal measurement error is a non-constant (locally monotone) function of a cross-sectionally heterogeneous price disturbance, portfolio sorting induces selection on the disturbance directly. If the return measurement error also depends non-trivially on the same component, assets assigned to the long and short portfolios systematically differ in their expected return measurement error. As a result, even when the true signal has no predictive content for future returns, long--short portfolio returns are biased.

This bias is not driven by misranking per se. While measurement error can lead to misallocation of assets across portfolios, misranking alone does not generate a mechanical premium. The bias arises specifically from the shared component between the signal and the return denominator. Consistent with this distinction, we show that for non-price signals, such as characteristics whose measurement error does not enter the return calculation, measurement error may attenuate true premia or alter portfolio composition, but it does not generate the mechanical CEIV bias identified here.

Having established the existence and sign of the bias under general conditions, we then impose stronger distributional assumptions to characterize its magnitude. Under linear signal contamination and joint normality, we derive a closed-form expression for the bias as a function of portfolio granularity (finer sorts amplify the bias), the share of signal variance attributable to noise ($\rho$), and the volatility of price measurement error ($\sigma_\delta$).

\subsection{Measurement structure}

\begin{definition}[Price noise]\label{def:price_noise}
Following \citet*{blume1983biases}, the observed price $\hat{P}_{i,t}$ relates to the true price $P_{i,t}$ by
\begin{equation}
\hat{P}_{i,t} = (1 + \delta_{i,t}) P_{i,t},
\end{equation}
where $\delta_{i,t}$ is a mean-zero measurement error with finite variance $\sigma_\delta^2 \equiv \Var(\delta_{i,t}).$
\end{definition}

\begin{definition}[Signal decomposition]\label{def:signal_decomp}
The observed sorting signal $\hat{s}_{i,t}$ decomposes as
\begin{equation}
\hat{s}_{i,t} = s_{i,t} + \eta_{i,t},
\end{equation}
where $s_{i,t}$ is the true signal and $\eta_{i,t}$ is the signal's measurement error.
\end{definition}

\begin{definition}[Return decomposition]\label{def:return_decomp}
Define the observed return by $\hat{r}_{i,t+1}\equiv \hat{P}_{i,t+1}/\hat{P}_{i,t}-1$. Then,
\begin{equation}
\hat{r}_{i,t+1} = r_{i,t+1} + \epsilon_{i,t},
\end{equation}
where $r_{i,t+1}\equiv P_{i,t+1}/P_{i,t}-1$ is the true return, and
\begin{equation}\label{eq:eps_exact}
\epsilon_{i,t} \equiv \frac{1+\delta_{i,t+1}}{1+\delta_{i,t}}-1.
\end{equation}
A first-order approximation yields $\epsilon_{i,t}\approx \delta_{i,t+1}-\delta_{i,t}$.
\end{definition}

\begin{assumption}[Serial independence]\label{assum:iid}
For each $i$, $\delta_{i,t}$ is independent across $t$ with $\E[\delta_{i,t}]=0$ and $\Var(\delta_{i,t})=\sigma_\delta^2<\infty.$
\end{assumption}

\subsection{Portfolio formation}

Each month, we sort bonds by their observed signal $\hat{s}_{i,t}$ and form long-short factors from the extreme tails. Let $\alpha$ denote the tail fraction, that is, the share of bonds assigned to each extreme portfolio (e.g., $\alpha = 0.10$ for decile sorts, $\alpha = 0.20$ for quintile sorts). The long portfolio $L$ contains bonds in the top $\alpha$ fraction of the signal distribution; the short portfolio $S$ contains bonds in the bottom $\alpha$ fraction:
\[
L \equiv \{\hat{s}_{i,t}\ge c_L\},\qquad S \equiv \{\hat{s}_{i,t}\le c_S\},
\]
where $c_L$ and $c_S$ are the corresponding quantile cutoffs. The population long-short return is
\begin{equation}\label{eq:LS_def}
\hat r^{LS}_{t+1} \equiv \E[\hat r_{i,t+1}\mid L]-\E[\hat r_{i,t+1}\mid S].
\end{equation}
Throughout, expectations are with respect to the cross-section at time $t$.

\subsection{General mechanism: Why the bias does not cancel}

\subsubsection{Long-short decomposition}

\begin{lemma}[Long--short decomposition]\label{lem:decomp}
For any sorting rule and any return decomposition $\hat r=r+\epsilon$,
\begin{equation}\label{eq:decomp}
\hat r^{LS}_{t+1}
=
\underbrace{\E[r_{i,t+1}\mid L]-\E[r_{i,t+1}\mid S]}_{r^{LS}_{t+1}}
+
\underbrace{\E[\epsilon_{i,t}\mid L]-\E[\epsilon_{i,t}\mid S]}_{\mathrm{Bias}_{t+1}}.
\end{equation}
\end{lemma}

\begin{proof}
Substitute $\hat r_{i,t+1}=r_{i,t+1}+\epsilon_{i,t}$ into Eq.~\eqref{eq:LS_def} and rearrange terms.
\end{proof}

\begin{remark}[When does the bias cancel out?]\label{rem:cancel}
By Eq.~\eqref{eq:decomp}, the bias cancels out if and only if $\E[\epsilon_{i,t}\mid L]=\E[\epsilon_{i,t}\mid S].$ A sufficient (but not necessary) condition is \emph{conditional mean independence}:
\[
\E[\epsilon_{i,t}\mid \hat s_{i,t}]=0 \quad \text{for all values of }\hat s_{i,t}.
\]
\end{remark}

The bias is zero only if portfolio assignment is mean-independent of the return error; price-based portfolio sorting violates this condition because it conditions on the very disturbance that enters the return denominator, so that long--short differencing cannot eliminate the bias.

\subsubsection{Sorting on the bias}

The next result formalizes the key intuition: for price-based signals, portfolio assignment depends on the same disturbance that enters the return measurement error, hence long and short portfolios have different conditional mean return errors.

\begin{assumption}[Overlap between signal noise and return error]\label{assum:overlap}
There exists a scalar disturbance $\delta_{i,t}$ such that the signal error and the return measurement error can be written as
\[
\eta_{i,t}=g(\delta_{i,t}),\qquad \epsilon_{i,t}=h(\delta_{i,t})+u_{i,t},
\]
where $g$ and $h$ are non-constant functions, and $\E[u_{i,t}\mid \delta_{i,t}]=0.$
\end{assumption}

\begin{assumption}[Monotone contamination]\label{assum:mono}
The function $g$ is (weakly) monotone on the support of $\delta_{i,t}$.
\end{assumption}

For the price-based signals in our `factor zoo' (yields, spreads, value ratios, past returns), the signal error is a monotone function of the price disturbance; see Section~\ref{sec:mmn} for discussion.

\begin{assumption}[Signal-noise independence]\label{assum:signal_noise_indep}
The true signal $s_{i,t}$ is independent of the price measurement error $\delta_{i,t}$.
\end{assumption}

This assumption holds when the true characteristic is determined independently of the trading process that generates measurement error. It may be violated if bonds with extreme true signals (e.g., very high yields) are also more illiquid and therefore have larger $|\delta_{i,t}|$; such dependence through the variance channel could amplify or dampen the bias beyond what Proposition~\ref{prop:sign} predicts.

\begin{proposition}[Sorting on the bias: sign result]\label{prop:sign}
Under Assumptions \ref{assum:iid}--\ref{assum:signal_noise_indep}, if $h(\delta)$ is (weakly) decreasing and $g(\delta)$ is (weakly) decreasing, then
\[
\E[\epsilon_{i,t}\mid L] \;\ge\; \E[\epsilon_{i,t}] \;\ge\; \E[\epsilon_{i,t}\mid S],
\]
with strict inequalities whenever the sorting is non-trivial. In particular, the long--short bias is nonnegative:
\[
\mathrm{Bias}_{t+1}=\E[\epsilon_{i,t}\mid L]-\E[\epsilon_{i,t}\mid S]\;\ge\; 0.
\]
\end{proposition}

\begin{proof}
By Definition \ref{def:signal_decomp}, $\hat s_{i,t}=s_{i,t}+g(\delta_{i,t})$. Under Assumption~\ref{assum:signal_noise_indep}, $s_{i,t}\perp\delta_{i,t}$. Holding $s_{i,t}$ fixed, since $g$ is decreasing, higher values of $\hat s_{i,t}$ correspond to (stochastically) lower values of $\delta_{i,t}$, and lower values correspond to (stochastically) higher values of $\delta_{i,t}$. Thus, relative to the unconditional distribution, the event $L$ (upper tail) tilts the distribution of $\delta_{i,t}$ downward, while $S$ (lower tail) tilts it upward.

Under Assumption \ref{assum:overlap}, $\E[\epsilon_{i,t}\mid \delta_{i,t}]=h(\delta_{i,t})$. Since $h$ is weakly decreasing, downward shifts in $\delta_{i,t}$ weakly increase $\E[\epsilon_{i,t}\mid \cdot]$ and upward shifts weakly decrease it. Therefore,
\[
\E[\epsilon_{i,t}\mid L] \ge \E[\epsilon_{i,t}] \ge \E[\epsilon_{i,t}\mid S],
\]
and the bias is nonnegative. Serial independence (Assumption \ref{assum:iid}) justifies interpreting the bias-relevant component of $\epsilon_{i,t}$ as driven by the time-$t$ disturbance (e.g., in the first-order approximation $\epsilon_{i,t}\approx \delta_{i,t+1}-\delta_{i,t}$, the $\delta_{i,t+1}$ term has mean zero conditional on time-$t$ sorting).
\end{proof}

\begin{remark}[No need for true predictability]\label{rem:nopredict}
Proposition \ref{prop:sign} does not require the true signal $s_{i,t}$ to predict true returns. Even in a pure-noise world where $r_{i,t+1}\equiv 0$ (so $r^{LS}_{t+1}=0$), the bias term in Lemma \ref{lem:decomp} can be nonzero because sorting is conditioning on the disturbance that enters $\epsilon_{i,t}$.
\end{remark}

If a disturbance $\delta_{i,t}$ enters both the signal error $\eta_{i,t}$ and the return error $\epsilon_{i,t}$, then portfolio assignment based on the noisy signal $\hat{s}_{i,t}$ conditions on a component of $\epsilon_{i,t}$. As a result, $\E[\epsilon_{i,t}\mid L] \neq \E[\epsilon_{i,t}\mid S]$ and the bias does not cancel out. This holds even when the true signal $s_{i,t}$ has no predictive power.

For non-price signals (e.g., rating), the measurement error does not enter the return measurement error, so sorting on a noisy non-price signal does not create CEIV bias.

\subsubsection{Non-price signals do not generate CEIV bias}

Let a \emph{non-price} characteristic (e.g., rating) be observed with error:
\[
\hat s_{i,t}=s_{i,t}+\xi_{i,t},
\]
where $\xi_{i,t}$ is a measurement error unrelated to return measurement.

\begin{corollary}[No CEIV bias for non-price signal noise]\label{cor:nonprice}
Suppose portfolio assignment depends on $(s_{i,t},\xi_{i,t})$ only through $\hat s_{i,t}=s_{i,t}+\xi_{i,t}$. If
\begin{equation}\label{eq:nonprice_cmi}
\E[\epsilon_{i,t}\mid s_{i,t},\xi_{i,t}]=0,
\end{equation}
then $\E[\epsilon_{i,t}\mid L]=\E[\epsilon_{i,t}\mid S]=0$, hence $\mathrm{Bias}_{t+1}=0$ and $\hat r^{LS}_{t+1}=r^{LS}_{t+1}$.
\end{corollary}

\begin{proof}
Using the law of iterated expectations and \eqref{eq:nonprice_cmi},
\[
\E[\epsilon_{i,t}\mid L]
=\E\!\big[\E[\epsilon_{i,t}\mid s_{i,t},\xi_{i,t}]\mid L\big]
=\E[0\mid L]=0,
\]
and similarly for $S$.
\end{proof}

\begin{remark}[``no CEIV bias'' does not mean ``measurement error is harmless'']\label{rem:caveat}
Corollary \ref{cor:nonprice} rules out the \emph{mechanical} long--short bias that arises when the same disturbance contaminates both the sorting variable and the return denominator. It does \emph{not} imply that non-price characteristics measured with error are innocuous: misclassification can attenuate $r^{LS}$, change portfolio composition and risk exposures, and can affect inference.
Noisy non-price signals can lead to misranking but they do not create a spurious long-short premium unless the signal noise also contaminates the return measurement error.
\end{remark}

\subsection{Closed-form characterization under stricter assumptions}
Proposition~\ref{prop:sign} establishes the sign of the bias under general conditions, and Corollary~\ref{cor:nonprice} shows that non-price signals are immune. We now impose distributional assumptions to derive the bias magnitude in closed form.

\begin{assumption}[Linear price contamination]\label{assum:linear}
For price-based signals, $\eta_{i,t}=a\,\delta_{i,t}$ for some constant $a\neq 0.$
\end{assumption}

\begin{assumption}[Joint normality and orthogonality]\label{assum:normal}
The vector $(s_{i,t},\eta_{i,t},\epsilon_{i,t})$ is jointly normal and $s_{i,t}\perp \eta_{i,t}.$ Since $\eta_{i,t}=a\,\delta_{i,t}$ under Assumption~\ref{assum:linear}, normality of $\eta_{i,t}$ implies normality of $\delta_{i,t}.$
\end{assumption}

\begin{proposition}[Correlation structure]\label{prop:corr}
Under Assumptions \ref{assum:iid} and \ref{assum:linear}, and using the first-order approximation $\epsilon_{i,t}\approx \delta_{i,t+1}-\delta_{i,t}$, we have
\begin{equation}
\Cov(\eta_{i,t},\epsilon_{i,t}) = -a\sigma_\delta^2 \neq 0.
\end{equation}
\end{proposition}

\begin{proof}
Under Assumption \ref{assum:iid}, $\delta_{i,t+1}$ is independent of $\delta_{i,t}$ and has mean zero conditional on time-$t$ sorting. Thus, the bias-relevant component of $\epsilon_{i,t}$ is $-\delta_{i,t}$, so
\[
\Cov(\eta_{i,t},\epsilon_{i,t})
\approx \Cov(a\delta_{i,t},-\delta_{i,t})
=-a\,\Var(\delta_{i,t})
=-a\sigma_\delta^2.
\]
\end{proof}

Define $\sigma_s^2\equiv \Var(s_{i,t})$, $\sigma_\eta^2\equiv \Var(\eta_{i,t})$, and $\sigma_{\hat s}^2\equiv \Var(\hat s_{i,t})=\sigma_s^2+\sigma_\eta^2$ (by $s\perp \eta$ in Assumption~\ref{assum:normal}). Let the \emph{noise share} of the observed signal be
\[
\rho \equiv \frac{\sigma_\eta}{\sigma_{\hat s}}\in[0,1].
\]
Note that $\rho$ denotes the noise share (a ratio of standard deviations), not a correlation coefficient. Under Assumption~\ref{assum:linear} with $\sigma_\eta = |a|\sigma_\delta$, the bias in Eq.~\eqref{eq:bias_closed} can equivalently be written as $2\kappa|a|\sigma_\delta^2/\sigma_{\hat{s}}$. We describe the bias as first-order in $\sigma_\delta$ in the sense that it scales linearly with price noise for a fixed noise share $\rho$.
Finally, define a positive tail constant
\begin{equation}\label{eq:kappa_def}
\kappa \equiv -\,\E\!\left[ Z \mid Z < \Phinv(\alpha)\right] > 0,
\qquad Z\sim\mathcal N(0,1).
\end{equation}

\begin{theorem}[Sorting-induced bias: closed form]\label{thm:closedform}
Under Assumptions \ref{assum:iid}, \ref{assum:signal_noise_indep}, \ref{assum:linear}, and \ref{assum:normal}, the bias in a long--short factor sorted on a price-based signal is
\begin{equation}\label{eq:bias_closed}
\mathrm{Bias}_{t+1} = 2\kappa\,\rho\,\sigma_\delta,
\end{equation}
where $\kappa$ is defined in Eq.~\eqref{eq:kappa_def} and $\sigma_\delta$ is the previously-defined standard deviation of the price noise.
\end{theorem}

\begin{proof}
Under joint normality, the conditional expectation of $\eta$ given $\hat s$ is linear:
\[
\E[\eta_{i,t}\mid \hat s_{i,t}]
=\frac{\Cov(\eta_{i,t},\hat s_{i,t})}{\Var(\hat s_{i,t})}\,\hat s_{i,t}
=\frac{\sigma_\eta^2}{\sigma_{\hat s}^2}\,\hat s_{i,t},
\]
since $\hat s=s+\eta$ and $s\perp \eta$.

For the long leg $L=\{\hat s\ge c_L\},$ where $c_L$ is the upper $\alpha$-quantile, standard normal truncation implies
\[
\E[\hat s_{i,t}\mid L] = +\kappa\,\sigma_{\hat s}.
\]
Hence,
\[
\E[\eta_{i,t}\mid L]
=\frac{\sigma_\eta^2}{\sigma_{\hat s}^2}\,\E[\hat s_{i,t}\mid L]
=+\kappa\,\frac{\sigma_\eta^2}{\sigma_{\hat s}}
=+\kappa\,\rho\,\sigma_\eta.
\]
Under Assumption~\ref{assum:linear}, $\eta=a\delta$ and $\sigma_\eta=|a|\sigma_\delta$. Without loss of generality, assume $a<0$ (as for our price-based signals); if $a>0$, redefine the signal as $-\hat{s}_{i,t}$, which swaps $L$ and $S$ and reverses the sign of $a$ without changing the bias magnitude. Then
\[
\E[\delta_{i,t}\mid L] = \frac{1}{a}\E[\eta_{i,t}\mid L] = -\kappa\,\rho\,\sigma_\delta,
\]
and by symmetry $\E[\delta_{i,t}\mid S]=+\kappa\,\rho\,\sigma_\delta.$

Finally, the first-order approximation (Definition~\ref{def:return_decomp}) gives $\epsilon_{i,t}\approx \delta_{i,t+1}-\delta_{i,t}$. Under Assumption~\ref{assum:iid}, $\delta_{i,t+1}$ is independent of time-$t$ information, so $\E[\delta_{i,t+1}\mid L]=\E[\delta_{i,t+1}]=0$. Hence $\E[\epsilon_{i,t}\mid L]\approx -\E[\delta_{i,t}\mid L]$, and similarly for $S$. Therefore,
\[
\E[\epsilon_{i,t}\mid L]-\E[\epsilon_{i,t}\mid S]
\approx
-\E[\delta_{i,t}\mid L]+\E[\delta_{i,t}\mid S]
=
2\kappa\,\rho\,\sigma_\delta,
\]
which yields the expression in Eq.~\eqref{eq:bias_closed}. The second-order remainder, approximately $\delta_{i,t}^2 - \delta_{i,t}\delta_{i,t+1}$, has positive conditional expectation in both tails but cancels exactly in the long--short difference under normality, because the symmetric truncation of $\delta_{i,t}$ induced by the long and short portfolios yields identical conditional second moments; under asymmetric noise distributions, a second-order residual may remain.
\end{proof}

\begin{corollary}[Pure-noise limit]\label{cor:pure}
If the true signal has no cross-sectional variation, $\sigma_s=0$ (equivalently, $\rho=1$), then
\[
\mathrm{Bias}_{t+1}=2\kappa\,\sigma_\delta.
\]
\end{corollary}

\begin{corollary}[Reversal signals]\label{cor:reversal}
Under Assumptions~\ref{assum:iid}, \ref{assum:signal_noise_indep}, and~\ref{assum:normal}, for reversal signals where
$\eta_{i,t} = \delta_{i,t} - \delta_{i,t-1}$, the bias is
$|\mathrm{Bias}| = \sqrt{2}\kappa\rho\sigma_\delta$.
\end{corollary}

\begin{proof}
The reversal structure $\eta_{i,t}=\delta_{i,t}-\delta_{i,t-1}$ does not satisfy the linear form $\eta=a\delta_t$ in Assumption~\ref{assum:linear}; we therefore compute $\Cov(\delta_{i,t},\hat{s}_{i,t})$ directly.

Under Assumption~\ref{assum:iid}, the bias-relevant component of the return
measurement error is $\epsilon_{i,t}\approx -\delta_{i,t}$. For the reversal
signal $\eta_{i,t}=\delta_{i,t}-\delta_{i,t-1}$, serial independence implies
$\Cov(\delta_{i,t},\delta_{i,t-1})=0$, so that
\[
\Var(\eta_{i,t})=\Var(\delta_{i,t})+\Var(\delta_{i,t-1})=2\sigma_\delta^2,
\qquad
\Cov(\delta_{i,t},\eta_{i,t})=\sigma_\delta^2.
\]
Under joint normality (Assumption~\ref{assum:normal}),
\[
\E[\delta_{i,t}\mid \hat s_{i,t}]
=
\frac{\Cov(\delta_{i,t},\hat s_{i,t})}{\Var(\hat s_{i,t})}\,\hat s_{i,t}
=
\frac{\sigma_\delta^2}{\sigma_{\hat s}^2}\,\hat s_{i,t},
\]
where $\Cov(\delta_{i,t},\hat s_{i,t})=\Cov(\delta_{i,t},\eta_{i,t})=\sigma_\delta^2$ (using $s_{i,t}\perp\delta_{i,t}$ and Assumption~\ref{assum:iid}) and $\Var(\hat s_{i,t})=\sigma_s^2+\sigma_\eta^2$ (using $s_{i,t}\perp\eta_{i,t}$ from Assumption~\ref{assum:normal}). For the long
leg, $\E[\hat s_{i,t}\mid L]=+\kappa\sigma_{\hat s}$, hence
\[
\E[\delta_{i,t}\mid L]=+\kappa\frac{\sigma_\delta^2}{\sigma_{\hat s}}
=+\kappa\frac{\rho}{\sqrt{2}}\sigma_\delta,
\]
where $\rho=\sigma_\eta/\sigma_{\hat s}$ and $\sigma_\eta=\sqrt{2}\sigma_\delta$.
By symmetry, $\E[\delta_{i,t}\mid S]=-\kappa\frac{\rho}{\sqrt{2}}\sigma_\delta$.
The bias is
\[
\mathrm{Bias}
=
\E[\epsilon\mid L]-\E[\epsilon\mid S]
\approx
-\E[\delta\mid L]+\E[\delta\mid S]
=
-\sqrt{2}\kappa\rho\sigma_\delta,
\]
so $|\mathrm{Bias}|=\sqrt{2}\kappa\rho\sigma_\delta$.
\end{proof}

\subsection{Empirical magnitude of the implementation gap}\label{app:lib-empirical}

Table~\ref{tab:lib_summary} quantifies the latent implementation bias across all 108 factors. For each factor, we compute the difference $\Delta_{f,t} = r_{f,t}^{\text{End}} - r_{f,t}^{\text{Bgn}}$ between month-end and month-begin long--short factor returns, holding the signal fixed at its month-end value. A large and statistically significant $\Delta$ indicates that the one-day implementation gap materially affects the measured premium.

The results are split sharply by signal type. For non-price signals (credit ratings, momentum, factor betas), the signal noise $\xi_{i,t}$ is independent of the price disturbance $\delta_{i,t}$, so $\Cov(\xi_{i,t}, \epsilon_{i,t}) = 0$ and Corollary~\ref{cor:nonprice} implies zero CEIV bias. Among the 30 price-based factors, 43--67\% exhibit a statistically significant gap across the four sort specifications (single-sort and within-firm, value- and equal-weighted), with average absolute differences of 0.13--0.30\% per month and average absolute $t$-statistics of 3.08--4.05. Among the 78 non-price factors, only 10--18\% show a significant gap, with average absolute differences of 0.03--0.10\% per month and average absolute $t$-statistics near one. The one-day implementation gap has no material effect on average premia for non-price signals, consistent with Corollary~\ref{cor:nonprice}.

\begin{table}[!ht]
\caption{Latent implementation bias across 108 factors.}
\begin{spacing}{1}
{\footnotesize
The table tests whether the one-day implementation gap produces a significant return difference for each of the 108 factors. For each factor and portfolio sort, $\Delta_{f,t} = r_{f,t}^{\text{End}} - r_{f,t}^{\text{Bgn}}$ measures the difference between month-end and month-begin long--short returns when both series use signal-adjusted sorting variables. Sig./Total counts factors with $|t(\Delta)| > 1.96$ (Newey-West). $\overline{|\Delta|}$ is the average absolute monthly return difference (\%). $\overline{|t|}$ is the average absolute $t$-statistic. Sample: 2002-09 to 2024-12.}
\end{spacing}
\vspace{-4mm}
\begin{center}
\label{tab:lib_summary}
\vspace{2mm}
\scalebox{0.85}{%
\begin{tabular}{l cccc cccc}
\toprule
 & \multicolumn{4}{c}{\textbf{Price-Based}} & \multicolumn{4}{c}{\textbf{Non-Price-Based}} \\
\cmidrule(lr){2-5} \cmidrule(lr){6-9}
\textbf{Sort} & \textbf{Sig./Total} & \textbf{\%} & \textbf{$\overline{|\Delta|}$} & \textbf{$\overline{|t|}$} & \textbf{Sig./Total} & \textbf{\%} & \textbf{$\overline{|\Delta|}$} & \textbf{$\overline{|t|}$} \\
\midrule
Single-Sorted (VW) & 13/30 & 43\% & 0.18 & 3.08 & 11/78 & 14\% & 0.05 & 1.10 \\
Single-Sorted (EW) & 15/30 & 50\% & 0.30 & 3.18 & 14/78 & 18\% & 0.10 & 1.20 \\
Within-Firm (VW) & 14/30 & 47\% & 0.13 & 3.32 & 8/78 & 10\% & 0.03 & 1.07 \\
Within-Firm (EW) & 20/30 & 67\% & 0.19 & 4.05 & 11/78 & 14\% & 0.04 & 1.26 \\
\bottomrule
\end{tabular}
}
\end{center}
\end{table}



\begin{thebibliography}{}

\bibitem[\protect\citeauthoryear{Abdi and Ranaldo}{Abdi and
  Ranaldo}{2017}]{abdi2017simple}
Abdi, F. and A.~Ranaldo (2017).
\newblock A simple estimation of bid-ask spreads from daily close, high, and
  low prices.
\newblock {\em Review of Financial Studies\/}~{\em 30}, 4437--4480.

\bibitem[\protect\citeauthoryear{Amihud}{Amihud}{2002}]{AM2002}
Amihud, Y. (2002).
\newblock Illiquidity and stock returns: {C}ross-section and time-series
  effects.
\newblock {\em Journal of Financial Markets\/}~{\em 5}, 31--56.

\bibitem[\protect\citeauthoryear{Andreani, Palhares, and Richardson}{Andreani
  et~al.}{2024}]{andreani2023computing}
Andreani, M., D.~Palhares, and S.~Richardson (2024).
\newblock Computing corporate bond returns: {A} word (or two) of caution.
\newblock {\em Review of Accounting Studies\/}~{\em 29}, 3887--3906.

\bibitem[\protect\citeauthoryear{Ang, Chen, and Xing}{Ang
  et~al.}{2006}]{ang2006downside}
Ang, A., J.~Chen, and Y.~Xing (2006).
\newblock Downside risk.
\newblock {\em Review of Financial Studies\/}~{\em 19}, 1191--1239.

\bibitem[\protect\citeauthoryear{Bai, Bali, and Wen}{Bai
  et~al.}{2019}]{BaiBaliWen_2019}
Bai, J., T.~G. Bali, and Q.~Wen (2019).
\newblock {RETRACTED}: Common risk factors in the cross-section of corporate
  bond returns.
\newblock {\em Journal of Financial Economics\/}~{\em 131}, 619--642.

\bibitem[\protect\citeauthoryear{Bai, Bali, and Wen}{Bai
  et~al.}{2021}]{bai2021there}
Bai, J., T.~G. Bali, and Q.~Wen (2021).
\newblock Is there a risk-return tradeoff in the corporate bond market?
  {T}ime-series and cross-sectional evidence.
\newblock {\em Journal of Financial Economics\/}~{\em 142}, 1017--1037.

\bibitem[\protect\citeauthoryear{Baker, Bloom, and Davis}{Baker
  et~al.}{2016}]{baker2016measuring}
Baker, S.~R., N.~Bloom, and S.~J. Davis (2016).
\newblock Measuring economic policy uncertainty.
\newblock {\em Quarterly Journal of Economics\/}~{\em 131}, 1593--1636.

\bibitem[\protect\citeauthoryear{Bali, Subrahmanyam, and Wen}{Bali
  et~al.}{2021a}]{bali2021long}
Bali, T.~G., A.~Subrahmanyam, and Q.~Wen (2021a).
\newblock Long-term reversals in the corporate bond market.
\newblock {\em Journal of Financial Economics\/}~{\em 139}, 656--677.

\bibitem[\protect\citeauthoryear{Bali, Subrahmanyam, and Wen}{Bali
  et~al.}{2021b}]{bali2021macroeconomic}
Bali, T.~G., A.~Subrahmanyam, and Q.~Wen (2021b).
\newblock The macroeconomic uncertainty premium in the corporate bond market.
\newblock {\em Journal of Financial and Quantitative Analysis\/}~{\em 56},
  1653--1678.

\bibitem[\protect\citeauthoryear{Bali, Subrahmanyam, and Wen}{Bali
  et~al.}{2023}]{bali2023macroeconomic}
Bali, T.~G., A.~Subrahmanyam, and Q.~Wen (2023).
\newblock The macroeconomic uncertainty premium in the corporate bond
  market---{C}orrigendum.
\newblock {\em Journal of Financial and Quantitative Analysis\/}.

\bibitem[\protect\citeauthoryear{Bao, Pan, and Wang}{Bao
  et~al.}{2011}]{bao2011illiquidity}
Bao, J., J.~Pan, and J.~Wang (2011).
\newblock The illiquidity of corporate bonds.
\newblock {\em Journal of Finance\/}~{\em 66}, 911--946.

\bibitem[\protect\citeauthoryear{Bartram, Grinblatt, and Nozawa}{Bartram
  et~al.}{2025}]{Bartram-Grinblatt-Nozawa-2021}
Bartram, S.~M., M.~Grinblatt, and Y.~Nozawa (2025).
\newblock Book-to-market, mispricing, and the cross-section of corporate bond
  returns.
\newblock {\em Journal of Financial and Quantitative Analysis\/}~{\em 60},
  1185--1233.

\bibitem[\protect\citeauthoryear{Baumann, Kakhbod, Livdan, Nazemi, and
  Sch{\"u}rhoff}{Baumann et~al.}{2025}]{baumann2025life}
Baumann, F., A.~Kakhbod, D.~Livdan, A.~Nazemi, and N.~Sch{\"u}rhoff (2025).
\newblock Life after default: How dealer intermediation improves default
  recovery.
\newblock Working Paper.

\bibitem[\protect\citeauthoryear{Baumann and Nazemi}{Baumann and
  Nazemi}{2025}]{baumann2025defaulted}
Baumann, F. and A.~Nazemi (2025).
\newblock Defaulted bonds: {A} hybrid asset priced by bond and equity markets.
\newblock Working Paper.

\bibitem[\protect\citeauthoryear{Benjamini and Hochberg}{Benjamini and
  Hochberg}{1995}]{benjamini1995controlling}
Benjamini, Y. and Y.~Hochberg (1995).
\newblock Controlling the false discovery rate: {A} practical and powerful
  approach to multiple testing.
\newblock {\em Journal of the Royal Statistical Society: {S}eries B
  (Methodological)\/}~{\em 57}, 289--300.

\bibitem[\protect\citeauthoryear{Bessembinder, Kahle, Maxwell, and
  Xu}{Bessembinder et~al.}{2008}]{bessembinder2008measuring}
Bessembinder, H., K.~M. Kahle, W.~F. Maxwell, and D.~Xu (2008).
\newblock Measuring abnormal bond performance.
\newblock {\em Review of Financial Studies\/}~{\em 22}, 4219--4258.

\bibitem[\protect\citeauthoryear{Blitz, Huij, and Martens}{Blitz
  et~al.}{2011}]{blitz2011residual}
Blitz, D., J.~Huij, and M.~Martens (2011).
\newblock Residual momentum.
\newblock {\em Journal of Empirical Finance\/}~{\em 18}, 506--521.

\bibitem[\protect\citeauthoryear{Blume and Stambaugh}{Blume and
  Stambaugh}{1983}]{blume1983biases}
Blume, M.~E. and R.~F. Stambaugh (1983).
\newblock Biases in computed returns: {A}n application to the size effect.
\newblock {\em Journal of Financial Economics\/}~{\em 12}, 387--404.

\bibitem[\protect\citeauthoryear{Bollerslev, Li, and Zhao}{Bollerslev
  et~al.}{2020}]{bollerslev2020}
Bollerslev, T., S.~Z. Li, and B.~Zhao (2020).
\newblock Realized semicovariances.
\newblock {\em Econometrica\/}~{\em 88}, 1515--1551.

\bibitem[\protect\citeauthoryear{Ceballos}{Ceballos}{2022}]{ceballos2021inflation}
Ceballos, L. (2022).
\newblock Inflation volatility risk and the cross-section of corporate bond
  returns.
\newblock Working Paper.

\bibitem[\protect\citeauthoryear{Chen and Zimmermann}{Chen and
  Zimmermann}{2022}]{chen2021open}
Chen, A.~Y. and T.~Zimmermann (2022).
\newblock Open source cross-sectional asset pricing.
\newblock {\em Critical Finance Review\/}~{\em 27}, 207--264.

\bibitem[\protect\citeauthoryear{Chen and Choi}{Chen and
  Choi}{2024}]{chen2024reaching}
Chen, Q. and J.~Choi (2024).
\newblock Reaching for yield and the cross section of bond returns.
\newblock {\em Management Science\/}~{\em 70}, 5226--5245.

\bibitem[\protect\citeauthoryear{Choi}{Choi}{2013}]{choi2013drives}
Choi, J. (2013).
\newblock What drives the value premium?: The role of asset risk and leverage.
\newblock {\em Review of Financial Studies\/}~{\em 26}, 2845--2875.

\bibitem[\protect\citeauthoryear{Choi, Han, Shin, and Yoon}{Choi
  et~al.}{2026}]{choi2026illiquid}
Choi, J., J.~Han, S.~S. Shin, and J.~H. Yoon (2026).
\newblock The more illiquid, the more expensive: A search-based explanation of
  the illiquidity premium.
\newblock Working paper.

\bibitem[\protect\citeauthoryear{Chordia, Goyal, Nozawa, Subrahmanyam, and
  Tong}{Chordia et~al.}{2017}]{chordia2017capital}
Chordia, T., A.~Goyal, Y.~Nozawa, A.~Subrahmanyam, and Q.~Tong (2017).
\newblock Are capital market anomalies common to equity and corporate bond
  markets? {A}n empirical investigation.
\newblock {\em Journal of Financial and Quantitative Analysis\/}~{\em 52},
  1301--1342.

\bibitem[\protect\citeauthoryear{Chung, Wang, and Wu}{Chung
  et~al.}{2019}]{ChungWangWu_2019}
Chung, K.~H., J.~Wang, and C.~Wu (2019).
\newblock Volatility and the cross-section of corporate bond returns.
\newblock {\em Journal of Financial Economics\/}~{\em 133}, 397--417.

\bibitem[\protect\citeauthoryear{Coase}{Coase}{1982}]{coase1982choose}
Coase, R.~H. (1982).
\newblock {\em How Should Economists Choose?}
\newblock G. Warren Nutter Lectures in Political Economy. Washington, D.C.:
  American Enterprise Institute.

\bibitem[\protect\citeauthoryear{Conrad, Gultekin, and Kaul}{Conrad
  et~al.}{1997}]{conrad1997profitability}
Conrad, J., M.~N. Gultekin, and G.~Kaul (1997).
\newblock Profitability of short-term contrarian strategies: Implications for
  market efficiency.
\newblock {\em Journal of Business \& Economic Statistics\/}~{\em 15},
  379--386.

\bibitem[\protect\citeauthoryear{Corwin and Schultz}{Corwin and
  Schultz}{2012}]{corwin2012simple}
Corwin, S.~A. and P.~Schultz (2012).
\newblock A simple way to estimate bid-ask spreads from daily high and low
  prices.
\newblock {\em Journal of Finance\/}~{\em 67}, 719--760.

\bibitem[\protect\citeauthoryear{Danyliv, Bland, and Nicholass}{Danyliv
  et~al.}{2014}]{danyliv2014convenient}
Danyliv, O., B.~Bland, and D.~Nicholass (2014).
\newblock A convenient liquidity measure.
\newblock {\em Journal of Trading\/}~{\em 9}, 38--49.

\bibitem[\protect\citeauthoryear{Dick-Nielsen}{Dick-Nielsen}{2014}]{DickNielsen2014HowTC}
Dick-Nielsen, J. (2014).
\newblock How to clean {E}nhanced {TRACE} data.
\newblock Working Paper.

\bibitem[\protect\citeauthoryear{Dick-Nielsen, Feldh{\"u}tter, and
  Lando}{Dick-Nielsen et~al.}{2012}]{dick2012corporate}
Dick-Nielsen, J., P.~Feldh{\"u}tter, and D.~Lando (2012).
\newblock Corporate bond liquidity before and after the onset of the subprime
  crisis.
\newblock {\em Journal of Financial Economics\/}~{\em 103}, 471--492.

\bibitem[\protect\citeauthoryear{Dick-Nielsen, Feldh\"{u}tter, Pedersen, and
  Stolborg}{Dick-Nielsen et~al.}{2023}]{feldhutter2023}
Dick-Nielsen, J., P.~Feldh\"{u}tter, L.~H. Pedersen, and C.~Stolborg (2023).
\newblock Corporate bond factors: {R}eplication failures and a new framework.
\newblock Working Paper.

\bibitem[\protect\citeauthoryear{Dickerson, Mueller, and Robotti}{Dickerson
  et~al.}{2023}]{dickerson2023priced}
Dickerson, A., P.~Mueller, and C.~Robotti (2023).
\newblock Priced risk in corporate bonds.
\newblock {\em Journal of Financial Economics\/}~{\em 150, article 103707}.

\bibitem[\protect\citeauthoryear{Dickerson, Robotti, and Nozawa}{Dickerson
  et~al.}{2025}]{dickerson2024factor}
Dickerson, A., C.~Robotti, and Y.~Nozawa (2025).
\newblock Factor investing with delays.
\newblock Working Paper.

\bibitem[\protect\citeauthoryear{Duarte, Jones, Khorram, Mo, and Wang}{Duarte
  et~al.}{2025}]{duarte2023too}
Duarte, J., C.~S. Jones, M.~Khorram, H.~Mo, and J.~L. Wang (2025).
\newblock Too good to be true: Look-ahead bias in empirical options research.
\newblock {\em Review of Financial Studies\/}.
\newblock Forthcoming.

\bibitem[\protect\citeauthoryear{Duarte, Jones, and Wang}{Duarte
  et~al.}{2024}]{duarte2024very}
Duarte, J., C.~S. Jones, and J.~L. Wang (2024).
\newblock Very noisy option prices and inference regarding the volatility risk
  premium.
\newblock {\em Journal of Finance\/}~{\em 79}, 3581--3621.

\bibitem[\protect\citeauthoryear{Elkamhi, Jo, and Nozawa}{Elkamhi
  et~al.}{2024}]{elkamhi2022one}
Elkamhi, R., C.~Jo, and Y.~Nozawa (2024).
\newblock A one-factor model of corporate bond premia.
\newblock {\em Management Science\/}~{\em 70}, 1875--1900.

\bibitem[\protect\citeauthoryear{Fama}{Fama}{1984}]{fama1984information}
Fama, E.~F. (1984).
\newblock The information in the term structure.
\newblock {\em Journal of Financial Economics\/}~{\em 13}, 509--528.

\bibitem[\protect\citeauthoryear{Fama and MacBeth}{Fama and
  MacBeth}{1973}]{FamaMacBeth_1973}
Fama, E.~F. and J.~D. MacBeth (1973).
\newblock Risk, return, and equilibrium: {E}mpirical tests.
\newblock {\em Journal of Political Economy\/}~{\em 81}, 607--636.

\bibitem[\protect\citeauthoryear{Fong, Holden, and Trzcinka}{Fong
  et~al.}{2017}]{fong2017}
Fong, K.~Y., C.~W. Holden, and C.~A. Trzcinka (2017).
\newblock What are the best liquidity proxies for global research?
\newblock {\em Review of Finance\/}~{\em 21}, 1355--1401.

\bibitem[\protect\citeauthoryear{Gebhardt, Hvidkjaer, and Swaminathan}{Gebhardt
  et~al.}{2005}]{gebhardt2005cross}
Gebhardt, W.~R., S.~Hvidkjaer, and B.~Swaminathan (2005).
\newblock The cross-section of expected corporate bond returns: {B}etas or
  characteristics?
\newblock {\em Journal of Financial Economics\/}~{\em 75}, 85--114.

\bibitem[\protect\citeauthoryear{Ghaderi, Plante, Roussanov, and Seo}{Ghaderi
  et~al.}{2024}]{ghaderi2024pricing}
Ghaderi, M., S.~Plante, N.~L. Roussanov, and S.~B. Seo (2024).
\newblock Pricing of corporate bonds: Evidence from a century-long
  cross-section.
\newblock Working Paper.

\bibitem[\protect\citeauthoryear{Harvey, Liu, and Zhu}{Harvey
  et~al.}{2016}]{harvey2016and}
Harvey, C.~R., Y.~Liu, and H.~Zhu (2016).
\newblock … and the cross-section of expected returns.
\newblock {\em Review of Financial Studies\/}~{\em 29}, 5--68.

\bibitem[\protect\citeauthoryear{Harvey and Siddique}{Harvey and
  Siddique}{2000}]{harvey2000conditional}
Harvey, C.~R. and A.~Siddique (2000).
\newblock Conditional skewness in asset pricing tests.
\newblock {\em Journal of Finance\/}~{\em 55}, 1263--1295.

\bibitem[\protect\citeauthoryear{He, Kelly, and Manela}{He
  et~al.}{2017}]{he2017intermediary}
He, Z., B.~Kelly, and A.~Manela (2017).
\newblock Intermediary asset pricing: {N}ew evidence from many asset classes.
\newblock {\em Journal of Financial Economics\/}~{\em 126}, 1--35.

\bibitem[\protect\citeauthoryear{Hong and Warga}{Hong and
  Warga}{2000}]{hong2000empirical}
Hong, G. and A.~Warga (2000).
\newblock An empirical study of bond market transactions.
\newblock {\em Financial Analysts Journal\/}~{\em 56}, 32--46.

\bibitem[\protect\citeauthoryear{Hou, Xue, and Zhang}{Hou
  et~al.}{2020}]{hou2020replicating}
Hou, K., C.~Xue, and L.~Zhang (2020).
\newblock Replicating anomalies.
\newblock {\em Review of Financial Studies\/}~{\em 33}, 2019--2133.

\bibitem[\protect\citeauthoryear{Houweling and Van~Zundert}{Houweling and
  Van~Zundert}{2017}]{houweling2017factor}
Houweling, P. and J.~Van~Zundert (2017).
\newblock Factor investing in the corporate bond market.
\newblock {\em Financial Analysts Journal\/}~{\em 73}, 100--115.

\bibitem[\protect\citeauthoryear{Israel, Palhares, and Richardson}{Israel
  et~al.}{2018}]{IPR2018}
Israel, R., D.~Palhares, and S.~Richardson (2018).
\newblock Common factors in corporate bond returns.
\newblock {\em Journal of Investment Management\/}~{\em 16}, 17--46.

\bibitem[\protect\citeauthoryear{Jegadeesh}{Jegadeesh}{1990}]{jegadeesh1990evidence}
Jegadeesh, N. (1990).
\newblock Evidence of predictable behavior of security returns.
\newblock {\em Journal of Finance\/}~{\em 45}, 881--898.

\bibitem[\protect\citeauthoryear{Jensen, Kelly, and Pedersen}{Jensen
  et~al.}{2023}]{jensen2023there}
Jensen, T.~I., B.~Kelly, and L.~H. Pedersen (2023).
\newblock Is there a replication crisis in finance?
\newblock {\em Journal of Finance\/}~{\em 78}, 2465--2518.

\bibitem[\protect\citeauthoryear{Jostova, Nikolova, Philipov, and
  Stahel}{Jostova et~al.}{2013}]{jostova2013momentum}
Jostova, G., S.~Nikolova, A.~Philipov, and C.~W. Stahel (2013).
\newblock Momentum in corporate bond returns.
\newblock {\em Review of Financial Studies\/}~{\em 26}, 1649--1693.

\bibitem[\protect\citeauthoryear{Kelly, Palhares, and Pruitt}{Kelly
  et~al.}{2023}]{KPP2023}
Kelly, B., D.~Palhares, and S.~Pruitt (2023).
\newblock Modeling corporate bond returns.
\newblock {\em Journal of Finance\/}~{\em 78}, 1967--2008.

\bibitem[\protect\citeauthoryear{Koijen, Lustig, and {Van Nieuwerburgh}}{Koijen
  et~al.}{2017}]{koijen2017}
Koijen, R.~S., H.~Lustig, and S.~{Van Nieuwerburgh} (2017).
\newblock The cross-section of managerial ability, incentives, and risk
  preferences.
\newblock {\em Journal of Monetary Economics\/}~{\em 91}, 1--17.

\bibitem[\protect\citeauthoryear{Lair and Blonk}{Lair and
  Blonk}{2024}]{lair2024valuations}
Lair, T. and J.~Blonk (2024).
\newblock Valuations in the dark: {W}hen independent valuators influence
  corporate bond returns.
\newblock Working Paper.

\bibitem[\protect\citeauthoryear{Leamer}{Leamer}{1983}]{leamer1983lets}
Leamer, E.~E. (1983).
\newblock Let's take the con out of econometrics.
\newblock {\em American Economic Review\/}~{\em 73}, 31--43.

\bibitem[\protect\citeauthoryear{Lin, Wang, and Wu}{Lin
  et~al.}{2011}]{LinWangWu_2011}
Lin, H., J.~Wang, and C.~Wu (2011).
\newblock Liquidity risk and expected corporate bond returns.
\newblock {\em Journal of Financial Economics\/}~{\em 99}, 628--650.

\bibitem[\protect\citeauthoryear{Linnainmaa and Roberts}{Linnainmaa and
  Roberts}{2018}]{linnainmaa2018history}
Linnainmaa, J.~T. and M.~R. Roberts (2018).
\newblock The history of the cross-section of stock returns.
\newblock {\em Review of Financial Studies\/}~{\em 31}, 2606--2649.

\bibitem[\protect\citeauthoryear{Liu and Wu}{Liu and
  Wu}{2021}]{liu2021reconstructing}
Liu, Y. and J.~C. Wu (2021).
\newblock Reconstructing the yield curve.
\newblock {\em Journal of Financial Economics\/}~{\em 142}, 1395--1425.

\bibitem[\protect\citeauthoryear{Menkveld, Dreber, Holzmeister, Huber,
  Johannesson, Kirchler, Neususs, Razen, Weitzel, and et~al.}{Menkveld
  et~al.}{2024}]{Menkveld_etal_2024}
Menkveld, A.~J., A.~Dreber, F.~Holzmeister, J.~Huber, M.~Johannesson,
  M.~Kirchler, S.~Neususs, M.~Razen, U.~Weitzel, and et~al. (2024).
\newblock Non-standard errors.
\newblock {\em Journal of Finance\/}~{\em 79}, 2339--2390.

\bibitem[\protect\citeauthoryear{Novy-Marx}{Novy-Marx}{2012}]{novy2012momentum}
Novy-Marx, R. (2012).
\newblock Is momentum really momentum?
\newblock {\em Journal of Financial Economics\/}~{\em 103}, 429--453.

\bibitem[\protect\citeauthoryear{P\'{a}stor and Stambaugh}{P\'{a}stor and
  Stambaugh}{2003}]{PS2003}
P\'{a}stor, L. and R.~F. Stambaugh (2003).
\newblock Liquidity risk and expected stock returns.
\newblock {\em Journal of Political Economy\/}~{\em 111}, 642--685.

\bibitem[\protect\citeauthoryear{Richardson and Palhares}{Richardson and
  Palhares}{2018}]{richardson2018illiquidity}
Richardson, S. and D.~Palhares (2018).
\newblock ({I}l)liquidity premium in credit markets: A myth?
\newblock {\em Journal of Fixed Income\/}~{\em 28}, 3--31.

\bibitem[\protect\citeauthoryear{Roll}{Roll}{1984}]{roll1984simple}
Roll, R. (1984).
\newblock A simple implicit measure of the effective bid-ask spread in an
  efficient market.
\newblock {\em Journal of Finance\/}~{\em 39}, 1127--1139.

\bibitem[\protect\citeauthoryear{Soebhag, Van~Vliet, and Verwijmeren}{Soebhag
  et~al.}{2024}]{soebhag2024non}
Soebhag, A., B.~Van~Vliet, and P.~Verwijmeren (2024).
\newblock Non-standard errors in asset pricing: {M}ind your sorts.
\newblock {\em Journal of Empirical Finance\/}~{\em 78}, 101517.

\bibitem[\protect\citeauthoryear{Stambaugh}{Stambaugh}{1988}]{stambaugh1988information}
Stambaugh, R.~F. (1988).
\newblock The information in forward rates: {I}mplications for models of the
  term structure.
\newblock {\em Journal of Financial Economics\/}~{\em 21}, 41--70.

\bibitem[\protect\citeauthoryear{Subrahmanyam}{Subrahmanyam}{2023}]{subrahmanyam2023corporatebonddata}
Subrahmanyam, A. (2023).
\newblock Corporate bond data projects: Some clarifications.
\newblock Working Paper.

\bibitem[\protect\citeauthoryear{Tobek}{Tobek}{2016}]{tobek2016}
Tobek, O. (2016).
\newblock Liquidity proxies using daily trading volume.
\newblock Working Paper.

\bibitem[\protect\citeauthoryear{van Binsbergen, Nozawa, and Schwert}{van
  Binsbergen et~al.}{2025}]{Binsbergen-Schwert-Nozawa-2023}
van Binsbergen, J.~H., Y.~Nozawa, and M.~Schwert (2025).
\newblock Duration-based valuation of corporate bonds.
\newblock {\em Review of Financial Studies\/}~{\em 38}, 158--191.

\bibitem[\protect\citeauthoryear{Walter, Weber, and Weiss}{Walter
  et~al.}{2024}]{Walter_etal_2024}
Walter, D., R.~Weber, and P.~Weiss (2024).
\newblock Methodological uncertainty in portfolio sorts.
\newblock Working Paper.

\bibitem[\protect\citeauthoryear{Wang, Wu, and Yang}{Wang
  et~al.}{2024}]{wang2024industry}
Wang, J., D.~Wu, and L.~Yang (2024).
\newblock Cross-bond momentum spillovers.
\newblock Working Paper.

\end{thebibliography}
\end{document}